\documentclass[11pt,notitlepage]{article}

\usepackage[affil-it]{authblk}
\usepackage{amsmath,amssymb,amsfonts,amscd,mathtools,amsthm}
\usepackage{cite}
\usepackage[colorlinks,allcolors=blue]{hyperref}
\usepackage[none]{hyphenat}
\usepackage[utf8]{inputenc}
\usepackage{array}
\usepackage{multirow,bigdelim}

\usepackage{graphicx}
\graphicspath{ {./images/} }
\usepackage{subcaption}

\usepackage{tikz-cd}   
\usepackage{tikz}

\numberwithin{equation}{section} 

\def \Q{\mathbb{Q}}
\def \Z{\mathbb{Z}}
\def \R{\mathbb{R}}
\def \C{\mathbb{C}}
\def \N{\mathbb{N}}

\renewcommand{\hom}{\text{Hom}}

\DeclareMathOperator{\Tor}{Tor}

\newtheorem{proposition}{Proposition}

\newtheorem{definition}{Definition}
\newtheorem{lemma}{Lemma}

\def \coker{ \, \text{coker }}
\def \lk{\ell k}

\DeclarePairedDelimiter\avg{\langle}{\rangle}

\def \tilde{\widetilde}
\def \hat{\widehat}

\usepackage[left=3cm,top=4cm,right=3cm,bottom=3cm]{geometry}

\title{Summing over homology groups of 3-manifolds}
\author[1,2]{Thomas Nicosanti}
\author[3]{Pavel Putrov}
\affil[1]{SISSA, Via Bonomea 265, Trieste 34136, Italy}
\affil[2]{INFN, Sezione di Trieste, Via Valerio 2, Trieste 34127, Italy}
\affil[3]{ICTP, Strada Costiera 11, Trieste 34151, Italy}

\date{}

\begin{document}

\maketitle
\begin{abstract}
We consider a toy model of a 3-dimensional topological quantum gravity. In this model, a contribution of a given 3-manifold is given by the partition function of an abelian Topological Quantum Field Theory (TQFT), with a topological boundary condition at the boundary. Using the fact that the TQFT partition function depends only on the first homology group of the 3-manifold with some additional structure, the sum over all 3-manifolds with fixed boundary can be rewritten as a sum over finitely generated abelian groups (and the extra structure). We present bounds on the universal weights in the sum, that is, the measure on the set of isomorphism classes of finitely generated abelian groups (with the extra structure) sufficient for the sum to be convergent. Moreover, with further assumptions on the measure, we argue the existence of a distribution of 2d TQFTs, such that the sum is equal to the ensemble average of their partition functions evaluated on the boundary. For a certain kind of measure, the sum over the abelian groups can be factorized into sums over abelian $p$-groups, and the analysis can be performed independently for each prime $p$. 

\end{abstract}

\tableofcontents

\section{Introduction and summary}

One can consider a hypothetical path integral of the theory of quantum gravity that mimics the path integral of a quantum field theory (QFT).  For a fixed spacetime manifold, it consists of integration over all metrics. Then one can consider summation over all possible diffeomorphism classes of the spacetime manifolds, often referred to as the ``sum over topologies''. It has a physical interpretation in terms of the creation and annihilation of ``baby universes'' \cite{Coleman:1988cy,Banks:1988je,Giddings:1988wv,Fischler:1989ka}.
The sum over topologies, in a sense, is analogous to the sum over isomorphism classes of principal bundles in gauge theories. Both the integration and the summation part are in general ill-defined, especially in dimensions $d>2$. For the integration part, apart from the problem --- already present in QFTs --- that the integration measure is mathematically not well defined, it is already not known what exactly the space of metrics that are integrated over should be. In particular, it is not clear whether one should allow singular metrics of a certain kind. For the summation part, even though the set of diffeomorphism classes of compact manifolds is countable, their explicit classification is a hard mathematical problem. Moreover, the sum is, in general, divergent and therefore requires some regularization.

The summation over different topologies can be considered separately from the integration over metrics. Namely, in $d$ dimensions, for a fixed topology $M$, one can consider the state $|M,\phi\rangle_\mathcal{T}\in \mathcal{H}_\mathcal{T}(\Sigma)$ given by a $d$-dimensional Topological Quantum Field Theory (TQFT) $\mathcal{T}$ in its Hilbert space associated to a $(d-1)$-manifold $\Sigma$. Here we explicitly indicate the choice of the diffeomorphism $\phi:\Sigma\xrightarrow{\sim} \partial M \subset M$ as a part of the data on which the state depends. That is, the pair $(M,\phi)$ defines a bordism from $\emptyset$ to $\Sigma$, which is sometimes also referred to as a bordered manifold with boundary $\Sigma$. If one understands $\mathcal{T}$ as a functor $\mathbf{Bord}_{d,d-1}\rightarrow \mathbf{Vect}_\C$ from the bordism category to the category of complex vector spaces, we have $|M,\phi\rangle_\mathcal{T}\equiv\mathcal{T}(M,\phi)\in \hom(\C,\mathcal{T}(\Sigma))\cong \mathcal{T}(\Sigma)\equiv \mathcal{H}_\mathcal{T}(\Sigma)$. One can then define the $d$-dimensional topological quantum gravity state corresponding to the transition from nothing to $\Sigma$ as a sum:
\begin{equation}
    |\Sigma;\mu\rangle_\mathcal{T}\coloneqq\sum_{[M,\phi:\,\Sigma\xrightarrow{\sim} \partial M]} \mu([M,\phi])\,|M,\phi\rangle_\mathcal{T}\qquad \in \mathcal{H}_\mathcal{T}(\Sigma),
    \label{top-gravity-sum-over-topologies}
\end{equation}
 with certain weights $\mu([M,\phi])\in \C$. Assuming $\mu([M,\phi])\geq 0$, the collection of all weights can be understood as a measure $\mu$ on the set of equivalence classes $[M,\phi]$.   The sum is performed over the diffeomorphism classes of bordisms, namely $(M,\phi)\sim (M',\phi')$ if and only if there is a diffeomorphism $\Phi:M\rightarrow M'$ such that $\Phi\circ\phi=\phi'$. In writing such a sum we use the fact that, by definition, the value of the TQFT is the same for equivalent bordisms. We are assuming that the weights are universal, in the sense that they do not depend on the choice of the TQFT $\mathcal{T}$. At the moment we would like to be very generic and we make no additional assumptions; in particular, so far there is no assumption about the independence of the weights on the choice of the diffeomorphism $\phi:\Sigma\xrightarrow{\sim} \partial M$ for a fixed diffeomorphism class $[M]$ of $M$. The implementation of such an assumption will be considered later in the paper. Note that just naively requiring $\mu([M,\phi])$ to be independent of $\phi$  may render the sum divergent, as in general there can be infinitely many pairs $[M,\phi]$ for a given $[M]$, so some regularization is required.

The summation over topologies is in tension with the holographic principle. In the most basic form it tells that when a local boundary condition $\mathfrak{B}$ in the theory $\mathcal{T}$ is imposed on $\Sigma$, the partition function of a $d$-dimensional quantum gravity should be equal to a partition function of a $(d-1)$-dimensional theory $\mathcal{B}$ supported on the boundary: $\langle \mathfrak{B}|\Sigma\rangle_\mathcal{T} =Z_{\mathcal{B}}(\Sigma)$, where $\langle \mathfrak{B}|$ is the corresponding boundary state in  $\mathcal{H}_\mathcal{T}(-\Sigma)\cong \mathcal{H}_\mathcal{T}(\Sigma)^*$. Here and throughout the paper we use the minus sign to denote the orientation change. When the boundary consists of disconnected components, say $\Sigma=\Sigma_1\sqcup\Sigma_2$, one expects the factorization $Z_\mathcal{B}(\Sigma)=Z_\mathcal{B}(\Sigma_1)Z_\mathcal{B}(\Sigma_2)$, while a priori $|\Sigma;\mu\rangle_\mathcal{T} \neq|\Sigma_1;\mu\rangle_\mathcal{T}\otimes |\Sigma_2;\mu\rangle_\mathcal{T}$ due to the presence of ``wormhole'' topologies, that is connected $d$-manifolds $M$ with boundary $\partial M\cong \Sigma=\Sigma_1\sqcup \Sigma_2$. One proposed solution to this issue is that one has to consider an \textit{ensamble} of $(d-1)$-dimensional boundary theories so that the gravity partition function equals to the average of their partition functions with respect to a certain probability distribution of the boundary theories $\mathcal{B}$:
\begin{equation}
    \langle \mathfrak{B}|\Sigma;\mu\rangle_\mathcal{T}=\langle Z_\mathcal{B}(\Sigma)\rangle.
\end{equation}
In two dimensions, such a relation between a sum of 2d TQFT partition functions and an ensemble average of 1d TQFTs was worked out explicitly in \cite{Marolf:2020xie,Balasubramanian:2020jhl,Gardiner:2020vjp,deMelloKoch:2021lqp} (see also \cite{Banerjee:2022pmw} for a general treatment of the sum of partition functions of 2d TQFTs). In three dimensions, the relation between a sum of 3d TQFT partition functions over a certain class of 3-manifolds (with a conformal boundary condition) and an ensemble of 2d Conformal Field Theories (CFT) was considered in \cite{Maloney:2020nni,Cotler:2020ugk,Maxfield:2020ale,Barbar:2023ncl,Aharony:2023zit}. The evidence for such a relation for the sum over all 3-manifolds was recently provided in \cite{Dymarsky:2024frx,Jafferis:2024jkb,Dymarsky:2025agh}.

In this paper, we consider a simplified version of such a holographic relation for 3d topological quantum gravity. For the 3d ``bulk'' theory, we will have an abelian TQFT with a \textit{topological} boundary condition on the 2d boundary. We will also argue that, under certain assumptions, the complete state (\ref{top-gravity-sum-over-topologies}) can be recovered by considering all possible topological boundary conditions. In this setting, one can recast the sum over topologies into a sum over some simple algebraic data. In order to do that, the first step is to realize the manifolds $M$ with boundary $\Sigma$ as the complements of a collection of handlebodies in a closed manifold. Namely, let us fix a collection of handlebodies $V$, the boundaries of which are connected components of $\Sigma$ with orientation reversed: $\partial V=-\Sigma$. Consider its embedding in a closed 3-manifold $Y$, $\iota: V\hookrightarrow Y$. Any pair $(M,\phi)$ in the sum (\ref{top-gravity-sum-over-topologies}) can be realized as $M=Y\setminus \iota(V)$ with $\phi= \iota|_{\partial V}:\Sigma\xrightarrow{\sim} \partial M\subset M$. The diffeomorphism\footnote{In 3 dimensions, diffeomorphism classes can be identified with homeomorphism classes, as any topological 3-manifold admits a unique smooth structure.} class of $(M,\phi)$ is the same for the pairs $(Y,\iota)$ and $(Y',\iota')$ such that there exists a diffeomorphism $\Phi:Y\rightarrow Y'$ satisfying $\iota'=\phi\circ \iota$, so we will consider the corresponding equivalence relation on the pairs. The sum (\ref{top-gravity-sum-over-topologies}) then can be rewritten in the following form:
\begin{equation}
    |\Sigma;\mu\rangle_\mathcal{T}=\sum_{[Y,\,\iota:V\hookrightarrow Y]} \mu'([Y,
    \iota])\,|Y\setminus \iota(V),\iota|_{\partial V}\rangle_\mathcal{T}.
    \label{top-gravity-sum-over-closed-manifolds}
\end{equation}
Note that this temporarily makes the sum larger, since there are different equivalence classes $[Y,\iota]$ that produce the same $[M,\phi]=[Y\setminus \iota(V),\iota|_{\partial V}]$. Therefore, the weights in the new sum are not uniquely determined by the original weights, as they only need to satisfy:
\begin{equation}
    \sum_{[Y,\iota]: [Y\setminus \iota(V),\iota|_{\partial V}]=[M,\phi]}\mu'([Y,\iota]) = \mu([M,\phi]).
\end{equation}

Next, one can use the fact that the partition function of an arbitrary abelian TQFT with a topological boundary condition $\mathfrak{B}$ depends on the pair $(Y,\iota)$ only through a certain algebraic data $\mathcal{A}$, namely the first homology group $H_1(V)$, the $\Q/\Z$-valued linking bilinear form on its torsion subgroup, and the induced homomorphism $\iota_*:H_1(V)\cong \Z^{g}\rightarrow H_1(Y)$ (where $g$ is the sum of genera of the components of $\Sigma$). Therefore, at least formally, one can replace the sum (\ref{top-gravity-sum-over-closed-manifolds}) with the sum over the algebraic data, after having evaluated it on $\mathfrak{B}$. In this way, each term can be understood to contain the sum over the pairs $(Y,\iota)$ that have the same algebraic data (the corresponding equivalence classes are denoted with double brackets):
\begin{equation}
    \langle \mathfrak{B}|\Sigma;\mu\rangle_\mathcal{T}=\sum_{\mathcal{A}} \mu_{\text{alg}}(\mathcal{A}) \langle \mathfrak{B}|Y\setminus \iota(V),\iota|_{\partial V}\rangle_\mathcal{T}|_{[[Y,\iota]]=\mathcal{A}},
    \qquad \mu_{\text{alg}}(\mathcal{A})=\sum_{[Y,\iota]:[[Y,\iota]]=\mathcal{A}}\mu'([Y,\iota]).
    \label{top-gravity-sum-over-alg-data}
\end{equation}
Assuming non-negativity, $\mu_{alg}$ can be understood as a measure on $\{\mathcal{A}\}$, the set of possible algebraic data.

Note that in principle it is possible to rewrite directly the sum (\ref{top-gravity-sum-over-topologies}) over manifolds $M$ with boundary in terms of certain other algebraic data, without passing to a sum over closed manifolds or imposing any particular boundary condition. Such data is the long exact sequence of homology groups associated to the pair $(M,\Sigma)$, together with bilinear pairings $H_1(M,\Sigma)\times H_1(M)\rightarrow \Z$ and $\Tor\,H_1(M,\Sigma)\times \Tor\,H_1(M)\rightarrow \Q/\Z$. The equivalence relation between manifolds having isomorphic data is known as homology-isomorphism, Y-equivalence, or Borromean surgery equivalence \cite{contreras2015borromean}. However, due to various constraints on the pieces of the data, it is harder to organize explicitly the sum over the equivalence classes. We are not taking this route in the present paper.

Because we choose a topological boundary condition (not a generic conformal one), the holographic dual will be an ensemble of 2d TQFTs (and not of conformal theories). The fact that the bulk TQFT is abelian provides a universal (i.e. independent of the actual choice of TQFT) relation between the gravity partition functions for disjoint unions of 2-manifolds and their connected sum. This gives a strong constraint on the possible form of the distribution of the boundary 2d TQFTs. In particular, we show that the distribution can be completely and explicitly determined in terms of the average of the partition function on a connected surface, assuming its dependence on the genus is of a certain form -- a Dirichlet series. We also provide sufficient conditions on the measure for this assumption to hold for an arbitrary bulk abelian 3d TQFT.

The rest of the paper is organized as follows. In Section \ref{sec:abelian-tqfts} we review necessary facts about abelian 3d TQFTs. In Section \ref{sec:sum-over-topologies} we consider the sum over the topologies and provide sufficient conditions on the measure for the sum to be absolutely convergent, and also have dependence on the genus of the boundary in the form of Dirichlet series. In Section \ref{sec:general-distribution} we study the universal constraints on the distribution of the 2d TQFTs imposed by the holographic relation with a non-specific bulk abelian 3d TQFT. We provide a general (under certain assumptions) solution to these constraints and show how a given bulk 3d TQFT provides a particular solution. Finally, in Section \ref{sec:extension} we comment on various generalizations. Appendix contains some technical calculations that were used to obtain the results stated in the main text.

\section{Abelian TQFTs and their partition function}
\label{sec:abelian-tqfts}

Given our goal of summing over all topologies with a fixed boundary of a partition function of an arbitrary 3d abelian TQFT, the first step consists of studying how the partition function depends on the choice of a 3-manifold.

\subsection{The partition function of an arbitrary 3d abelian TQFT}

An arbitrary 3d abelian bosonic TQFT can be specified by a finite abelian group $D$ with a non-degenerate even quadratic form $q \colon D \to \Q / \Z$, see e.g. \cite{Belov:2005ze,Kapustin:2010hk,Kaidi:2021gbs}. The physical theory also depends on the chiral central charge, an integer, with its value modulo 8 fixed by $D$ and $q$. A quadratic form is said to be non-degenerate and even if the associate bilinear form $b_q$ is so. The group $D$ is called the discriminant group and classifies the bulk line operators of the theory, i.e. the anyons, while $q$ gives the braiding, i.e. the additional phase acquired by two anyons when exchanged. The elements of the group $D$ also give a basis in the Hilbert space of the TQFT assigned to a 2-torus. Namely, a state corresponding to an element $c\in D$ is realized by a solid torus with the corresponding line operator inserted at its core (a circle in the interior to which the solid torus deformation retracts).

Classically, such a TQFT can always be expressed as an even level $K$ abelian Chern-Simons with gauge group the torus $U(1)^n$. By even level $K$ we mean an $n\times n$ integral even symmetric bilinear form, that can be understood as the matrix of coefficients of the $U(1)$-gauge fields in the Chern-Simons action. Indeed, one can always find the bilinear form $K$ such that $ b_q =  K^{-1} \mod 1$, $D = \coker{K}$, and the chiral central charge is $\sigma(K)$ -- the signature of $K$ -- modulo $8$.

We introduce the partition function following the approach of Deloup\footnote{Our convention is slightly different since we replace the $\tau$-invariant with its complex conjugate to make contact with the physics literature.} as in~\cite{deloup1999linking} and~\cite{deloup2001abelian}. We first define the $\tau$-invariant, that is a topological invariant of a pair $(Y, \mathcal{L})$ of a closed, oriented, connected 3-manifold $Y$ and an oriented framed link $\mathcal{L}$ with components colored by the elements of $D$. The partition function of an abelian TQFT is equal to it up to a normalization, which we will fix later.

Given a closed, oriented, connected 3-manifold $Y$ and an oriented framed $g$-component link $\mathcal{L} = \mathcal{L}_1 \cup \dots \cup \mathcal{L}_g$ in $Y$, we call link presentation of the pair $(Y, \mathcal{L})$ a pair of links $(\mathcal{J}, \mathcal{J}')$ in $S^3$ where $\mathcal{J} = \mathcal{J}_1 \cup \dots \cup \mathcal{J}_m$ is the surgery link of $Y$, while $\mathcal{J}'$ is a link yielding $\mathcal{L}$ after surgery. We will use $J$ and $J'$ to denote the linking matrices of $\mathcal{J}$ and $\mathcal{J}'$ respectively, and $L$ to denote the linking matrix of the total link $\mathcal{J} \cup \mathcal{J}'$. Then, the $\tau$-invariant for the TQFT $(D,q)$ of the manifold $(Y,\mathcal{L})$ with the components of the link colored by $c \in D^g$ is defined in the following way:
\begin{equation}
    \label{eqn:DeloupPartitionFunction}
    \tau(Y, \mathcal{L}; D,q; c) = \frac{\exp \left[ \frac{\pi i}{4} \sigma(K) \sigma(L)  \right]}{|D|^{m/2} } \sum_{x \in  \Z^m \otimes D} \exp\left[ - 2 \pi i  (L  \otimes q) (x \oplus c) \right],
\end{equation}
where $L \otimes q$ is the unique quadratic form on $\Z^{m+g}\otimes D$ satisfying:
\begin{equation}
    (L \otimes q) (x \otimes y) = x^T L x \cdot q(y).
\end{equation}

We stress that by definition $\tau$ does not distinguish between unknots and non-trivial knots, thus we may assume that $\mathcal{J} \cup \mathcal{J}'$ is a link of unknots.

In the case of a closed 3-manifold without a link, i.e. with $\mathcal{L} = \emptyset$, the formula reads:
\begin{equation}
    \tau(Y;D,q) \equiv  \tau(Y, \emptyset; D,q; 0) = \frac{\exp \left[\frac{\pi i}{4} \sigma(K) \sigma(L)  \right]}{|D|^{m/2} } \sum_{x \in  \Z^m \otimes D} \exp\left[ - 2 \pi i  (L  \otimes q) (x) \right],  
\end{equation}
while the partition function in the physical normalization is:
\begin{equation}
    Z(Y; D,q) = \frac{1}{|D|^{(m+1)/2} } \sum_{x \in  \Z^m \otimes D} \exp\left[ - 2 \pi i  (L  \otimes q) (x) \right].
    \label{partition-function-no-link-surgery-formula}
\end{equation}

The difference in the two normalizations is such that on a 3-sphere $\tau$ evaluates to 1 in any theory, while $Z(S^3)=|D|^{-1/2}$. Moreover, $Z(Y;D,q)$ generally suffers from framing anomaly, which means its phase is ambiguous, unless one fixes the 2-framing of the tangent bundle $Y$, or the signature of a 4-manifold $X$ bounded by $Y=\partial X$ \cite{Witten:1988hf,atiyah1990framings}. The latter can be identified with $\sigma(L)$\footnote{The surgery representation of $Y$ provides a 4-manifold $X$ bounded by $Y$, obtained by attaching 2-handles to the 4-ball $D^4$ along the tubular neighborhoods of the components of $\mathcal{L}\subset S^3=\partial D^4$. The signature of this 4-manifold coincides with $\sigma(L)$.}. Later on, we will restrict ourselves to the case of zero chiral central charge $\sigma(K)=0$, so that this ambiguity will disappear. 

Now, consider the case when $\mathcal{L} \ne \emptyset$. The $\tau$-invariant, up to normalization, then gives the partition function on the compact manifold $M = Y \setminus \bar{U}$, where $U$ is a tubular neighborhood of $\mathcal{L}$. The boundary of $M$ is $\partial M \cong (T^2)^{\sqcup g}$, disjoint union of $g$-copies of a 2-torus $T^2$. This construction allows us to consider the partition function on any compact oriented connected 3-manifold, the boundary of which consists of a disjoint union of any number of tori. We shall enter into the details of this correspondence, as well as the relation to the manifolds with more general boundary components, later in this section. In this interpretation, the elements $c_i\in D$, $i=1,\ldots,g$ serve the role of the basis elements in the Hilbert spaces of the TQFT assigned to each torus. 

We can then retrieve the correct physical normalization in the presence of boundaries from the knowledge of the physical partition function on a closed manifold. It suffices to enforce the following prescription:
\begin{equation}
    \frac{1}{|D|^{g/2}} \sum_{c \in D^g} Z(Y,\mathcal{L};D,q;c) = Z(Y';D,q),
\end{equation}
where $Y'$ is the closed 3-manifold obtained by surgery on the link $\mathcal{J} \cup \mathcal{J}'$. This fixes the normalization to be:
\begin{equation}
    Z(Y,\mathcal{L};D,q;c) = \frac{1}{|D|^{(m+1)/2} } \sum_{x \in  \Z^m \otimes D} \exp\left[- 2 \pi i  (L  \otimes q) (x\oplus c) \right].
    \label{abelian-TQFT-partition-function-surgery-formula}
\end{equation}

Once the TQFT is fixed, a crucial result of Deloup~\cite{deloup2001abelian} says that $\tau(Y, \mathcal{L}; D,q;c)$ is determined by the following data:
\begin{itemize}
    \item the first homology group $H_1(Y)$;
    \item the linking pairing $\lk \colon \Tor H_1(Y) \otimes \Tor H_1(Y) \to \Q / \Z$ on the torsion subgroup;
    \item the framed 1-cycle $\hat\theta = \sum_k c_k \otimes \mathcal{L}_k$ in $Y$ with coefficients in $D$.
\end{itemize}
Note the the dependence on $H_1(Y)$ can be split into the dependence on the pair $(b_1(Y),\Tor H_1(Y))$, where $b_1(Y)$ is the first Betti number $b_1(Y)$, which equals to the rank of $H_1(Y)$.

Denote by $\theta$ its corresponding class in the homology group $H_1(Y;D)\cong D\otimes H_1(Y)$. As was shown in \cite{deloup2001abelian}, the invariant $\tau$ vanishes unless $\theta\in D\otimes \Tor H_1(Y)$ (which is always satisfied for $Y$ a rational homology sphere). Moreover, if this condition holds, then either $\tau$ vanishes or:
\begin{equation}
    \tau(Y,\mathcal{L};D,q;c) = |H^1(Y;D)|^{1/2} \exp \left[ 2\pi i \, \alpha(\Tor H_1(Y), \lk, \hat\theta) \right],
    \label{tau-abs-value}
\end{equation}
with $\alpha$ taking value in $\Q / \Z$. We refer to~\cite{deloup2001abelian} for an explicit expression of $\alpha$ and a characterization of the vanishing of the $\tau$-invariant.

Similar considerations also apply to the partition function $Z$ and can be deduced from the $\tau$-invariant.

\subsubsection{Manifolds with general boundary}

\label{sec:general-boundary}

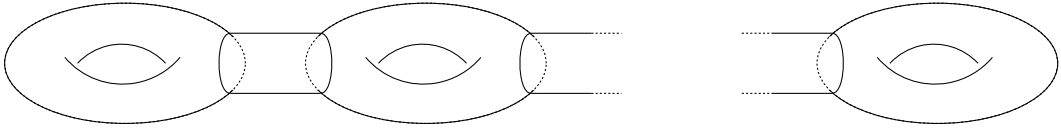
\begin{figure}[t]
    \centering
\begin{tikzpicture}[scale=4]
\draw (0,0) ellipse (0.4 and 0.2);
\draw[white,fill=white] (0.35,0) circle (.087cm);
\draw[dash pattern={on .8pt off .8pt}] (0,0) ellipse (0.4 and 0.2);
\draw (-0.2,0.02) arc (40+180:140+180:0.25);
\draw (.135,0) arc (133+90+180:227+90+180:0.20);
\draw (0.35,0.1) .. controls (0.3,0.1) and (0.3,-0.1) .. (0.35,-0.1);
\draw (0.35,0.1) -- (.65, 0.1);
\draw (0.35,-0.1) -- (.65, -0.1);
\draw (1,0) ellipse (0.4 and 0.2);
\draw[white,fill=white] (.65,0) circle (.087cm);
\draw (0.65,0.1) .. controls (0.7,0.1) and (0.7,-0.1) .. (0.65,-0.1);
\draw[white,fill=white] (1.35,0) circle (.087cm);
\draw[dash pattern={on .8pt off .8pt}] (1,0) ellipse (0.4 and 0.2);
\draw (1-0.2,0.02) arc (40+180:140+180:0.25);
\draw (1.135,0) arc (133+90+180:227+90+180:0.20);
\draw (1.35,0.1) .. controls (1.3,0.1) and (1.3,-0.1) .. (1.35,-0.1);
\draw (1.35,0.1) -- (1.55, 0.1);
\draw[dash pattern={on .8pt off .8pt}] (1.55, 0.1) -- (1.65, 0.1);
\draw (1.35,-0.1) -- (1.55, -0.1);
\draw[dash pattern={on .8pt off .8pt}] (1.55, -0.1) -- (1.65, -0.1);
\draw[dash pattern={on .8pt off .8pt}] (2.05, 0.1) -- (2.15, 0.1);
\draw (2.15,0.1) -- (2.35, 0.1);
\draw[dash pattern={on .8pt off .8pt}] (2.05, -0.1) -- (2.15, -0.1);
\draw (2.15,-0.1) -- (2.35, -0.1);
\draw (2.7,0) ellipse (0.4 and 0.2);
\draw[white,fill=white] (2.35,0) circle (.087cm);
\draw (2.35,0.1) .. controls (2.4,0.1) and (2.4,-0.1) .. (2.35,-0.1);
\draw[dash pattern={on .8pt off .8pt}] (2.7,0) ellipse (0.4 and 0.2);
\draw (2.7-0.2,0.02) arc (40+180:140+180:0.25);
\draw (2.7+.135,0) arc (133+90+180:227+90+180:0.20);
\end{tikzpicture}
    \caption{The genus $g$ handle-body $V_g$.}
    \label{fig:canonicalHandleBody}
\end{figure}

In order to consider manifolds with boundary components of any genus, we will proceed as follows. First, let us fix a particular solid torus $V_1$ (i.e. a specific representative in the homeomorphism class). Then its boundary gives a particular 2-torus $T^2\coloneqq-\partial V_1$, where, as before, the minus refers to the opposite orientation. Next, we construct a genus $g$ handle-body $V_g$ as a gluing of $g$ copies of $V_1$ and $g-1$ copies of the solid cylinder $D^2\times [0,1]$, as depicted in Figure~\ref{fig:canonicalHandleBody}. Namely, we choose an embedding of $D^2$ into the first and the $g$-th copies of $V_1$ and a pair of non-overlapping embeddings of $D^2$ into each of the other copies. Using such embeddings, we attach the first cylinder to the first and the second copies of $V_1$, the second cylinder to the second and the third copies of $V_1$, and so on. The boundary of the resulting handle-body $V_g$ gives a particular genus $g$ surface $\Sigma_g\coloneqq-\partial V_g$. This provides a realization of the TQFT Hilbert space $\mathcal{H}(\Sigma_g)$ in terms of the line operators. Namely, for an abelian TQFT, the basis elements can be given by the insertion of a line colored by $c_i\in D$ at the core of the $i$-th copy of $V_1$ in $V_g$. That is we have fixed the isomorphisms:
\begin{equation}
    \mathcal{H}(\Sigma_g)\cong \mathcal{H}(T^2)^{\otimes g}\cong \left(\,\bigoplus_{c\in D} \mathcal \C|c\rangle\right)^{\otimes g}.
\end{equation}

Now let us consider a compact manifold $M$ with a choice of isomorphism $\phi:\partial M\xrightarrow{\sim} \Sigma = \Sigma_{g_1}\sqcup \Sigma_{g_2}\sqcup \ldots$ As was already considered in the introduction, any such pair can be realized as a complement $Y\setminus \iota(V)$, in a closed manifold $Y$, with an embedding $\iota: V\rightarrow Y$, where $V=V_{g_1}\sqcup V_{g_2}\sqcup\ldots$, so that $\phi=\iota|_{\partial V}$. The embedding $\iota$ can be restricted to the solid tori $V_1$ in the handle-bodies $V_{g_i}$. Such a restriction would produce a different manifold\footnote{$f.t.$ stands for ``filled tubes''.} $M^{f.t.}\coloneqq Y\setminus \iota(V_1^{\sqcup (g_1+g_2+\ldots)})$ with $\phi^{f.t.}\coloneqq\iota|_{\partial V_1^{\sqcup (g_1+g_2+\ldots)}}:(T^2)^{\sqcup(g_1+g_2+\ldots)}\xrightarrow{\sim} \partial M^{f.t.}$. Equivalently, one can obtain $M^{f.t.}$ by attaching $(g_1-1)+(g_2-1)+\ldots$ copies of the solid cylinder $D^2\times [0,1]$ to $M$ along the cylinders $S^1\times [0,1]$ at the boundary, according to the way $\Sigma_{g_i}$ were defined above. The manifolds $M$ and $M^{f.t.}$ are schematically illustrated in Figure~\ref{fig:filledTubes}. 

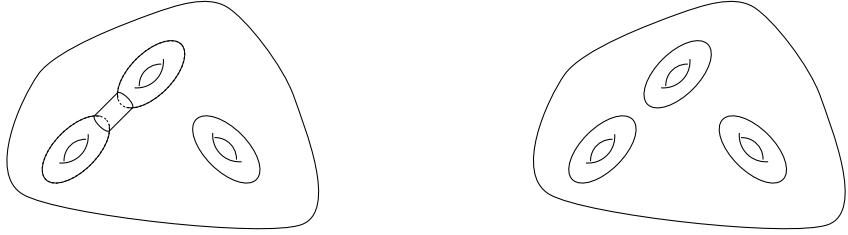
\begin{figure}[t]
    \centering
\begin{tikzpicture}[x={(1cm,1cm)},y={(-1cm,1cm)}]
\draw[x={(1cm,0cm)},y={(0cm,1cm)}] plot [smooth cycle] coordinates {(-.7,-.6)  (-.5,1) (1,1.8) (2,1.9)  (2.9,.7) (3,-1)};
\draw (0,0) ellipse (0.4 and 0.2);
\draw[white,fill=white] (0.35,0) circle (.087cm);
\draw[dash pattern={on .8pt off .8pt}] (0,0) ellipse (0.4 and 0.2);
\draw (-0.2,0.02) arc (40+180:140+180:0.25);
\draw (.135,0) arc (133+90+180:227+90+180:0.20);
\draw (0.35,0.1) .. controls (0.3,0.1) and (0.3,-0.1) .. (0.35,-0.1);
\draw (0.35,0.1) -- (.65, 0.1);
\draw (0.35,-0.1) -- (.65, -0.1);
\draw (1,0) ellipse (0.4 and 0.2);
\draw[white,fill=white] (.65,0) circle (.087cm);
\draw (0.65,0.1) .. controls (0.7,0.1) and (0.7,-0.1) .. (0.65,-0.1);
\draw[dash pattern={on .8pt off .8pt}] (1,0) ellipse (0.4 and 0.2);
\draw (1-0.2,0.02) arc (40+180:140+180:0.25);
\draw (1.135,0) arc (133+90+180:227+90+180:0.20);
\draw[x={(1cm,-1cm)},y={(1cm,1cm)}] (1,1) ellipse (0.4 and 0.2);
\draw[x={(1cm,-1cm)},y={(1cm,1cm)}] (1-0.2,1.02) arc (40+180:140+180:0.25);
\draw[x={(1cm,-1cm)},y={(1cm,1cm)}] (1.135,1) arc (133+90+180:227+90+180:0.20);
\begin{scope}[shift={({3.5},{-3.5})}]
\draw[x={(1cm,0cm)},y={(0cm,1cm)}] plot [smooth cycle] coordinates {(-.7,-.6)  (-.5,1) (1,1.8) (2,1.9)  (2.9,.7) (3,-1)};
\draw (0,0) ellipse (0.4 and 0.2);
\draw (-0.2,0.02) arc (40+180:140+180:0.25);
\draw (.135,0) arc (133+90+180:227+90+180:0.20);
\draw (1,0) ellipse (0.4 and 0.2);
\draw (1-0.2,0.02) arc (40+180:140+180:0.25);
\draw (1.135,0) arc (133+90+180:227+90+180:0.20);
\draw[x={(1cm,-1cm)},y={(1cm,1cm)}] (1,1) ellipse (0.4 and 0.2);
\draw[x={(1cm,-1cm)},y={(1cm,1cm)}] (1-0.2,1.02) arc (40+180:140+180:0.25);
\draw[x={(1cm,-1cm)},y={(1cm,1cm)}] (1.135,1) arc (133+90+180:227+90+180:0.20);
\end{scope}
\end{tikzpicture}
    \caption{The complement $Y \setminus \iota(V) = M$ on the left and the filled-tubes manifold $M^{f.t}$ on the right.}
    \label{fig:filledTubes}
\end{figure}

Let $|M,\phi\rangle \in \bigotimes_{i}\mathcal{H}(\Sigma_{g_i})$ be the state provided by the TQFT on the manifold $M$ and $\langle c^{(i)}| \in \mathcal{H}(-\Sigma_{g_i})\cong \mathcal{H}(\Sigma_{g_i})^*$, $c^{(i)}\in D^{g_i}$, the state provided by $V_{g_i}$ with the lines $c^{(i)}_j,\;j=1,\ldots,g_i$ inserted at the cores of the $g_i$ solid tori $V_1$ inside $V_{g_i}$. We then have:
\begin{equation}
    \left(\langle c^{(1)}|\otimes \langle c^{(2)}|\otimes \ldots \right) |M,\phi\rangle = \left(\bigotimes_{i,j}\langle c^{(i)}_j|\right)|M^{f.t.},\phi^{f.t.}\rangle = Z\left(Y,\mathcal{L};D,q;\bigoplus_{i}c^{(i)}\right),
\end{equation}
where $\mathcal{L}$ is the framed $g=g_1+g_2+\ldots$ -component link in $Y$ obtained by the restriction of the embedding $\iota:V\hookrightarrow Y$ to the cores of solid tori $V_1$ inside $V=\sqcup_{i} V_{g_i}$ and their normal bundles.

\subsubsection{Topological boundary conditions}
\label{sec:top-bc}

We are interested in computing the sum over manifolds with a fixed boundary (\ref{top-gravity-sum-over-topologies}), recasted into the sum over closed manifolds (\ref{top-gravity-sum-over-closed-manifolds}), of partition functions of an abelian TQFT with a topological boundary condition. The topological boundary conditions are known to be in one-to-one correspondence with Lagrangian subgroups $\mathfrak{L}\subset D$ with respect to the quadratic form $q$ \cite{Kapustin:2010hk}. Imposing the boundary condition corresponding to a Lagrangian subgroup $\mathfrak{L}_i$ on the $i$-th boundary component $\Sigma_{g_i}$ of the bordism $(M=Y\setminus \iota(V),\phi=\iota|_{\partial V})$ results in pairing the state $|M,\phi\rangle $ with the boundary state:
\begin{equation}
\langle \mathfrak{L}_i| \coloneqq\sum_{c\in \mathfrak{L}^{g_i}_i\subset D^{g_i}}\langle c| \qquad \in \mathcal{H}(-\Sigma_{g_i}).
\end{equation}
Such a state is in particular invariant under the action of the mapping class group of $\Sigma_{g_i}$ on $\mathcal{H}(-\Sigma_{g_i})$. The pairing results in a quantity independent of the restriction of $\phi:\partial M\xrightarrow{\sim}\Sigma\equiv \sqcup_{i}\Sigma_{g_i}$ on the corresponding boundary component.

When the same boundary condition is imposed for all the components, i.e. $\mathfrak{L}_i \equiv\mathfrak{L}$ $\forall i$, writing the boundary state $\langle \mathfrak{B}|=\bigotimes_i \langle \mathfrak{L}|^{(i)} \in \bigotimes_i \mathcal{H}(-\Sigma_{g_i})$, we have for connected $M$:
\begin{equation}
    Z(M;\mathfrak{B})\equiv \langle \mathfrak{B}|M,\phi\rangle =\left(\bigotimes_i \langle \mathfrak{L}|^{(i)}\right)|M,\phi\rangle= 
    \sum_{c\in \mathfrak{L}^g} Z(Y,\mathcal{L};D,q;c),
    \label{Lagrangian-b-c}
\end{equation}
where $\mathcal{L}$ is the $g$-component link in $Y$ ($g=\sum_i g_i$) defined by the embedding $\iota:V\rightarrow Y$ as above.

We recall that a Lagrangian subgroup of a pair $(D,q)$ is an isotropic and coisotropic subgroup of $D$ with respect to the bilinear form $b_q$ associated with $q$. In particular, this implies that $|\mathfrak{L}|^2 = |D|$ and $\sigma(K) = 0 \mod 8$, as was argued in \cite{Kapustin:2010hk}. Therefore, a TQFT admitting topological boundary conditions has no framing anomaly. 
Finally, it was shown in~\cite{Kaidi:2021gbs} that an abelian bosonic TQFT admits topological boundary conditions if and only if it is an abelian Dijkgraaf-Witten theory.

We now observe an additional characterization of an abelian TQFT admitting topological boundary conditions due to~\cite{Kaidi:2021gbs}. First of all, notice that any abelian TQFT $\mathcal{T} = (D,q)$ can be factorized into prime factors $\mathcal{T}_{p_i} = (D_{p_i}, q_{p_i})$, where $D_{p_i}$ are the Sylow $p$-subgroups of $D$, and since such splitting is orthogonal with respect to $q$, the quadratic form decomposes accordingly. Then, an abelian TQFT $\mathcal{T}$ admits topological boundary conditions if and only if each factor $\mathcal{T}_{p_i}$ does.

The hypothesis of the existence of topological boundary conditions simplifies the analysis at the beginning of this section. First of all, the relation between the TQFT partition function and the Deloup $\tau$-invariant does not involve the ambiguous phase factor in this case, i.e. one simply has $Z=|D|^{-1/2}\tau$. Moreover, when all components of $c\in D^g$ are in the same Lagrangian subgroup, as it happens in the sum in the right-hand side of (\ref{Lagrangian-b-c}), the invariant depends on the homology class $\theta=\sum_{i=1}^g{c_i}\otimes [\mathcal{L}_i]\in H_1(Y,\mathfrak{L})\subset H_1(Y,D)$, but not on the choice of a framed 1-cycle $\hat\theta$ representing it. This is a consequence of the invariance under the action of the mapping class group of the boundary of $M=Y\setminus \iota(V)$, but it can also be seen explicitly as follows. Recall that $L$ is the linking matrix of the link $\mathcal{J} \cup \mathcal{J}'$ in $S^3$ and is related to the linking matrices of $\mathcal{J}$ and $\mathcal{J}'$ through:
\begin{equation}
L =
    \begin{pmatrix}
        J & F \\
        F^T & J'
    \end{pmatrix},
\end{equation} 
where $F \colon \Z^g \to \Z^m$ gives the linking numbers between the link components of $\mathcal{J}$ and $\mathcal{J}'$. The exponent in the surgery formula (\ref{eqn:DeloupPartitionFunction}) is expressed in terms of those blocks of $L$ as follows:
\begin{equation}
    (L\otimes q)(x\oplus c)=(J\otimes q)(x)+(F\otimes b_q)(x,c)+(J'\otimes q)(c).
    \label{surgery-exponent-block-decomposition}
\end{equation}
The homology group $H_1(M)\cong \coker J$ and the linking form $\lk$ on its torsion part are determined by $J$. Notice that $J$ has rank equal to $m - b_1$. Then, there exists a non-singular integer matrix $\tilde{J}$ and a unimodular matrix $P$ such that $P^T J P = \tilde{J} \oplus 0_{b_1}$, so that $\Tor H_1(Y)\cong \coker \tilde{J}$ and $\lk = \tilde{J}^{-1} \mod \Z$. The dependence on the framed 1-cycle is through the blocks $F$ and $J'$. Namely, the homology class $\theta$ is fixed by $F$ while the choice of the framed representative $\hat\theta$ is determined by $J'$. When all $c_i$ are in the same Lagrangian subgroup, the last term in (\ref{surgery-exponent-block-decomposition}) identically vanishes and the dependence on $J'$, and therefore on framing, disappears. 
The homology class of $\theta$ is determined by the collection $f_i\coloneqq[\mathcal{L}_i]\in H_1(Y),\; i=1,\ldots,g$, which can be combined into an element $f\in H_1(Y)^g$ or, equivalently, into the map $\Z^g \rightarrow H_1(Y)$. The latter is the pushforward map $f=\iota_*:H_1(V)\cong H_1(\mathcal{L})\rightarrow Y$, where the generators of $H_1(V)\cong\Z^g$ are given by the cores inside the copies of the solid torus $V_1$ in $V$. In terms of the surgery data, we have $f = p \circ F$, where $p$ is the projection map to $\coker J$.

Therefore, when the same topological boundary condition is imposed on all boundary components, for any choice of the abelian 3d TQFT, the partition function depends on $(Y,\iota:V\hookrightarrow Y)$ only through the triple $(H_1(Y),\lk,f)$. That is, one can write:
\begin{equation}
    Z(Y,\mathcal{L};D,q;c)\equiv Z( H_1(Y),\lk,f;D,q;c).
    \label{partition-function-alg-dependence}
\end{equation}
On the other hand, any triple $(A,\ell,f)$, where $A$ is a finitely generated abelian group, $\ell$ is a non-degenerate bilinear form $\Tor A\times \Tor A \rightarrow \Q/\Z$, and $f:\Z^g\rightarrow A$, can be realized by some 3-manifold $Y$ and an embedding of a collection of handle-bodies $\iota:V\hookrightarrow Y$ with total genus $g$. Such triples give the algebraic data anticipated in the introduction, that is $\mathcal{A}=(H_1(Y),\lk,f)$ in the summation formula (\ref{top-gravity-sum-over-alg-data}).

Moreover, the TQFT partition function depends on $f\in \hom(\Z^g,A)\cong A^g$ only through $\theta\equiv \sum_{i=1}^gc_i\otimes f_i\in \mathfrak{L}\otimes A\subset D\otimes A$, so one can write:
\begin{equation}
    Z(A,\ell,f;D,q;c)\equiv Z \left( A,\ell;D,q;\textstyle\sum_{i=1}^g c_i\otimes f_i \right).
    \label{tqft-partition-function-theta-dependence}
\end{equation}

\subsubsection{Factorization properties}
\label{sec:factorization}
The TQFT partition function (equivalently, the $\tau$-invariant) behaves nicely under the connected sum operation. Indeed, given a link $\mathcal{L}_1$ in $Y_1$ and a link $\mathcal{L}_2$ in $Y_2$, it satisfies:
\begin{equation}
    Z(Y_1 \# Y_2, \mathcal{L}_1 \cup \mathcal{L}_2; D, q; c_1\oplus c_2) = \frac{Z(Y_1, \mathcal{L}_1; D,q; c_1) \cdot Z(Y_2, \mathcal{L}_2; D,q;c_2)}{Z(S^3,\emptyset;D,q;0)}.
        \label{connected-sum-factorization}
\end{equation}
When all the components of $c_1$ and $c_2$ are in the same Lagrangian subgroup, one can write this relation in terms of the algebraic data used in (\ref{partition-function-alg-dependence}) as follows:
\begin{multline}
    Z(A_1\oplus A_2,\ell_1\oplus \ell_2,f_1\oplus f_2; D, q; c_1\oplus c_2) = \\\frac{Z( A_1,\ell_1,f_1; D,q; c_1) \cdot Z(A_2,\ell_2,f_2; D,q; c_2)}{Z(0,0,0;D,q;0)}.
    \label{connected-sum-factorization-algebraic}
\end{multline}

Moreover, the partition function is multiplicative with respect to the direct sum of TQFTs. By this we mean that, given $(D, q; c) = (D_1 \oplus D_2, q_1 \oplus q_2; c_1 \oplus c_2)$, we have:
\begin{equation}
    Z(Y, \mathcal{L}; D, q; c) = Z(Y, \mathcal{L}; D_1,q_1; c_1) \cdot Z(Y, \mathcal{L}; D_2,q_2;c_2).
    \label{factorization-discriminant-group}
\end{equation}

One can also consider a different operation that preserves the number of boundary components. Consider a pair $(Y,\mathcal{L})$ realized via surgery on a pair of framed links in $S^3$ of the form $(\mathcal{J}_1\cup \mathcal{J}_2, \mathcal{J}')$, where $J_1$ and $J_2$ are unlinked with each other.
Let the pairs $(Y_1,\mathcal{L}_1)$ and $(Y_2,\mathcal{L}_2)$ be realized by the pairs $(\mathcal{J}_1,\mathcal{J}')$ and $(\mathcal{J}_2,\mathcal{J}')$ respectively, and notice that for closed manifolds without links, this implies $Y=Y_1\# Y_2$. In an abelian TQFT, without loss of generality one can assume that $\mathcal{L}=\mathcal{L}_1\#\mathcal{L}_2$, i.e. that each component of $\mathcal{L}$ is a connected sum of the corresponding components of $\mathcal{L}_1$ and $\mathcal{L}_2$. From the surgery formula, we then have:
\begin{equation}
    Z(Y_1\# Y_2,\mathcal{L}_1\# \mathcal{L}_2;D,q;c) = \frac{Z(Y_1,\mathcal{L}_1;D,q;c) \cdot Z(Y_2,\mathcal{L}_2;D,q;c)}{Z(S^3, \mathcal{J}';D,q;c)} .
\end{equation}
When all the components of $c$ are in the same Lagrangian subgroup, the relation simplifies as the dependence on the linking matrix of $\mathcal{J}'$ disappears:
\begin{equation}
    Z(Y_1\# Y_2,\mathcal{L}_1\# \mathcal{L}_2;D,q;c) = \frac{Z(Y_1,\mathcal{L}_1;D,q;c) \cdot Z(Y_2,\mathcal{L}_2;D,q;c)}{Z(S^3, \emptyset;D,q;0)} .
    \label{eqn:surgeryLinkFactorizationFormula}
\end{equation}
We will refer to this relation as the surgery link factorization formula. We stress that this construction is important since it does not affect the boundary and will allow us to reduce the computation of the partition function on an arbitrary manifold to the evaluation of simpler building blocks.

 In terms of the algebraic data, the surgery link factorization formula reads:
\begin{equation}
    Z(A_1\oplus A_2,\ell_1\oplus \ell_2,f_1\oplus f_2; D, q; c) = \\\frac{Z(A_1,\ell_1,f_1; D,q; c) \cdot Z( A_2,\ell_2,f_2; D,q; c)}{Z(0,0,0;D,q;0)}.
    \label{factorization-boundary-preserving-formula-alg}
\end{equation}
Note that, unlike in (\ref{connected-sum-factorization-algebraic}), here the direct sum $f_1\oplus f_2$ is performed only with respect to the codomain, but not the domain, which is now fixed to $\Z^g$. Because of this, the summation over the Lagrangian subgroup in (\ref{Lagrangian-b-c}) spoils the factorization property (\ref{eqn:surgeryLinkFactorizationFormula}), but not the one in (\ref{connected-sum-factorization}).

However, in terms of the dependence on $\theta\in A\otimes \mathfrak{L}$ as in (\ref{tqft-partition-function-theta-dependence}), both (\ref{connected-sum-factorization-algebraic}) and (\ref{factorization-boundary-preserving-formula-alg}) can be written simply as:
\begin{equation}
    Z(A_1\oplus A_2,\ell_1\oplus \ell_2; D, q; \theta_1+\theta_2) = \\\frac{Z(A_1,\ell_1; D,q; \theta_1) \cdot Z( A_2,\ell_2; D,q; \theta_2)}{Z(0,0,0;D,q;0)}.
    \label{factorization-theta}
\end{equation}

\subsubsection{Classification of finite bilinear forms}
\label{sec:finiteForms}

We briefly review the classification of bilinear forms on finite abelian groups following~\cite{Geiko:2022qjy}. For the original works, we refer to~\cite{wall1963quadratic} and~\cite{kawauchi1980algebraic}.

The classification is essential for our purposes since, thanks to the surgery link factorization formula, the partition function factorizes compatibly with the decomposition of finite abelian groups and bilinear forms on them.

We study the classification of pairs $(A, \ell)$ of a finite abelian group $A$ and a symmetric bilinear form $\ell \colon A \otimes A \to \Q / \Z$ up to equivalence. Two pairs are equivalent, that is $(A,\ell)\sim (A',\ell')$, if there exists an isomorphism $\phi:A\rightarrow A'$ such that $\ell=\ell'\circ(\phi\otimes\phi)$. The first step consists of writing $A$ as a product of its Sylow $p$-subgroups and noticing that the bilinear form splits accordingly:
\begin{equation}
(A,\ell) \simeq \bigoplus_{p \text{ prime}} (A_p,\ell_{p}).
\end{equation}
Next, for each prime, one considers the decomposition into irreducible components of the following form:
\begin{equation}
(A_p,\ell_p) \simeq \bigoplus_{i\in I_p} ((\Z/p^{m_i}\Z)^{d_i},\ell_{p,m_i,d_i}^{(i)})
\label{bilinear-form-decomposition}
\end{equation}
where $\ell_{p,m_i,d_i}^{(i)}$ is a bilinear form on $(\Z/p^{m_i}\Z)^{d_i}$, for some lists of positive integers $\{ m_i,d_i \}_{i\in I_p}$. 

Then, we are left to consider the classification of irreducible forms $\ell_{p,m,d}$ on $(\Z/p^m\Z)^d$. The discussion is considerably different depending on whether $p$ is an odd prime. If that is the case, there are only two possible choices of irreducible bilinear forms on $\Z/{p^m}\Z$:
\begin{align}
    X_{p^m} &\text{ on } \Z/{p^m}\Z\colon \qquad \ell_{p,m,1}(x,y) = \frac{xy}{p^m} \mod 1, \\
    Y_{p^m} &\text{ on } \Z/{p^m}\Z\colon \qquad \ell_{p,m,1}(x,y) = r\frac{ x y}{p^m} \mod 1,
\end{align}
for $r$ any quadratic non-residue mod $p$. We remark that if $p = 3 \mod 4$, then $r = -1$ is a quadratic non-residue; however, there is no canonical choice for $r$ when $p = 1 \mod 4$. 

These two bilinears generate any other bilinear form on the Sylow $p$-subgroup. However, the decomposition (\ref{bilinear-form-decomposition}) is not unique in general, and therefore we must also consider possible relations among the generators. The only such relation is:
\begin{equation}
    2 X_{p^m} = 2 Y_{p^m}
\end{equation}
on $\Z/{p^m}\Z \oplus \Z/{p^m}\Z$. Thus, $dX_{p^m}$ and $(d-1)X_{p^m}+Y_{p^m}$ are the only two inequivalent bilinear forms on $(\Z/{p^m}\Z)^d$.

On the other hand, if $p=2$ the classification is much more involved. All the possible bilinear forms $\ell_{2}$ on the Sylow $2$-subgroup are generated by:
\begin{align}
    A_{2^m} &\text{ on } \Z/{2^m}\Z \colon \qquad \ell_{2,m,1}(x,y) = \frac{xy}{2^m} \mod 1,  \\
    B_{2^m} &\text{ on } \Z/{2^m}\Z \colon \qquad \ell_{2,m,1}(x,y) = -\frac{xy}{2^m} \mod 1,  \\
    C_{2^m} &\text{ on } \Z/{2^m}\Z \colon \qquad \ell_{2,m,1}(x,y) = 5\frac{xy}{2^m} \mod 1,  \\
    D_{2^m} &\text{ on } \Z/{2^m}\Z \colon \qquad \ell_{2,m,1}(x,y) = -5\frac{xy}{2^m} \mod 1,  \\
    E_{2^m} &\text{ on } (\Z/{2^m}\Z)^2  \colon \qquad \ell_{2,m,2}(x,y) = x^T 
    \begin{pmatrix}
        0 & 1/2^m\\
        1/2^m &0
    \end{pmatrix} y
    \mod 1,  \\
    F_{2^m} &\text{ on } (\Z/{2^m}\Z)^2 \colon \qquad \ell_{2,m,2}(x,y) = x^T     
    \begin{pmatrix}
        1/2^{m-1} & 1/2^m\\
        1/2^m &1/2^{m-1}
    \end{pmatrix} y
    \mod 1.
\end{align}
Unfortunately, the relations among these generators are rather involved. We avoid listing here all the possibilities.

Thus, thanks to the surgery link factorization formulas (\ref{eqn:surgeryLinkFactorizationFormula})-(\ref{factorization-boundary-preserving-formula-alg}), it is enough to compute the partition function on the generators of finite bilinear forms.

\subsubsection{Counting finite abelian groups with bilinear forms}
\label{sec:counting-groups-with-forms}

As a prototypical algebraic version of the sum over all 3-manifolds, one can consider the following generating Dirichlet series for the numbers of isomorphism classes of abelian groups of given order (see e.g. \cite{ivic2012riemann}):
\begin{multline}
    \sum_{n\geq 1}\#\left\{\substack{\text{isomorphism classes}\\ \text{of finite abelian groups}\\\text{of order $n$}}\right\}\,n^{-s}=\prod_{p\in\{\text{primes}\}}\sum_{\text{partitions }\Lambda}p^{-s|\Lambda|}=\\
    \prod_{p\in\{\text{primes}\}}\prod_{m\geq 1}\frac{1}{1-p^{-sm}}=\prod_{m\geq 1}\zeta(ms),
    \label{dir-gen-series-ab-groups}
\end{multline}
which is convergent for $\mathrm{Re}(s)>1$. 

For our purpose, however, it is more relevant to consider its analogue counting non-equivalent pairs $(A,\ell)$, where $A$ is a finite abelian group and $\ell$ is a bilinear form on it, as considered in the previous section. Let $a(n)$ denote the number of such pairs with $|A|=n$. Then, according to the classification reviewed above, we have:
\begin{equation}
    \Xi(s)\coloneqq\sum_{n\geq1}a(n)\,n^{-s}=
    \prod_{p\in\{\text{primes}\}}\Xi_p(s),
    \label{dir-gen-series-bil-forms}
\end{equation}
where:
\begin{equation}
    \Xi_p(s)\coloneqq\sum_{n\geq1}a(p^k)\,p^{-ks}.
\end{equation}
For odd $p$, using the fact that for each $(\Z/p^m\Z)^d,\;d\geq 1$ factor there are two inequivalent forms, we have:
\begin{equation}
    \Xi_p(s)=\prod_{m\geq 1}\frac{1+p^{-ms}}{1-p^{-ms}},\qquad p\neq2.
    \label{counting-pnot3-groups}
\end{equation}
For $p=2$ we can use the results about the normal forms by Miranda \cite{miranda1984nondegenerate} to write the generating series (each normal form contributes with $2^{-xs}/(1-{2}^{-ys})$ for some $x$ and $y$):
\begin{multline}
\Xi_2(s)=\left(\frac{1}{1-2^{-s}}+\frac{2^{-2s}}{1-2^{-2s}}\right)\cdot
\left(\frac{1+2^{-2s}+2^{-4s}+2^{-6s}}{1-2^{-2s}}+\frac{2\cdot 2^{-4s}}{1-2^{-4s}}\right)\times
\\
\prod_{m\geq 3}\left(1+2^{-ms}+2^{-2ms}+2^{-3ms}+\frac{2\cdot 2^{-2ms}}{1-2^{-2ms}}+\frac{3\cdot 2^{-ms}+2\cdot 2^{-2ms}+2\cdot 2^{-3ms}+ 2^{-4ms}}{1-2^{-ms}}\right)=\\
\frac{1+2^{-s}+2^{-2s}}{1-2^{-2s}}\cdot
\frac{1+2\cdot 2^{-2s}+4\cdot 2^{-4s}+2\cdot 2^{-6s}+2^{-8s}}{1-2^{-4s}}\times
\\
\prod_{m\geq 3}\frac{1+4\cdot 2^{-ms}+7\cdot 2^{-2ms}+4\cdot 2^{-3ms}+ 2^{-4ms}}{1-2^{-2ms}}.
\end{multline}
Note that the details of the expression for $\Xi(s)$ are not going to be relevant for further analysis. However, what will be important is that the Dirichlet series (\ref{dir-gen-series-bil-forms}) converges for $\mathrm{Re}(s)>1$, similarly to (\ref{dir-gen-series-ab-groups}), as standard arguments show.

\subsubsection{Reciprocity formulae for Gauss sums}

Expressions like the one in the partition function~\eqref{eqn:DeloupPartitionFunction} are called quadratic Gauss sums, or generalizations thereof. They are quite challenging to evaluate directly. Luckily, we have powerful tools at our disposal, called reciprocity formulae for Gauss sums. Not only are they crucial in computing our partition functions but also they give rise to highly non-trivial dualities among abelian TQFTs.

The first reciprocity formula we state is due to Deloup~\cite{deloup1999linking}.
Let $d$ be a non-zero integer and $A$ an even $m\times m$ matrix of integers of rank $r$. We can always find a non-singular matrix $\tilde{A}$ such that there exists an unimodular matrix $P$ satisfying $P^T A P = \tilde{A} \oplus 0_{m-r}$. Then, the reciprocity formula is:
\begin{equation}
    \label{eqn:reciprocityClosedManifold}
    d^{-m/2} \sum_{x \in (\Z/d \Z)^m} \exp \left[ \frac{\pi i}{d} \, \hat{x}^T A x \right] = \frac{d^{(m-r)/2} e^{\pi i /4 \, \sigma(A)}}{|\det \tilde{A}|^{1/2}} \sum_{y \in \coker{\tilde{A}}} \exp \left[ -\pi i d \, \hat{y}^T \tilde{A}^{-1} \hat{y} \right],
\end{equation}
where $\hat{x}$ and $\hat{y}$ are lifts in $\Z^m$ and $\Z^r$ respectively.
We stress that this formula can be helpful in determining the partition function for closed manifolds, as we shall show in what follows. 

A non-empty boundary introduces a linear term in~\eqref{eqn:DeloupPartitionFunction} in addition to the quadratic one. Thus, we need a more general formula that has been proven in~\cite{deloup2007reciprocity}.
Let $A$ and $d$ be as above, with $P^T A P = \tilde{A} \oplus 0_{m-r}$. Given $\psi = \psi_1 \oplus \psi_2 \in \R^r \oplus \R^{m-r}$ such that $d \psi \in \Z^m$, the reciprocity formula is:
\begin{align}
    \notag
    d^{-m/2}& \sum_{x \in (\Z/d\Z)^m} \exp \left[ \frac{\pi i}{d} ( \hat{x}^T A x + 2d \, \hat{x}^T \psi )\right] =\\ 
    &= \delta\left(\psi_2 \in \Z^{m-r}\right) \frac{d^{(m-r)/2} e^{\pi i /4 \, \sigma(A)}}{|\det \tilde{A}|^{1/2}} \sum_{y \in \coker{\tilde{A}}} \exp \left[ -\pi i d  (\hat{y}+ \psi_1)^T \tilde{A}^{-1} (\hat{y}+\psi_1) \right].
    \label{eqn:reciprocityOpenManifold}
\end{align}

The latter reciprocity formula can be used, in particular, to verify that the surgery formula (\ref{partition-function-no-link-surgery-formula}) for the partition function of an abelian TQFT agrees with other explicit formulas for abelian Chern-Simons theories that have appeared in the literature. Indeed, recall that the partition function of a level $K$ $U(1)^n$ Chern-Simons theory on a closed connected manifold $Y$ can be written in the appropriate normalization\footnote{Such that $Z(S^1 \times S^2) = 1$.} as:
\begin{equation}
    Z_K(Y) = \frac{|\det K|^{(b_1(Y)- 1 )/2}}{ |\Tor H_1(Y)|^{n/2}} e^{ -\frac{\pi i}{4} \sigma(K) } \sum_{x \in \Tor H_1(Y) \otimes \Z^n} \exp\left[  \pi i  (\lk \otimes K) (x) \right].
    \label{abelian-CS-partition-function}
\end{equation}
See for example~\cite{Guo:2018vij}. This expression looks different from the one for $Z(Y;D,q)$ given above. However, it follows from applying the first reciprocity formula that the two expressions coincide: $Z_K(Y) = Z(Y; D, q)$ given the appropriate $D$ and $q$.

Moreover, even though it might not be evident, the reciprocity formulae establish a duality between abelian TQFTs. Their effect is to exchange the role of the quadratic form $q$ with the linking matrix $L$ associated with a closed manifold, up to an overall phase proportional to the anomaly. To be more specific, this duality acts on abelian spin-TQFTs, since it can map bosonic theories to fermionic ones. See, for example,~\cite{Kaidi:2021gbs} for a more detailed discussion of this duality. We remark that because of this correspondence, the sum over closed $3$-manifolds could be interpreted as a sum over abelian TQFTs.

\subsection{Dijkgraaf-Witten TQFT}

So far, we have not specified any particular abelian 3-dimensional TQFT. However, in order to make our discussion concrete and compute the gravitational partition function, we will need to fix a theory.

Our focus will be on Dijkgraaf-Witten (DW) TQFTs, as they are anomaly-free and admit topological boundary conditions, as requested. They were first introduced by Dijkgraaf-Witten in~\cite{Dijkgraaf:1989pz} as topological gauge theories with a finite gauge group $G$. On a closed manifold $Y$, the partition function is:
\begin{equation}
    Z(Y) = \frac{1}{|G|} \sum_{\gamma \in \hom(\pi_1(Y), G) } \langle \gamma^* \alpha , [Y] \rangle,
\end{equation}
where $\alpha \in H^3(BG, U(1) )$ and $\gamma^* \colon H^3(BG, U(1)) \to H^3(Y, U(1))$. Here $BG$ is the classifying space, and we have implicitly used that a map $\gamma \in \hom(\pi_1(Y), G)$ determines a map $\gamma \colon Y \to BG$ (up to homotopy). We say that the DW theory is twisted by $\alpha$ unless $\alpha$ is the trivial cohomology class; in that case, we call the theory untwisted. 

The definition of the partition function on a manifold $M$ with boundary is a bit more involved. Let $\phi: \Sigma  \xrightarrow{\sim} \partial M \subset M$ and recall that the states in $\mathcal{H}(\Sigma)$ are generated by the elements $c \in \hom(\pi_1(\Sigma), G)$. Then, the partition function evaluated on such $c$ should only get contributions from maps $\gamma$ that restrict to $c$ on the boundary, i.e. $\gamma \circ \phi_* = c$, where $\phi_*$ is the map between the fundamental groups induced by $\phi$. However, one also has to deal with the problem that the action of the theory is no longer defined by the formula:
\begin{equation}
    e^{2\pi i S[\gamma]} = \langle \gamma^* \alpha, [Y] \rangle 
\end{equation}
due to the absence of the fundamental class $[Y]\in H_3(Y,\Z)$ when $\partial Y\neq \emptyset$. This can be resolved in various ways, see e.g. \cite{Dijkgraaf:1989pz,Freed:1991bn,wakui1992dijkgraaf}. However, we will not dive into the details of this procedure since it will not be needed for the calculations performed later in this paper. We remark that we will always consider abelian gauge groups, so that the fundamental group $\pi_1(Y)$ can be replaced with its abelianization $H_1(Y)$.

\subsubsection{Untwisted abelian Dijkgraaf-Witten TQFT}
\label{sec:untwisted-dw}

Let us focus on the case of an untwisted Dijkgraaf-Witten theory with an abelian gauge group $G$.  Since the action is now trivial, the partition function on a closed manifold $Y$ becomes:
\begin{equation}
\label{eqn:DWpartitionFunctionClosedMfld}
    Z(Y) = \frac{|\hom(H_1(Y),G)|}{|G|}.
\end{equation}
 For the same reason, the partition function on a manifold with boundary $M$, with $\phi: \Sigma  \xrightarrow{\sim} \partial M \subset M$, evaluated on a boundary state $c \in \hom(H_1(\Sigma),G)$ is given simply by:
\begin{equation}
    \langle c|M,\phi\rangle = \frac{1}{|G|} \left| \left\{ \gamma \in \hom(H_1(M),G) \colon \gamma \circ \phi_* = c \right\} \right|.
\end{equation}
The theory always admits the topological boundary condition corresponding to the \textit{Neumann boundary condition} on the gauge fields, given by the boundary state $\langle \mathfrak{B}|=\sum_{c\in \hom(H_1(\Sigma),G)} \langle c|$. The TQFT partition function with such a boundary condition is:
\begin{equation}
   Z(M;\mathfrak{B})\equiv \langle \mathfrak{B}| M,\phi \rangle =  \frac{|\hom(H_1(M),G)|}{|G|}.
    \label{untwisted-dw-neumann-bc}
\end{equation}
The result is independent of $\phi$, as expected.

As was shown in Section \ref{sec:general-boundary}, without loss of generality, one can consider $\Sigma = (T^2)^{\sqcup g}$. To see the agreement of (\ref{untwisted-dw-neumann-bc}) with the general analysis done in Section \ref{sec:top-bc}, we need to rewrite the formula in terms of the triple $(H_1(Y), \lk, f)$, where, as before, $Y$ is the closed manifold obtained from $M$ by gluing back $g$ solid tori $V_1$ and $f=\iota_*$ is the map induced in the first homology by the inclusion $\iota:V_1^{\sqcup g}\hookrightarrow Y$ of these solid tori in $Y$.

First, we observe that $\hom(H_1(M),G) \cong H_2(M, \partial M; G)$ by the universal coefficient theorem and Poincaré duality.
Then, consider the long exact sequence for the relative homology of the pair $(Y, {U})$ with $G$-coefficients, $U=\iota(V_1^{\sqcup g})$ being the image of the solid tori in $Y$:
\begin{equation}
    \dots \to H_2({U};G) \to H_2(Y;G) \to H_2(Y,{U};G) \to H_1({U};G) \to H_1(Y;G) \to \dots
\end{equation}
Applying the universal coefficient theorem, we get:
\begin{equation}
    \dots \to 0 \xrightarrow{\hphantom{f\otimes G}} H_2(Y;G) \xrightarrow{\hphantom{f\otimes G}} H_2(Y,{U};G) \xrightarrow{\hphantom{f\otimes G}} G^g \xrightarrow{f \otimes G} H_1(Y) \otimes G \to \dots 
\end{equation}
From the exactness of the above sequence, we have:
\begin{equation}
    |H_2(Y,U;G)|=|H_2(Y;G)|\cdot |\ker(f\otimes G)|.
\end{equation}
Applying Poincar\'e duality and the excision theorem, it follows that $H_2(Y;G)\cong H^1(Y;G)\cong \hom(H_1(Y);G)$ and $H_2(Y,U;G)\cong H_2(M,\partial M;G)\cong H^1(M;G)\cong \hom(H_1(M),G)$.  Therefore:
\begin{multline}
    Z(M;\mathfrak{B})=|G|^{-1}|\hom(H_1(Y);G)|\cdot|\ker(f\otimes G)|=\\
    =|G|^{b_1(Y)-1}|\Tor H_1(Y)\otimes G|\cdot |\ker(f\otimes G)|.
    \label{part-fn-untwisted-dw-neumann-via-f}
\end{multline}

In terms of the general classification of abelian TQFTs, such a theory has the discriminant group $D=G\oplus G^*$, where $*$ stands for the Pontryagin dual, i.e. $G^*=\hom(G,\Q/\Z)$, and the quadratic form $q:D\rightarrow \Q/\Z$ is given by $q(a,\alpha)=\alpha(a)$. Finally we remark that the considered Neumann boundary condition corresponds to the Lagrangian subgroup $\mathfrak{L}=G\subset D$. In Section \ref{sec:cyclicDW} we will show how the result above can be reproduced via the surgery formula (\ref{abelian-TQFT-partition-function-surgery-formula}) for abelian TQFTs.

\subsubsection{Untwisted Dijkgraaf-Witten TQFT with a cyclic gauge group}
\label{sec:cyclicDW}
An arbitrary finite abelian group is isomorphic to a direct sum of cyclic groups: $G\cong \bigoplus_i \Z/N_i\Z$. Due to the factorization property (\ref{factorization-discriminant-group}), the partition function of a general untwisted abelian DW theory is then a product of partition functions of theories with cyclic gauge groups. 

In the rest of this section, we therefore consider the case of an untwisted DW theory with cyclic gauge group $G = \Z/N\Z$. The discriminant group is then given by $D=G\oplus G^*\cong \Z/N\Z\oplus \Z/N\Z$, with the quadratic function $q(u,v)=uv/N\mod 1,\;u,v\in \Z/N\Z$. Such a theory can be realized as a $U(1)^2$ Chern-Simons theory with the level matrix:
\begin{equation}
    K_{\text{DW}} = 
    \begin{pmatrix}
        0 &  N \\
        N & 0
    \end{pmatrix}.
\end{equation}
Indeed, we have:
\begin{equation}
    D = \coker K_{\text{DW}} = \Z/N\Z \oplus \Z/N\Z,
\end{equation}
and:
\begin{equation}
    q(u,v) =  1/2 \, (u,v)^T K_{\text{DW}}^{-1} (u,v) \mod 1 = uv/N \mod 1.
\end{equation}

We can then compute the partition function on a closed manifold $Y$ using the formula (\ref{abelian-CS-partition-function}) for a general abelian CS theory. First of all, notice that the theory is anomaly-free since the signature of $K_{\text{DW}}$ is vanishing. We recover the usual formula (\ref{eqn:DWpartitionFunctionClosedMfld}) for untwisted Dijkgraaf-Witten TQFT as follows:
\begin{align}
    \notag
    Z(Y) &= \frac{N^{b_1(Y)-1}}{|\Tor H_1(Y)|} \sum_{x,y \in \Tor H_1(Y)} e^{2\pi i N \lk(x,y)} \\ \notag
    & = N^{b_1(Y) - 1} |\{ x \in \Tor H_1(Y) \colon N x = 0 \mod 1\}| \\
    & = N^{b_1(Y) - 1}|\Tor H_1(Y) \otimes \Z/N\Z|=|G|^{-1}|\hom (\pi_1(Y),G)|.
\end{align}

More generally, given a bordered manifold $(M,\phi)$ with boundary $(T^2)^{\sqcup g}$, we adopt the description in terms of a pair of framed links $(\mathcal{J}, \mathcal{J}')$ in $S^3$. Then, the general surgery formula (\ref{abelian-TQFT-partition-function-surgery-formula}) for the partition function gives us:
\begin{equation}
    \langle c |M,\phi\rangle  = Z(\mathcal{J}, \mathcal{J}'; D, q;c) 
    = \frac{1}{N^{m+1}} \sum_{x \in \Z^m \otimes (\Z/N\Z)^2} \exp [ - 2\pi i (L \otimes q)(x,c) ].
\end{equation}
Recall that the linking matrix $L$ has the following form:
\begin{equation}
    L = \begin{pmatrix}
        J &F \\
        F^T & J'
    \end{pmatrix}.
\end{equation}
Thus, taking $x = (u,v) \in (\Z/N\Z)^m \oplus (\Z/N\Z)^m$ and $c= (a, b) \in (\Z/N\Z)^g \oplus (\Z/N\Z)^g$, we observe that:
\begin{equation}
    (L \otimes q) ((u,v),(a,b)) = \frac{1}{N} \left({u}^T J {v} + {a}^T J' {b} + {v}^T F {a} + {u}^T F {b} \right)\mod 1.
    \label{Lq-expression-untwisted-dw}
\end{equation}

Next, we impose a topological boundary condition. As we have already stressed, Dijkgraaf-Witten TQFTs always admit topological boundary conditions. In particular, in the untwisted case, one can always choose a Lagrangian subgroup $\mathfrak{L} \subseteq D=G\oplus G^*$ of the form is $\mathfrak{L} = G$, which is equivalent to the Neumann boundary condition considered in Section \ref{sec:untwisted-dw}. In the case of the cyclic gauge group, the corresponding boundary state then reads:
\begin{equation}
\langle \mathfrak{L}| = \sum_{a\in (\Z/N\Z)^g}\langle (a,0)|,  
\label{lagrangian-state-cyclic-dw}
\end{equation}
and the quadratic form (\ref{Lq-expression-untwisted-dw}) can be reduced to:
\begin{equation}
    (L \otimes q) ((u,v),(a,0)) = \frac{1}{N} \left({u}^T J {v} +  {v}^T F {a}  \right)\mod 1.
\end{equation}

The next step is applying a more general reciprocity formula for Gauss sums proven by Deloup and Turaev in~\cite{deloup2007reciprocity}. Doing so, we find:
\begin{multline}
     \langle (a,0) |M,\phi\rangle  =  \, \frac{1}{N |\det{\tilde{J}}|} \sum_{x \in (\Z/N\Z)^{b_1(Y)}} \exp \left[ \frac{2 \pi i}{N} x^T F_{\text{free}} \, a  \right] \times \\
       \sum_{r,s \in \coker{\tilde{J}}} \exp \left[ 2\pi i ( N r^T \tilde{J}^{-1} s + r^T \tilde{J}^{-1} F_{\text{tor}} \, a ) \right],
\end{multline}
where $\tilde{J}$ is the non-singular matrix such that $P^T J P = \tilde{J} \oplus \mathbf{0}_{b_1(Y)}$ for some unimodular matrix $P$, and $F_{\text{tor}}\oplus F_{\text{free}}=P^TF$. Moreover, the result does not depend on $F_{\text{tor}}a$ but rather on its image $\theta_\text{tor}$ in  $\coker{\tilde{J}} \otimes \Z/N\Z$. Observe that $\coker{\tilde{J}}$ and $(\tilde{J}^{-1} \mod 1)$ are respectively the torsion part of the first homology group and the linking form of the closed 3-manifold $Y$ obtained by gluing back $g$ solid to $M$ according to the map $\phi$.

The first sum forces $\theta_\text{free}\coloneqq(F_{\text{free}}a \mod N)=0\in \Z^{b_1(Y)}\otimes \Z/N\Z$, otherwise the partition function vanishes. Similarly, the second sum forces $\theta_\text{tor} = 0 \in \Tor H_1(Y) \otimes \Z/N\Z$. We obtain:
\begin{equation}
     \langle (a,0) |M,\phi\rangle   = N^{b_1(Y) - 1} \,|\Tor H_1(Y) \otimes \Z/N\Z|\,  \delta(\theta\equiv \theta_\text{tor}\oplus\theta_\text{free} = 0).
\end{equation}

Recall that $F$, composed with the projection to $\coker J$, gives the map $f \colon H_1(V_1^{\sqcup g})\cong \Z^g \to H_1(Y)$. Therefore, the vanishing condition on $\theta$ is equivalent to the request that $a$ belongs to the kernel of $f \otimes \Z/N\Z$. 

Finally, evaluating the partition function on the boundary state (\ref{lagrangian-state-cyclic-dw}) yields:
\begin{equation}  \label{eqn:DWpartitionFunctionExample}
    Z(M;\mathfrak{L})\equiv \langle \mathfrak{L}|M,\phi\rangle =N^{b_1(Y) -1} \,|\Tor H_1(Y)\otimes \Z/N\Z| \cdot|\ker(f \otimes \Z/N\Z)|,
\end{equation}
which is in agreement with (\ref{part-fn-untwisted-dw-neumann-via-f}).

\section{Summing over topologies}
\label{sec:sum-over-topologies}

One of the main goals of this work is to analyze the convergence properties, as well as the dependence on various input parameters, of the sum over topologies, rewritten as a sum over certain algebraic data (\ref{top-gravity-sum-over-alg-data}). In the previous section, in the case when the same topological boundary condition is imposed on all connected components of the boundary and the 3-manifold itself is connected, we described what this algebraic data is and the dependence of the TQFT partition function on it. Namely, it is a triple $\mathcal{A}\equiv (A,\ell,f:\Z^g\rightarrow A)$, when the 3-manifold has non-empty boundary, or a couple $\mathcal{A}\equiv (A,\ell)$ when it is closed. Note that we want to formally distinguish the cases when the boundary is a non-empty collection of spheres, that is $g=0$ and $f$ is the unique map from the trivial group, and when the boundary is empty, that is, when $f$ is absent.  When the 3-manifold is disconnected, the algebraic data $\mathcal{A}$ is then a collection of such triples and couples for each connected component. Namely, in (\ref{top-gravity-sum-over-alg-data}) let $Y=\sqcup_{i}Y^{(i)}$, where each $Y^{(i)}$ is connected. Respectively, we have $\iota=\sqcup_i \iota^{(i)}$ with $\mathrm{Im}\,\iota^{(i)}\subset Y^{(i)}$. Then $\mathcal{A}=\{\mathcal{A}^{(i)}\}_i$ where $\mathcal{A}^{(i)}$ is a triple/couple corresponding to the pair $(Y^{(i)},\iota^{(i)})$. That is $\mathcal{A}^{(i)}=(H_1(Y^{(i)}),\lk(Y),\iota_*^{(i)}|_{H_1})$. In order to make things as explicit and computable as possible, we will make various assumptions on the summation weights $\mu_\text{alg}(\mathcal{A})$. From now on we will suppress the ``alg'' subscript and will simply use $\mu(\mathcal{A})\equiv \mu_\text{alg}(\mathcal{A})$.

The first important assumption, which is already implicitly made in writing (\ref{top-gravity-sum-over-alg-data}), is that  $\mu(\mathcal{A})$ is, as a function of the algebraic data $\mathcal{A}$, the same for different choices of $\Sigma$ compatible with that algebraic data. In particular, for a contribution of a connected 3-manifold, it depends only on the total genus of $\Sigma$: that is, it is the same for connected $\Sigma=\Sigma_{g}$ and disconnected $\Sigma=\Sigma_{g_1}\sqcup \Sigma_{g_2}\sqcup\ldots$ with $g=g_1+g_2+\ldots$.  By default, we will also assume that $\mu(\mathcal{A})\geq 0$, so that $\mu$ can be understood as a choice of measure.

Next, we make the natural assumption that the measure factorizes with respect to the disjoint union, up to the automorphism factor corresponding to the permutation of the closed connected components:
\begin{equation}
    \mu(\mathcal{A})=\frac{1}{\prod_{m}N_m!}\,\prod_i \mu(\mathcal{A}^{(i)}).
    \label{eqn:measure-factorizes}
\end{equation}
Here $N_m$ are the numbers of closed connected components with the same algebraic data, that is, the numbers of repeating triples in $\mathcal{A}$. Such a factor is not introduced for components with boundary, because the components of $\Sigma$ are considered ordered. This assumption is along the lines of what has been previously considered in the literature (see e.g. \cite{Banerjee:2022pmw}).

\subsection{Definition and basic properties}
\label{sec:cosmologicalPartitionFunction}

With the assumptions above, for the topological boundary condition corresponding to a Lagrangian subgroup $\mathfrak{L}\subset D$, the sum (\ref{top-gravity-sum-over-alg-data}) for general $\Sigma$ can be expressed through the sums:
\begin{multline}
\mathcal{Z}_g(D,q;\mathfrak{L};\mu)\coloneqq
\\
\sum_{A\in \left\{\substack{\text{finitely gen.}\\\text{ab. groups}}\right\}/\sim}\;\;
\sum_{\ell\in\left\{\substack{\text{bil. forms}\\\text{on }\Tor A}\right\}/\sim}
\;\;\sum_{f\in \hom(\Z^g,A)}\mu(A,\ell,f) \sum_{c\in \mathfrak{L}^g} Z(A,\ell,f;D,q;c)
\label{genus-g-sum-algebraic-explicit}
\end{multline}
 corresponding to contributions of \textit{connected} 3-manifolds with the boundary of total genus $g$. Namely, the original full sum for the general boundary $\Sigma$ can be expressed as a polynomial in $\mathcal{Z}_g$ for different $g$, see Section \ref{sec:general-distribution} for details. We will refer to $\mathcal{Z}_g$ as the \textit{connected} gravitational partition function. As before, in the formula above, the pair $(D,q)$ determines the bulk 3d abelian TQFT. The summand in the right-hand side is the partition function (\ref{partition-function-alg-dependence}).

 It is worth clarifying the summation range and order. The most external sum is over finitely generated abelian groups $A$, up to isomorphism. We order such groups first by their rank (a non-negative integer), and then by the size of the torsion subgroup $\Tor A\subset A$. With the rank and the size of the torsion subgroup fixed, there is a finite number of isomorphism classes, so any order can be chosen there. For a fixed group $A$, one then sums over bilinear forms $\ell:\Tor A\otimes \Tor A\rightarrow \Q/\Z$, up to the equivalence relation described in Section \ref{sec:finiteForms}. Finally, for a fixed pair $(A,\ell)$, the summation is performed over $f\in \hom(\Z^g,A)\cong A^g$. We can order this set lexicographically, using the already chosen order on $A$. Note that, in principle, the summation over $\hom(\Z^g,A)$ leads to overcounting, as there may be residual automorphisms of $A$ that preserve $\ell$ (the bilinear form on the torsion subgroup $A$) but relate different choices of $f$. For different $f$ that can be related by such an automorphism, the partition function $Z(A,\ell,f;D,q;c)$ is necessarily the same. This is not an issue for our purposes. However, it implies that the actual dependence on the measure $\mu$ is only through the partial sums $\sum_{f\text{ related by automorphisms}}\mu(A,\ell,f)$, that is the measures of certain subsets of the set of all possible triples $(A,\ell,f)$. Note that $\mu(A,\ell,\cdot)$ for a fixed pair $(A,\ell)$ can be understood as a measure on $A^g$. In order to explicitly avoid the overcounting, one would need to classify the triples $(A,\ell,f)$ (instead of just pairs $(A,\ell)$) up to automorphisms. This goes outside the scope of this paper.

We will now derive sufficient   on the measure $\mu$ for the series to be convergent, and also for the dependence on the genus $g$ to have the form of a Dirichlet series:
\begin{equation}
   \mathcal{Z}_g(D,q;\mathfrak{L};\mu)=\sum_{n=1}^{\infty}b_n\,n^{g},
   \label{Dirichlet-series-genus-g-general}
\end{equation}
where the coefficients $b_n$ are $g$-independent. The latter property is crucial and will be used in Sections \ref{sec:general-distribution} and \ref{sec:holography}.

Let us focus on the sums over $f$ and $c$ for a fixed pair $(A,\ell)$. As was argued in Section \ref{sec:top-bc}, the TQFT partition function depends on $c\in \mathfrak{L}^g \subset D^g$ and $f\in \hom(\Z^g,A)\cong A^g$ only through $\theta\equiv \sum_{i=1}^gc_i\otimes f_i\in \mathfrak{L}\otimes A$, so one can write:
\begin{equation}
    Z(A,\ell,f;D,q,c)\equiv Z(A,\ell;D,q;\sum_{i=1}^g c_i\otimes f_i).
\end{equation}
Therefore, the sum over $f$ and $c$ in (\ref{genus-g-sum-algebraic-explicit}) for a given pair $(A,\ell)$ can be rewritten as follows:
\begin{multline}
    \mathcal{Z}_g^{(A,\ell)}(D,q;\mathfrak{L};\mu) = \sum_{f\in \hom(\Z^g,A)}\mu(A,\ell,f) \sum_{c\in \mathfrak{L}^g} Z(A,\ell,f;D,q;c)=
    \\
    \sum_{\substack{f\in A^g\\c\in \mathfrak{L}^g}} \sum_{\theta\in \mathfrak{L}\otimes A}\sum_{\alpha\in (\mathfrak{L}\otimes A)^*}\frac{\mu(A,\ell,f)}{|\mathfrak{L}\otimes A|}\,Z(A,\ell;D,q;\theta) \, e^{2\pi i\alpha\left(\theta-\sum_{i=1}^gc_i\otimes f_i\right)}=\\
    \sum_{\substack{\alpha\in (\mathfrak{L}\otimes A)^*\\\cong \hom(A,\mathfrak{L}^*)}}\frac{|\mathfrak{L}|^g\,\mu(A,\ell,(\ker \alpha)^g)}{|\mathfrak{L}\otimes A|} \sum_{\theta\in \mathfrak{L}\otimes A}Z(A,\ell;D,q;\theta) \, e^{2\pi i\alpha\left(\theta\right)},
    \label{sum-fixed-A-ell}
\end{multline}
where:
\begin{equation}
    \mu(A,\ell,S^g)\coloneqq\sum_{f\in S^g\subset A^g}\mu(A,\ell;f)\equiv \mu(\{(A,\ell,f)\}_{f\in S^g}),\qquad S\subseteq A,
\end{equation}
is the measure of the subset $\{(A,\ell,f)\}_{f\in S^g}$. In the last equality, $\ker \alpha$ is understood as the kernel of the map $\alpha:A\rightarrow \mathfrak{L}^*$. From the above, one can see that $\mu(A,\ell,\cdot)$, understood as a measure on $A^g$, actually only needs to be defined on the $\sigma$-algebra generated by the subgroups of the form $S^g\subseteq A^g$.

First, let us note that the sum:
\begin{equation}
    Z_\mathfrak{L}(A,\ell;D,q;\alpha)\coloneqq \frac{|\mathfrak{L}|}{|\mathfrak{L}\otimes A|}\sum_{\theta\in \mathfrak{L}\otimes A}Z(A,\ell;D,q;\theta) \,e^{2\pi i\alpha\left(\theta\right)}
\end{equation}
is the partition function of the 3d TQFT with gauged non-anomalous 1-form symmetry $\mathfrak{L}\subset D$, since $\theta\in \mathfrak{L}\otimes A\cong H_1(Y;\mathfrak{L})\cong H^2(Y;\mathfrak{L})$. It has dual 0-form symmetry $\mathfrak{L}^*$, and  $\alpha \in\hom(A,\mathfrak{L}^*)\cong H^1(Y;\mathfrak{L}^*)$ plays the role of its background gauge field (see e.g. \cite{Furlan:2024pvy} for the general description of gauging higher form symmetries in TQFTs in terms of cohomology groups). Since a Lagrangian subgroup is gauged, the resulting theory is trivial in the absence of the background 0-form field. This trivial theory can be understood as the theory on the other side of a boundary with topological boundary condition corresponding to $\mathfrak{L}$. Considered in the presence of a background 0-form symmetry $\mathfrak{L}^*$, the theory is an invertible TQFT (an SPT in condensed matter terminology), and we have:
\begin{equation}
    |Z_\mathfrak{L}(A,\ell;D,q;\alpha)|=1.
    \label{gauged-1-form-absolute-value}
\end{equation}
We remark that gauging a Lagrangian subgroup of 1-form symmetry has also been considered in \cite{Benini:2022hzx} in the context of 3d quantum gravity, as an alternative to summing over topologies. 

The property (\ref{gauged-1-form-absolute-value}) can be used to obtain an absolute bound for the total sum in (\ref{sum-fixed-A-ell}). Assuming the positivity of the measure, we have:
\begin{multline}
    \left|\mathcal{Z}_g^{(A,\ell)}(D,q;\mathfrak{L};\mu)\right|\leq 
    |\mathfrak{L}\otimes  A|\,|\mathfrak{L}|^{g-1}\,\mu(A,\ell,A^g)=\\
        |\mathfrak{L}\otimes \Tor A|\,|\mathfrak{L}|^{g-1+\mathrm{rank}\,A}\,\mu(A,\ell,A^g)\\
        \leq
        |\Tor A|^{I(D)}\,|D|^\frac{g-1+\mathrm{rank}\,A}{2}\,\mu(A,\ell,A^g),
\end{multline}
where $I(D)$ is the minimal number of generators of the group $D$ (i.e. the minimal number of cyclic group factors) and we used the facts that $I(\mathfrak{L})\leq I(D)$ and $|D|=|\mathfrak{L}|^2$. Now, assume that the measure satisfies the following bound:
\begin{equation}
    \mu(A,\ell,A^g)\equiv \mu(\{(A,\ell,f\}_{f\in A^g})\leq C\,\lambda^{\mathrm{rank}\,A}\,|\Tor A|^{-s},
    \label{measure-sufficient-bound}
\end{equation}
for some $s\in \R$, $\lambda \in \R_{\geq 0}$ and $C>0$. This immediately implies the following absolute bound on the series (\ref{genus-g-sum-algebraic-explicit}):
\begin{multline}
    |\mathcal{Z}_g(D,q;\mathfrak{L};\mu)|\leq \sum_{b_1\geq 0}\sum_{n\geq 1}C\,|D|^\frac{g+b_1-1}{2}\,\lambda^{b_1}\,a(n)\,n^{I(D)-s}=
    \frac{C\,|D|^\frac{g-1}{2}}{1-\lambda|D|^{\frac{1}{2}}}\,\Xi\left(s-I(D)\right),
\end{multline}
where $a(n)$ are the coefficients in the generating Dirichlet series (\ref{dir-gen-series-bil-forms}). The sums are therefore absolutely convergent for $s>I(D)+1$ and $\lambda<|D|^{-1/2}$. These conditions depend on the choice of group $D$, that is, the choice of the 3d abelian TQFT. To ensure convergence for all $D$ with the same measure $\mu$, it is sufficient to require that for any $\lambda$ and $s$, there exists $C$ such that (\ref{measure-sufficient-bound}) is satisfied for any finite abelian group $A$. That is, $\mu(A,\ell,A^g)\rightarrow 0$ faster than exponentially as $\mathrm{rank}\,A\rightarrow \infty$ and faster than a power law as $|\Tor A|\rightarrow\infty$.

Having obtained sufficient conditions for the absolute convergence of the sum (\ref{genus-g-sum-algebraic-explicit}), we will now present sufficient conditions for the result to be of the form of a Dirichlet series (\ref{Dirichlet-series-genus-g-general}). First, assume the measure factorizes symmetrically with respect to the $g$ boundary tori, that is:
\begin{equation}
    \mu(A,\ell,f)=\mu(A,\ell)\prod_{i=1}^g\nu(A,\ell,f_i),
\end{equation}
for some $f$-independent $\mu(A,\ell)$ and $\nu(A,\ell,f_i)$ depending on a single component $f_i\in A$. It then follows that:
\begin{equation}
    \mu(A,\ell,S^g)=\mu(A,\ell)\,\nu(A,\ell,S)^g,
\end{equation}
where:
\begin{equation}
    \nu(A,\ell,S)\coloneqq\sum_{\phi\in S\subset A}
    \nu(A,\ell,\phi).
\end{equation}
The form $(\ref{Dirichlet-series-genus-g-general})$ is then achieved if:
\begin{equation}
     \nu(A,\ell,S)\in \Z_{\geq 0},\qquad \forall S\subseteq A.
     \label{measure-integrality-condition}
\end{equation}
The coefficients of the Dirichlet series are then given by:
\begin{equation}
    b_n=\sum_{(A,\ell)} \frac{\mu(A,\ell)}{|\mathfrak{L}\otimes A|}\sum_{\substack{\theta\in \mathfrak{L}\otimes \Tor A\\\alpha\in \hom(A,\mathfrak{L}^*):\;\nu(A,\ell,\ker\alpha)|\mathfrak{L}|=n}}Z(A,\ell;D,q;\theta) \, e^{2\pi i\alpha\left(\theta\right)},
    \label{Dirichlet-coefficients-general-explicit}
\end{equation}
where we used the fact that the partition function vanishes unless $\theta\in \mathfrak{L} \otimes \Tor A$ \cite{deloup2001abelian}. Note that $b_n=0$ unless $n$ is a multiple of $|\mathfrak{L}|$. 

\subsection{An example of measure}
\label{sec:measure-example}

In this section, we will present some additional conditions on the measure for which the sum over topologies can be simplified to the point that it can be evaluated for a simple enough choice of the 3d TQFT.

The integrality condition (\ref{measure-integrality-condition}) can be satisfied, for example, if:
\begin{equation}
    \nu(A,\ell,\phi)=\left\{
    \begin{array}{@{}cl}
         1 & \text{if } \phi \in \Tor A,  \\
         0 & \text{otherwise},
    \end{array}
    \right.
    \label{nu-simple-choice}
\end{equation}
so that:
\begin{equation}
    \nu(A,\ell,S)=|\Tor S|\in \Z_{\geq 0}.
\end{equation}
In this case moreover, since $\Tor\ker\alpha=\ker \alpha|_{\Tor A}$, it is possible to reduce the sum over all finitely generated abelian groups $A\cong \Z^{b_1}\oplus \Tor A$ to just finite abelian groups $\tilde{A}=\Tor A$:
\begin{multline}
    \mathcal{Z}_g(D,q;\mathfrak{L};\mu)=\\
    \sum_{\tilde{A},\ell}\sum_{\vphantom{\tilde{A}}b_1\geq 0}\mu(\Z^{b_1}\oplus \tilde{A},\ell) \,|D|^{\frac{b_1}{2}}
     \sum_{\theta\in \mathfrak{L}\otimes \tilde{A}}\sum_{\substack{\alpha\in (\mathfrak{L}\otimes \tilde{A})^*\\\cong \hom(\tilde{A},\mathfrak{L}^*)}}\frac{|\mathfrak{L}|^g\,|\ker{\alpha}|^g}{|\mathfrak{L}\otimes \tilde{A}|}\,Z(\tilde{A},\ell;D,q;\theta) \,e^{2\pi i\alpha\left(\theta\right)},
     \label{sum-over-topologies-tor-measure}
\end{multline}
where we also used the factorization property of the partition function (\ref{factorization-boundary-preserving-formula-alg}) with $Z(\Z,0;D,q;0)\equiv Z(S^2\times S^1,\emptyset;D,q;0)=1$.

The sum can be further evaluated if we assume the following factorization property of the measure:
\begin{equation}
    \mu(\mathbb{Z}^{b_1}\oplus \tilde{A})=
    \mu(\Z,0)^{b_1}\prod_{p\in \text{primes}}\mu(A_p,\ell_p),
\end{equation}
where $A_p$ is the $p$-torsion subgroup of: 
\begin{equation}  A=\Z^{b_1}\oplus\tilde{A}=\Z^{b_1}\oplus\bigoplus_{p\in\text{primes}}A_p,
\end{equation}
and $\ell_p$ is the corresponding component of the bilinear form $\ell=\bigoplus_p\ell_p$ on $\tilde{A}=\Tor A$, as in Section \ref{sec:finiteForms}.

The data $(D,q)$ describing the 3d abelian TQFT, and the Lagrangian subgroup $\mathfrak{L}\subset D$ can also be split with respect to the prime numbers. Namely, the discriminant group $D$ decomposes into its Sylow $p'$-subgroups $D=\bigoplus_{p'} D_{p'}$ and, in a similar fashion, $\mathfrak{L} = \bigoplus_{p'} \mathfrak{L}_{p'}$. Notice that all but a finite number of $D_{p'}$ and $\mathfrak{L}_{p'}$ are trivial groups. This splitting is orthogonal with respect to the quadratic form $q$, meaning that it decomposes as $q=\bigoplus_{p'} q_{p'}$. We then decompose $\theta\in \tilde{A}\otimes \mathfrak{L}$ respectively as $\theta=\sum_{p,p'}\theta_{p,p'}$, with $\theta_{p,p'}\in A_p\otimes \mathfrak{L}_{p'}$. Note that  $A_p\otimes \mathfrak{L}_{p'}=0$ unless $p=p'$.

Now, making use of the factorization properties of the TQFT partition function reviewed in Section \ref{sec:factorization}, we have:
\begin{multline}
    Z( \tilde{A},\ell;D,q;\theta) = \prod_{p':D_{p'}\neq 0}|D_{p'}|^{-\frac{1}{2}}\prod_{p:A_p\neq 0} |D_{p'}|^{\frac{1}{2}} \,Z(A_p,\ell_p;D_{p'},q_{p'};\theta_{p,p'})=\\
    \prod_{p:D_p\neq 0} Z(A_p,\ell_p;D_{p},q_{p};\theta_{p,p}),
\end{multline}
where we used that for $p\ne p'$, we have:
\begin{equation}
    {Z}(A_p,\ell_p;D_{p'},q_{p'};\theta_{p,p'}) = Z(A_p,\ell_p;D_{p'},q_{p'};0)=|D_{p'}|^{-\frac{1}{2}},
\end{equation}
where the second equality follows from the statement that the closed-manifold partition function of a TQFT admitting topological boundary conditions satisfies ${Z}(A, \ell; D, q;0) = {Z}(0, 0; D, q)=|D|^{-\frac{1}{2}}$, whenever $|A|$ and $|D|$ are coprime (see e.g. \cite{Kaidi:2021gbs}).

Using both the factorization of the measure and the factorization of the TQFT partition function, the sum (\ref{sum-over-topologies-tor-measure}) can also be decomposed with respect to the primes:
\begin{multline}
    \mathcal{Z}_g(D,q;\mathfrak{L};\mu)=
    \frac{1}{1-\mu(\Z,0)|D|^{\frac{1}{2}}} \cdot
    \prod_{p\in\text{primes}}
    \sum_{\substack{(A_p,\ell_p):\\\text{finite abelian}\\p\text{-groups}\\\text{with pairings}}}\mu(A_p,\ell_p) \times \\
     \sum_{\theta_{p}\in \mathfrak{L}_p\otimes {A}_p}\sum_{\substack{\alpha_p\in (\mathfrak{L}_p\otimes {A}_p)^*\\\cong \hom({A}_p,\mathfrak{L}_p^*)}}\frac{|\mathfrak{L}_p|^g\,|\ker{\alpha_p}|^g}{|\mathfrak{L}_p\otimes {A_p}|}\,Z(A_p,\ell_p;D_p,q_p;\theta_{p}) \, e^{2\pi i\alpha_p\left(\theta_{p}\right)},
     \label{sum-over-topologies-tor-factorized-measure}
\end{multline}
having introduced $\theta_p \equiv \theta_{p,p}$.
We consider now a very specific choice of measure that satisfies the factorization property above, with $\nu$ as in (\ref{nu-simple-choice}) and:
\begin{equation}
    \label{eqn:explicitMeasure}
    \mu(A,\ell)=\lambda^{\mathrm{rank}A}\, |\Tor A|^{-s}.
\end{equation}
This choice moreover saturates the bound (\ref{measure-sufficient-bound}) with $C=1$ and $s$ there replaced with $s-g$. We then have:
\begin{multline}
    \mathcal{Z}_g(D,q;\mathfrak{L};\mu)=
    \frac{1}{1-\lambda|D|^{\frac{1}{2}}}\cdot\prod_{\substack{p\in\text{primes}\\ p\nmid |D|}}\Xi_p(s-g)\times
    \\
\prod_{\substack{p\in\text{primes}\\ p\mid |D|}}
        \sum_{\substack{(A_p,\ell_p):\\\text{finite abelian}\\p\text{-groups}\\\text{with pairings}}}|A_p|^{-s}
     \sum_{\theta_{p}\in \mathfrak{L}_p\otimes {A}_p}\sum_{\substack{\alpha_p\in (\mathfrak{L}_p\otimes {A}_p)^*\\\cong \hom({A}_p,\mathfrak{L}_p^*)}}\frac{|\mathfrak{L}_p|^g\,|\ker{\alpha_p}|^g}{|\mathfrak{L}_p\otimes {A_p}|}\,Z(A_p,\ell_p;D_p,q_p;\theta_{p}) \, e^{2\pi i\alpha_p\left(\theta_{p}\right)},
     \label{sum-over-topologies-specific-measure}
\end{multline}
where $\Xi_p$ is the generating function considered in Section \ref{sec:counting-groups-with-forms}. In particular, for the trivial TQFT with $D=0$ we have:
\begin{equation}
     \mathcal{Z}_g(0,0;0;\mu)=
    \frac{\Xi(s-g)}{1-\lambda}.
\end{equation}

In the case of untwisted Dijkgraff-Witten TQFT with Neumann boundary conditions, the expression~\eqref{sum-over-topologies-specific-measure} can be simplified further. Let $G$ the gauge group, so that $D=G\oplus G^*$ and $\mathfrak{L}=G$. The partition function of untwisted DW theory has the following simple form (see Section \ref{sec:cyclicDW}):
\begin{equation}
    Z(A,\ell;D,q;\theta)=\frac{|G\otimes A|}{|G|} \,\delta(\theta).
    \label{untwisted-DW-partition-function}
\end{equation}
We then have:
\begin{multline}
    \mathcal{Z}_g(D,q;\mathfrak{L};\mu)=
    \frac{1}{1-\lambda|D|^{\frac{1}{2}}}\cdot\prod_{\substack{p\in\text{primes}\\ p\nmid |G|}}\Xi_p(s-g)\times
    \\
\prod_{\substack{p\in\text{primes}\\ p\mid |G|}}|G_p|^{g-1}
        \sum_{\substack{(A_p,\ell_p):\\\text{finite abelian}\\p\text{-groups}\\\text{with pairings}}}|A_p|^{-s}
     \sum_{{\alpha_p\in \hom({A}_p,G_p^*)}}|\ker{\alpha_p}|^g.
     \label{sum-over-topologies-specific-measure-DW}
\end{multline}

Let us now consider the simplest non-trivial choice of the gauge group: $G=\Z/p\Z$ for some prime $p\neq 3$. The factor in (\ref{sum-over-topologies-specific-measure-DW}) corresponding to this particular prime reads:
\begin{multline}
  p^{g-1}
        \sum_{(A_p,\ell_p)}|A_p|^{-s}
     \sum_{{\alpha_p\in \hom({A}_p,\Z/p\Z)}}|\ker{\alpha_p}|^g=
     \\
      p^{g-1}
        \sum_{(A_p,\ell_p)}|A_p|^{-s}
     \left((|A_p|p^{-1})^g \,(p^{\mathrm{dim}_{\Z/p\Z}(A_p\otimes\Z/p\Z)}-1)+|A_p|^g\right)=\\
     p^{-1}(p^g-1)
       \sum_{(A_p,\ell_p)}|A_p|^{g-s}+ p^{-1}
       \sum_{(A_p,\ell_p)}|A_p|^{g-s} \,p^{\mathrm{dim}_{\Z/p\Z}(A_p\otimes\Z/p\Z)}=\\
       p^{-1}(p^g-1)\,\Xi_p(s-g)+ p^{-1}\prod_{m\geq 1}\frac{1+p^{1-m(s-g)}}{1-p^{1-m(s-g)}},
\end{multline}
where in the first equality we used the fact that for all maps $\alpha_p:A_p\rightarrow \Z/p\Z$, except the zero one, we have $|\ker{\alpha_p}|=|A_p|p^{-1}$ and the total number of maps is equal to $p^{\mathrm{dim}_{\Z/p\Z}(A_p\otimes\Z/p\Z)}$, and in the last equality we used a refined version of (\ref{counting-pnot3-groups}):
\begin{equation}
    \sum_{(A_p,\ell_p)}|A_p|^{-s}\,x^{\mathrm{dim}_{\Z/p\Z}(A_p\otimes\Z/p\Z)}=\prod_{m\geq 1}\frac{1+xp^{-ms}}{1-xp^{-ms}}.
    \label{sum-over-pnot3-groups-refined}
\end{equation}
Note that $\mathrm{dim}_{\Z/p\Z}(A_p\otimes\Z/p\Z)\equiv b$ in the decomposition:
\begin{equation}
    A_p=\bigoplus_{i=1}^b\Z/p^{m_i}\Z.
\end{equation}
We arrive at the following explicit expression for the connected gravitational partition function:
\begin{equation}
     \mathcal{Z}_g(D,q;\mathfrak{L};\mu)=
    \frac{\Xi(s-g)}{1-\lambda|D|^{\frac{1}{2}}} \cdot 
    \left(
  p^{-1}(p^g-1)+ p^{-1}\prod_{m\geq 1}\frac{1+p^{1-m(s-g)}}{1-p^{1-m(s-g)}}\,\frac{1-p^{-m(s-g)}}{1+p^{-m(s-g)}}
    \right).
\end{equation}

\section{Distribution of 2d TQFTs for general abelian 3d TQFTs}
\label{sec:general-distribution}

Naively applying the holographic principle, we would expect that summing over all the possible bulk topologies with a fixed boundary yields a new theory on the boundary itself. However, we should be careful in identifying the theory living on the boundary, since considering possibly disconnected bulks spoils the factorization property of the theory with respect to the disjoint union. Indeed, recall that one expects the partition function $Z$ of a theory to factorize in the following way:
\begin{equation}
    Z(M \sqcup N) = Z(M) \cdot Z(N).
\end{equation}
However, this property will fail to hold when applied to the sum over topologies, due to the presence of wormhole topologies, that is, connected bordisms between disconnected components of the possibly disconnected boundary. A proposed generalization of the holographic principle suggests promoting the dual theory to a statistical ensemble of theories. In the following, we will analyze the consequences of such a bulk-boundary correspondence.

Since we were interested in abelian 3d TQFTs on the bulk, we shall focus on ensembles of 2d theories. Moreover, we have already argued that the dual theories can be considered topological when appropriate boundary conditions are imposed. 
Thus, in this section, we aim to find some constraints on the boundary ensemble of theories, given an arbitrary abelian 3d TQFT on the bulk. We will see that the condition of being abelian will play a crucial role.

\subsection{Bulk-boundary correspondence}

Given a possibly disconnected closed and oriented 2-manifold $\Sigma$, we expect a correspondence:
\begin{equation}
\mathcal{ZG}(\Sigma;\mathcal{T}) \leftrightarrow \avg{\mathrm{T}_\xi(\Sigma)}_\xi,
\end{equation}
where $\mathcal{ZG}(\Sigma;\mathcal{T})$  is the ``gravity'' partition function with abelian 3d TQFT $\mathcal{T}$ in the bulk, while $\mathrm{T}_\xi$ denotes a family of 2d TQFTs on the boundary, parameterised by some data $\xi$. We use $\avg{\cdot}_\xi$ for the average over the statistical ensemble of theories $\mathrm{T}_\xi$.
For the gravity side, we will consider the sum over bulk topologies $\mathcal{ZG}(\Sigma;\mathcal{T}) =\langle \mathfrak{B}|\Sigma;\mu\rangle_\mathcal{T}$ as in (\ref{top-gravity-sum-over-alg-data}), in the case where the same topological boundary condition $\mathfrak{L}$ is imposed on all connected boundary components. Note that $\langle\mathfrak{B}| \equiv \bigotimes_i \langle\mathfrak{L}|^{(i)}$ implicitly depends on the number of connected components of $\Sigma$. In Section \ref{sec:sum-over-topologies} we analyzed its convergence and pointed out that it can be expressed through the connected gravitational partition functions (\ref{genus-g-sum-algebraic-explicit}).

We now address the problem of the lack of factorization of $\mathcal{ZG}$, which explains the need for a statistical ensemble of theories on the boundary. By lack of factorization, we mean that $\mathcal{ZG}$ is not multiplicative under taking a disjoint union of boundary 2-manifolds, hence it cannot be a partition function of a single 2d theory. 

Specifically, consider the sum over topologies $\mathcal{ZG}(\Sigma)$ and suppose that $\Sigma$ has two disconnected components $\Sigma = \Sigma_g \sqcup \Sigma_{g'}$, then $\mathcal{ZG}(\Sigma)$ will get contribution from two different kinds of terms, which are depicted in Figure~\ref{fig:wormholes}. 
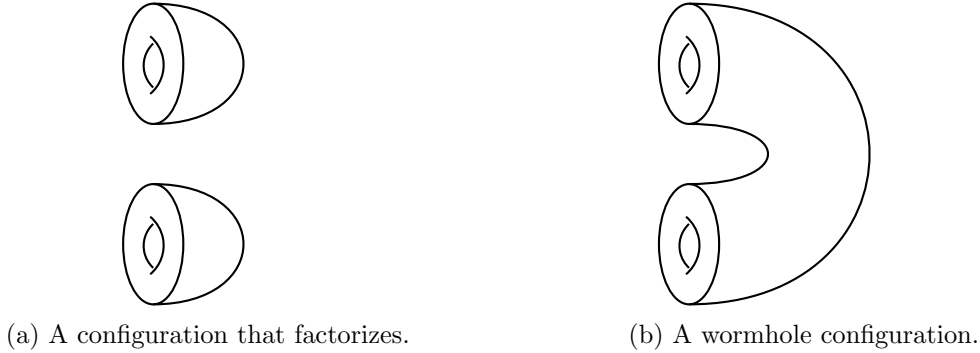
\begin{figure}[t]
\centering
\subcaptionbox{A configuration that factorizes.}[.4\textwidth]{
\scalebox{2}{\begin{tikzpicture}
\draw (0,0) ellipse (0.2 and 0.4);
\draw (-0.02,-0.2) arc (-50:50:0.25);
\draw (0,0.135) arc (133:227:0.20);
\draw (0,-0.4) .. controls (.8,-0.4) and (.8, 0.4) .. (0,0.4);
\draw (0,-1.2) ellipse (0.2 and 0.4);
\draw (-0.02,-0.2-1.2) arc (-50:50:0.25);
\draw (0,0.135-1.2) arc (133:227:0.20);
\draw (0,-0.4-1.2) .. controls (.8,-0.4-1.2) and (.8, 0.4-1.2) .. (0,0.4-1.2);
\end{tikzpicture}}
}%
\hspace{0.1\textwidth} 
\subcaptionbox{A wormhole configuration.}[.4\textwidth]{
\scalebox{2}{\begin{tikzpicture}
\draw (2.4,0) ellipse (0.2 and 0.4);
\draw (2.4-0.02,-0.2) arc (-50:50:0.25);
\draw (2.4,0.135) arc (133:227:0.20);
\draw (2.4,-1.2) ellipse (0.2 and 0.4);
\draw (2.4-0.02,-0.2-1.2) arc (-50:50:0.25);
\draw (2.4,0.135-1.2) arc (133:227:0.20);
\draw (2.4,-0.4) .. controls (3.1,-0.4) and (3.1, -0.8) .. (2.4,-.8);
\draw (2.4,-1.6) .. controls (4,-1.6) and (4, 0.4) .. (2.4, 0.4);
\end{tikzpicture}}
}%
\caption{Two different kind of terms contribute to $\mathcal{ZG}(T^2 \sqcup T^2)$.}
\label{fig:wormholes}
\end{figure}
 Any $3$-manifold in the sum will be a disjoint union of a number $n$ of closed components, while its boundary will be either bounding a single 3-manifold or it will be split into two different connected components. Then, the algebraic data $\mathcal{A}$ is actually a list $\{ \mathcal{A}^{(i)}\}_i$ of $n$ couples corresponding to closed components, together with either one or two triples corresponding to the manifolds with boundary. Therefore, recalling that we assumed the measure $\mu$ to factorize with respect to disjoint union, as in (\ref{eqn:measure-factorizes}), we obtain:
\begin{align}
    \mathcal{ZG}(\Sigma) = 
     \sum_{n=0}^\infty \frac{1}{n!} \mathcal{Z}_0^n \cdot \left(\mathcal{Z}_{g+g'} + \mathcal{Z}_g \,\mathcal{Z}_{g'} \right)
    =  e^{\mathcal{Z}_0} (\mathcal{Z}_{g+g'} + \mathcal{Z}_g \,\mathcal{Z}_{g'}),
    \label{eqn:lackOfFactorization}
\end{align}
where $\mathcal{Z}_g$ is the connected gravitational partition function (\ref{genus-g-sum-algebraic-explicit}), depending on the TQFT data $\mathcal{T} = (D,q)$ and on the topological boundary condition $\mathfrak{L}$ that we kept implicit to lighten the notation. As pointed out in Section \ref{sec:top-bc}, the partition function of an abelian TQFT on a connected 3-manifold with a topological boundary condition depends only on the total genus of the boundary component (up to a possibly decoupled 2d TQFT supported on the boundary). Assuming the same holds for the summation measure (as in Section \ref{sec:sum-over-topologies}), the connected gravitational partition function only depends on the total genus of the $2$-manifold. This is a key step in obtaining (\ref{eqn:lackOfFactorization}).

We remark that similar expressions can be obtained for $\Sigma$ with any number of connected components. The computation ultimately reduces to enumerating all possible ways of connecting the boundary components via $3$-dimensional bordisms.
In particular, the gravitational partition function of the vacuum is given by $\mathcal{ZG}(\emptyset) = e^{\mathcal{Z}_0}$.

On the other hand, if $\mathcal{ZG}$ were to factorize, we would expect $\mathcal{ZG}(\Sigma_g \sqcup \Sigma_{g'})$ to be equal to:
\begin{equation}
    \mathcal{ZG}(\Sigma_g) \, \mathcal{ZG}(\Sigma_{g'}) =  e^{2\mathcal{Z}_0}  \mathcal{Z}_{g} \, \mathcal{Z}_{g'}.
\end{equation}

Thus, the lack of factorization can be understood as a consequence of the presence of the so-called wormhole configurations. This peculiar choice of words is dictated by the interpretation of each connected component of the boundary as a parallel universe. Hence, a wormhole is a bulk configuration connecting two different universes. 

Nonetheless, one could be tempted to force the factorization of $\mathcal{ZG}$. However, we show that requiring factorization is equivalent to $\mathcal{Z}_g = 0$ for any closed surface of genus $g$. The fact that this condition is sufficient for factorization is straightforward since everything trivializes. On the other hand, assuming factorization for $\Sigma = (S^2)^{\sqcup k}$, one can easily check that it implies $\mathcal{Z}_0 = 0$. Then, we can extend this result to any other genus $g$ by requiring the factorization for $S^2 \sqcup \Sigma_g$. We stress that we have extensively used the property that the connected gravitational partition function depends only on the total genus $g$ of $\Sigma$.

We note that untwisted Dijkgraaf-Witten theories cannot satisfy this vanishing condition on $\mathcal{Z}_g$.
Hence, the lack of factorization is intrinsic to the sum over topologies and invalidates the usual holography principle. The way around it is that the bulk theory is not dual to a boundary theory, but rather to an ensemble of theories living on the boundary, as this would explain the failure of $\mathcal{ZG}$ to factorize. 

Let us formulate the bulk-boundary correspondence more precisely. Define $\zeta \coloneqq\mathcal{Z}_0$. Thus, the bulk-boundary correspondence is given by the following identification:
\begin{equation}
e^{-\zeta} \mathcal{ZG}(\Sigma) = \avg{\mathrm{T}_\xi(\Sigma)}_\xi.
\end{equation}
The extra factor on the left-hand side is required to satisfy the normalization condition for the distribution of the boundary theories:
\begin{equation}
    \avg{\mathrm{T}_\xi(\emptyset)}_\xi\equiv \avg{1}_\xi=1,
\end{equation}
where we used the fact that the partition function of any 2d TQFT on an empty spacetime is 1.

\subsection{Brief review of 2d TQFTs}

We shall explain in which way $\mathrm{T}_\xi$ constitutes a statistical ensemble of theories. To do so, we first recall some basic facts about 2d TQFTs.

It is well-known that (non-extended) 2d TQFTs are in one-to-one correspondence with commutative Frobenius algebras. Since we are dealing with unitary 2d TQFTs, the algebras get promoted to semi-simple Frobenius algebras. We briefly describe how this correspondence works. Suppose that $\mathrm{T}$ is a unitary 2d TQFT, then $\mathcal{V} \coloneqq \mathrm{T}(S^1)$ is a finite-dimensional vector space, which has the physical meaning of the Hilbert space of states on a circle. Moreover, $\mathcal{V}$ is a commutative Frobenius algebra, where the multiplication $m(\cdot, \cdot)$ is given by $\mathrm{T}$ applied to a pair of pants and the trace $\theta(\cdot)$ is given by $\mathrm{T}$ applied to a disk.

For a semi-simple algebra, we are free to choose a basis of idempotents $e_i$ for $i=1, \dots, N$. Next, denoting $\theta_i = \theta(e_i)$, one can show that:
\begin{equation}
\mathrm{T}_{N,\{\theta_i\}} (\Sigma_g) = \sum_{i=1}^N \theta_i^{1-g},
\label{2dTQFT-partition}
\end{equation}
for any Riemann surface $\Sigma_g$ of genus $g$.

Hence, we see that the data specifying our boundary $2$d TQFT is $\xi = (N,\{\theta_i\})$. Notice that $(N,\{\theta_i\})$ characterises the whole TQFT, even though we just need the value of its partition function on closed $2$-manifolds.

\subsection{Statistical ensemble of theories}

Now that the data defining the family of boundary $2$d TQFTs are clear, we focus on the meaning of this statistical ensemble. In the most general situation, we will have that the dimension of the Hilbert space $N$ is a random variable following a discrete probability distribution supported on the natural numbers and, similarly, we will have $N$ random variables $\{ \theta_i \}$ that follow the same, perhaps continuous, probability distribution supported on the reals. A priori, there is no reason to suppose that the $\theta_i$'s are independent from each other and $N$. Moreover, we will argue later that assuming such an independence condition is inconsistent with our bulk theory. For the moment, let us work in the above assumptions and denote the joint probability density function $f_{N,\theta_1, \dots, \theta_N}(n,\alpha_1,\dots, \alpha_n)$. We write:
\begin{equation}
f_{N,\theta_1, \dots, \theta_N}(n,\alpha_1,\dots, \alpha_n) = f_{\theta_1, \dots, \theta_N | N}(\alpha_1,\dots, \alpha_n | n) \, f_N(n).
\end{equation}
One could be tempted to naively guess that the $\theta_i$'s are independent of each other. However, this condition turns out to be too constraining, as we will discuss later. Nevertheless, without loss of generality (because the partition functions (\ref{2dTQFT-partition}) are symmetric in $\theta_i$), we will assume that distribution $f_{\theta_1, \dots, \theta_N | N}(\alpha_1,\dots, \alpha_n | n)$ is symmetric with respect to to the permutations of $\theta_i$.

The average of the partition functions of the boundary theories on a connected Riemann surface $\Sigma_g$ reads:
\begin{align}
\notag
\avg{\mathrm{T}_\xi(\Sigma_g)}_\xi &= \sum_{n= 0}^\infty \int d \alpha_1 \cdots d \alpha_n \,f_{N,\theta_1, \dots, \theta_N}(n,\alpha_1,\dots, \alpha_n) \,\mathrm{T}_{n, \{\alpha_i\}}(\Sigma_g) \\
& = \mathbb{E}[N \theta^{1-g}],
\end{align}
where $\theta \equiv \theta_i$ for some $i$.

For a disjoint union of two Riemann surfaces, $\Sigma = \Sigma_g \sqcup \Sigma_{g'}$, we have:
\begin{align}
\notag
\avg{\mathrm{T}_\xi(\Sigma_g \sqcup \Sigma_{g'})}_\xi 
& = \sum_{n= 0}^\infty f_N(n) \int \left( \prod_{i=1}^n d \alpha_i \right) \, f_{\theta_1, \dots, \theta_N | N}(\alpha_1,\dots, \alpha_n | n) \sum_{j=1}^n \alpha_j^{1-g} \sum_{k=1}^n \alpha_k^{1-g'} \\ \notag
& \
\begin{aligned}
 =  \sum_{n= 0}^\infty f_N(n) \biggl( \sum_{j \ne k} \int d\alpha_j  \,d\alpha_k \, f_{\theta, \theta' | N}&(\alpha_j,\alpha_k | n)\, \alpha_j^{1-g} \, \alpha_k^{1-g'}  \\ 
 & + \sum_{j}  \int d\alpha_j \, f_{\theta | N}(\alpha_j | n) \, \alpha_j^{2-g-g'}  \biggr)
 \end{aligned} \\
& = \mathbb{E}\left[ (N^2 - N)  \theta^{1-g} \theta'^{1-g'} \right] + \mathbb{E}[ N \theta^{2-g-g'}],
\end{align}
where $\theta \equiv \theta_i$ and $\theta' \equiv \theta_j$  for some $i\neq j$, and we use
\begin{equation}
    f_{\theta_{i_1},\ldots,\theta_{i_k}|N}(\alpha_{i_1},\ldots,\alpha_{i_k}|n):=\int \prod_{i\in \{1,\ldots,n\}\setminus \{i_1,\ldots,i_k\}} \hspace{-3em}d\alpha_i \;\;f_{\theta_1,\ldots,\theta_N}(\alpha_1,\ldots,\alpha_n|n)
\end{equation}
to denote the reduced probability densities.

The partition function of a disjoint union of $k$ Riemann surfaces is hard to express in a closed form. Nonetheless, it will be useful to manipulate its expression though we postpone the discussion until the next section.

We conclude with another simple example. Consider the disjoint union of $k$ tori:
\begin{align}
\notag
\avg{\mathrm{T}_\xi((T^2)^{ \sqcup k })}_\xi 
& =  \sum_{n= 0}^\infty f_N(n) \int \left( \prod_{i=1}^n d \alpha_i \right) f_{\theta_1, \dots, \theta_N | N}(\alpha_1,\dots, \alpha_n | n)  \prod_{l=1}^k \sum_{j=1}^n 1 \\
\notag
& =  \sum_{n= 0}^\infty f_N(n) \, n^k \\
& = \mathbb{E}[N^k].
\end{align}
Notice that such configurations yield all the statistical moments of the random variable $N$.

\subsection{The abelian condition}
\label{sec:abelian-condition}

It turns out that the bulk theory being abelian is a strong condition on the ensemble of dual theories. We introduce the \textit{abelian condition} as the set of constraints that the ensemble of 2d TQFTs has to satisfy to be consistent with an abelian theory on the bulk.

The crucial property of an abelian 3d TQFT is that it does not distinguish between disjoint unions and connected sums on the boundary (see Section \ref{sec:top-bc}). Hence (assuming the summation weights also respect this property, see Section \ref{sec:sum-over-topologies}), the connected gravitational partition function $\mathcal{Z}_g$ only depends on the total genus of the boundary surface.

For example, we deduce the following equations from the disjoint union of two Riemann surfaces and the disjoint union of $k$ tori, respectively:
\begin{align}
\avg{\mathrm{T}_\xi(\Sigma_g \sqcup \Sigma_{g'})}_\xi &= \avg{\mathrm{T}_\xi(\Sigma_g)}_\xi \, \avg{\mathrm{T}_\xi(\Sigma_{g'}) }_\xi+ \avg{\mathrm{T}_\xi(\Sigma_{g+g'} ) }_\xi, \\
\avg{\mathrm{T}_\xi((T^2)^{\sqcup k})}_\xi &= B_k\left( \avg{\mathrm{T}_\xi(T^2)}_\xi, \avg{\mathrm{T}_\xi(\Sigma_2)}_\xi, \dots, \avg{\mathrm{T}_\xi(\Sigma_k)}_\xi \right),
\end{align}
where $B_k(x_1, \dots, x_k)$ are the complete exponential Bell polynomials. 

These two equations are already very constraining on the distributions for $N$ and $\theta$. Indeed, plugging into these the results of the previous section we get:
\begin{align}
\label{eqn:AbelianConditionOnTwo}
&\mathbb{E}\left[ (N^2 - N)   \theta^{1-g} \theta'^{1-g'}  \right] + \mathbb{E}[ N \theta^{2-g-g'}] = \mathbb{E}[N \theta^{1-g}] \, \mathbb{E}[N \theta^{1-g'}] + \mathbb{E}[N \theta^{1-g - g'}], \\
& \mathbb{E}[N^k] = B_k \left(\mathbb{E}[N], \mathbb{E}[N\theta^{-1}], \dots, \mathbb{E}[N\theta^{1-k}]\right).
\label{eqn:AbelianConditionOnTori}
\end{align}

The general condition cannot be expressed in closed form, but we can still work out a list of constraints.

\begin{proposition}[The abelian condition]
    Consider a collection of $k$ Riemann surfaces $\{ \Sigma_{g_i}\}$ of genus $g_i$. The abelian condition then reads as follows:
        \begin{equation}
            \label{eqn:fullAbelianCondition}
            \sum_{p \in P(\{ g_1, \dots g_k\}) } \prod_{q \in p} \mathbb{E} \left[ N \theta^{1- \sum_{g \in q} g } \right]
            = \sum_{p \in P(\{ g_1, \dots g_k\}) }   \mathbb{E} \left[ \frac{N!}{(N-|p|)!} \prod_{q \in p} \theta_q^{\sum_{g \in q} (1-g)} \right],
        \end{equation}
where $P(\{ g_1, \dots g_k\})$ denotes the set of partitions of the set $\{ g_1, \dots g_k\}$ and, for a fixed partition $p$, $\theta_q\equiv\theta_{i_q}$ for some $i_q$ such that $i_{q}\neq i_{q'}$ for $q\neq q'\in p$.
\end{proposition}

Indeed, the left-hand side comes from taking into account all possible ways of joining the disconnected surfaces via connected sum. To be more precise, each contribution corresponds to a partition $p$ of the set $\{ \Sigma_{g_i}\}_i$ of Riemann surfaces, where each part $p_j$ of $p$ corresponds to the connected sum of all the surfaces in that part. Thus:
\begin{equation}
\avg{\mathrm{T}_\xi( \sqcup_{i=1}^k \Sigma_{g_i} )}_\xi = \sum_{p \in P(\{ g_1, \dots, g_k\}) } \prod_{q \in p} \avg{\mathrm{T}_\xi( \Sigma_{\sum_{g \in q} g} )}_\xi =  \sum_{p \in P(\{ g_1, \dots, g_k\}) } \prod_{q \in p} \mathbb{E} \left[ N \theta^{1- \sum_{g \in q} g } \right].
\end{equation}

Similarly, for the right-hand side we first directly compute:
\begin{equation}
\avg{\mathrm{T}_\xi( \sqcup_{i=1}^k \Sigma_{g_i} )}_\xi = \mathbb{E} \left[ \prod_{i=1}^k \sum_{j_i =1}^N \theta_{j_i}^{1-g_i} \right].
\end{equation}
To re-express this term, we notice that we have different contributions depending on whether or not certain $j_i$'s are equal. This is explicitly specified by the partition $p$ of the sequence $\{ j_1, \dots, j_k\}$, where the blocks consist of equal $j_i$'s. In particular, having equal indices corresponds to the contribution of the genera of the corresponding surfaces to add up. Finally, notice that any combination of indices has to be counted $N(N-1) \cdots (N-|p|+1)$ times. It follows that:
\begin{equation}
\mathbb{E} \left[ \prod_{i=1}^k \sum_{j_i =1}^N \theta_{j_i}^{1-g_i} \right] = \sum_{p \in P(\{ g_1, \dots g_k\}) }   \mathbb{E} \left[ \frac{N!}{(N-|p|)!} \prod_{q \in p} \theta_q^{\sum_{g \in q} (1-g)} \right],
\end{equation}
which gives the equality between the right-hand and the left-hand sides of the abelian condition stated in the proposition.

\subsection{Independent variables}

As we anticipated above, we want to rule out the possibility that the random variables $N$ and $\theta$ are independent. Indeed, we will deduce that this assumption and the abelian condition are too constraining together. 

Assume that all random variables $\theta_i$ are independent of each other and also independent of the random variable $N$, i.e.:
\begin{equation}
f_{\theta_1,\ldots,\theta_N|N}(\alpha_1,\ldots,\alpha_N | n) = \prod_{i=1}^N f_{\theta}(\alpha_i).  
\end{equation}
Then, we may rewrite the abelian conditions~\eqref{eqn:AbelianConditionOnTwo} and~\eqref{eqn:AbelianConditionOnTori}:
\begin{align}
&(N_2 - N_1) \,\theta_{(1-g)} \, \theta_{(1-g')} + N_1 \, \theta_{(2-g-g')} = N_1^2 \, \theta_{(1-g)} \,\theta_{(1-g')} + N_1 \, \theta_{(1-g-g')}, \label{eqn:1stAbelianConditionWithIndependence} \\
&N_k = B_k(N_1, N_1 \,\theta_{(-1)}, \dots, N_1 \,\theta_{(1-k)}) ,\label{eqn:2ndAbelianConditionWithIndependence}
\end{align}
where we have introduced the moments of the two distributions, namely $N_k \coloneqq \mathbb{E}[N^k]$ and $\theta_{(k)} \coloneqq \mathbb{E}[\theta^k]$.

Assuming that $N_1 > 0$, these conditions can only be solved by a degenerate distribution supported on the single value $\theta \equiv 1/\theta_{(-1)}$. Moreover, the $N$-distribution has to satisfy:
\begin{equation}
N_k = B_k(N_1, N_1 \lambda, \dots, N_1 \lambda^{k-1}),
\end{equation}
having denoted $\lambda \coloneqq \theta_{(-1)}$. This implies that the probability density function for $N$ is:
\begin{align}
f_{N}(n) = \left\{
    \begin{array}{@{}cl}
         e^{-N_1/\lambda} \frac{1}{k!} \left( \frac{N_1}{\lambda} \right)^k & \text{if } n = k \lambda \  \text{for } k \in \N,  \\
         0 & \text{otherwise},
    \end{array}
    \right.
\end{align}
that is a Poisson-like distribution supported on $\lambda \cdot \N$.  This fixes $\lambda$ to be an integer since we require $N$ to take value in $\N$.

This distribution is consistent with the abelian condition (\ref{eqn:fullAbelianCondition}). To sum up, we have seen that the abelian condition is quite powerful, at least assuming independence among the random variables. It constrains the distribution rules of $N$ and $\theta$, leaving only two parameters, $N_1$ and $\lambda$, that depend on the explicit theory in the bulk. We stress that $\lambda$ has to be a nonzero natural number to make sense of this result. We conclude this section by remarking that this solution is to restrictive for holographic principle to work at least in some generality, ruling out the possibility that $N$ and $\theta$ are independent random variables.

\subsection{Solving the abelian condition}
\label{sec:solving-abelian-condition}

We aim to solve the abelian condition in full generality, with the only assumption being the symmetry of the distribution with respect to the permutation of $\theta_i$.

Looking at the abelian condition, there is a clear hierarchy. Indeed, the equation corresponding to $\Sigma = \sqcup_{i=1}^k \Sigma_{g_i}$ depends on the distributions $f_N, f_{\theta | N}, \dots, f_{\theta_1, \dots, \theta_l | N}$ for $l$ equal to the number of Riemann surfaces of genus different from one in $\Sigma$. Thus, we can opt for an inductive approach.

The first step consists of considering $k$ copies of tori, for which the abelian condition gives the equation~\eqref{eqn:AbelianConditionOnTori}. Hence, suppose we know $x_k \coloneqq \mathbb{E}[N\theta^{1-k}]$ for any $k \ge 1$, then this allows us to compute the moment-generating function for the $N$-distribution:
\begin{equation}
    M_N(t)\coloneqq \mathbb{E}[e^{tN}] = \sum_{k=0}^\infty N_k \frac{t^k}{k!} = \sum_{k=0}^\infty  B_k(x_1, \dots, x_k)  \frac{t^k}{k!} = \exp \left[ \sum_{k=1}^\infty x_k \frac{t^k}{k!} \right],
\end{equation}
where we used some basic facts about Bell polynomials. This assumes that the $N$-distribution admits a moment-generating function, which is not necessarily true. For example, the higher moments might be divergent. This is indeed what can happen in our setting, as we will point out later. Furthermore, we stress that from this expression it is not apparent that $M_N(t)$ corresponds to a discrete distribution, as it will turn out later.

For the time being, we will not worry about these problems and assume that the moment-generating function exists. Hence, we are also supposing that the power series in the argument of the exponential is convergent in a sufficiently small disk centred at $t=0$.

Next, we shall consider surfaces of the form $\Sigma = \Sigma_g \sqcup (T^2)^{\sqcup k}$.
Directly from the definition, we can check that:
\begin{equation}
\avg{\mathrm{T}_\xi( \Sigma_g \sqcup  (T^2)^{\sqcup k} )}_\xi = \mathbb{E}[ N^{k+1} \, \theta^{1-g}].
\end{equation}

On the other hand, performing some combinatorics:
\begin{equation}
\avg{\mathrm{T}_\xi( \Sigma_g \sqcup  (T^2)^{\sqcup k} )}_\xi = \sum_{n=0}^k \binom k n \avg{\mathrm{T}_\xi(   (T^2)^{\sqcup k-n} )}_\xi \,\avg{\mathrm{T}_\xi( \Sigma_{g+n}  )}_\xi,
\end{equation}
so that:
\begin{equation}
 \mathbb{E}[ N^{k+1} \theta^{1-g}] = \sum_{n=0}^k \binom k n  \mathbb{E}[N^{k-n}] \,   \mathbb{E}[ N \theta^{1-g-n}] .
\end{equation}

Now, define $\Phi(t,r) \coloneqq \mathbb{E}[N \theta \exp (t N + r/\theta)]$. 
It is straightforward to obtain the moment-generating function for $N$ and $\theta^{-1}$ in terms of $\Phi$:
\begin{equation}
M_{N,\theta^{-1}}(t,r)\equiv \mathbb{E}[\exp(tN+\theta/r)] = \int_{-\infty}^t dt \, \frac{d}{dr} \Phi(t,r) .
\label{M-Phi-relation}
\end{equation}

Hence, it is enough to compute $\Phi(t,r)$. We start by noticing that after some manipulations, we may write:
\begin{equation}
\Phi(t,r) = M_N(t) \, ( \mathbb{E}[ N \theta] + \log M_N(t+r) ).
\end{equation}
All the information is now encoded in the moment-generating function for $N$. Finally, plugging this into the right-hand side of (\ref{M-Phi-relation}) yields:
\begin{equation}
M_{N,\theta^{-1}}(t,r) = \int^t_{-\infty} dt \,\frac{d}{dr} M_N(t) \, ( \mathbb{E}[ N \theta] + \log M_N(t+r) ) =  \int^t_{-\infty} dt \, \frac{M_N(t)}{M_N(t+r)}  \frac{d}{dr}  M_N(t+r) .
\end{equation}

\subsubsection{Discrete dimensions}

Before continuing, let us deal with the problem of imposing a discrete distribution supported on integers for the dimension. Imposing such a condition at the level of the moment-generating function amounts to requiring it to be $2\pi i$ periodic. However, determining whether a power series is of a periodic function can be quite challenging. We avoid this problem by considering the following sufficient condition, which will be satisfied by a large class of examples. Namely, we require the $x_k$ to be a Dirichlet series:
\begin{equation}
\label{eqn:DirichletSeriesHypothesis}
x_k = \sum_{n=1}^\infty b_n n^k,
\end{equation}
with non-negative coefficients $b_n \ge 0$. We remark that negative coefficients would lead to negative probabilities. 

We also require these series to be convergent for any $k \le k_* \in \N$, possibly with $k_* = + \infty$. We remark that a finite value of $k_*$ means that our distribution only admits a finite number of moments. This is not a problem a priori, and indeed happens if the bound (\ref{measure-sufficient-bound}) is not satisfied for arbitrary $s$ (as in the simple example of the measure considered in Section \ref{sec:measure-example}). However, some of the infinite series below should be understood formally, unless $k_* = + \infty$. Nevertheless, even if $k_*<\infty$, the resulting distribution itself is well-defined, and reproduces the correct $x_k$ for $k\leq k_*$.

We now check that this condition implies that $N$ is supported on the integers. Indeed, we can plug this expression into the moment-generating function and, assuming that the series is absolutely convergent in an open neighborhood of $t=0$, we obtain:
\begin{equation}
    M_N(t) = e^{-x_0} \exp \left[ \sum_{k=0}^\infty \sum_{n=1}^\infty b_n \frac{(tn)^k}{k!} \right] = e^{-x_0}  \exp \left[ \sum_{n=1}^\infty b_n e^{tn} \right].
\end{equation}
Thus, the moment-generating function is manifestly $2\pi i$ periodic as we want.

Under this additional hypothesis, we can explicitly determine the probability density functions:
\begin{align}
f_N(n) &= e^{-x_0} c_n, \\
f_{N,\theta^{-1}}(n,m) &= \frac{e^{-x_0}}{n} m \,b_m c_{n-m},
\label{first-two-fs}
\end{align}
where we have defined the $c_k$'s to be the coefficients in the series expansion:
\begin{equation}
    \exp \left[ \sum_{k=1}^\infty b_k x^k \right] = \sum_{k=0}^\infty c_k x^k.
\end{equation}
Notice that it is intended $c_k = b_k = 0$ for $k < 0$. Moreover, for $k\geq 0$, $c_k$ can be expressed as a polynomial in $b_{i},\; i=1,\ldots,k$. The obtained distribution for $\theta^{-1}$ turns out to be discrete and supported on positive integers, so $f_{N,\theta^{-1}}$ in (\ref{first-two-fs}) means the distribution function where $\theta^{-1}$ is considered as a \textit{discrete} random variable, although a priori it was a continuous one\footnote{Explicitly, in terms of the previously considered continuous distributions, we have:
\begin{equation}
    f_{N,\theta^{-1}}(n,m)\,\delta(\alpha_1-1/m) \equiv f_N(n) \,f_{\theta_1|N}(\alpha_1|n)
    = f_N(n)\int d\alpha_2\ldots d\alpha_N\,f_{\theta_1,\ldots,\theta_N|N}(\alpha_1,\alpha_2,\ldots,\alpha_N|n).
\end{equation}
}. Similarly, the distributions for multiple $\theta_i^{-1}$ will be considered as distributions of discrete variables from now on.

\subsubsection{The full solution}

One could repeat the above procedure to solve the abelian condition with $\Sigma = \Sigma_g \sqcup \Sigma_{g'} \sqcup  (T^2)^{\sqcup k}$ and so on. By doing so, we infer the general solution.  

\begin{proposition}
\label{prop:dist-solution}
Consider the following probability density function:
\begin{equation}
f_{N,\theta^{-1}_1, \dots, \theta^{-1}_k}(n, m_1, \dots, m_k) = e^{-x_0} \frac{(n-k)!}{n!} \sum_{p \in P(\{ 1, \dots k\}) } c_{n-\sum_{i=1}^l m_{p_i}} \prod_{i=1}^l \delta(p_i) \,b_{m_{p_i}} \frac{m_{p_i}!}{(m_{p_i}- |p_i| )!},
\end{equation}
where $p = \{ p_1, \dots, p_l \}$ is a partition of the set $\{ 1, \dots, k\}$ of length $l$ and:
\begin{equation}
\delta(\{\mu_1, \dots, \mu_r\}) \equiv \prod_{\alpha=1}^{r-1} \delta_{m_{\mu_\alpha}, m_{\mu_{\alpha+1}}}, \qquad m_{\{\mu_1, \dots, \mu_r\}} \equiv m_{\mu_1} = \dots = m_{\mu_r}.
\end{equation}

Then, such a function satisfies:
\begin{equation}
\sum_{m_{l+1}} \dots \sum_{m_k} f_{N,\theta^{-1}_1, \dots, \theta^{-1}_k}(n, m_1, \dots, m_k) =  f_{N,\theta^{-1}_1, \dots, \theta^{-1}_l}(n, m_1, \dots, m_l),
\end{equation}
and for $k=0$ and $k=1$ it restricts to $f_N$ and $f_{N,\theta^{-1}}$, as above. 

Moreover, the distribution corresponding to $f_{N,\theta^{-1}_1, \dots, \theta^{-1}_k}$ satisfies the abelian condition for disjoint unions of up to $k$ Riemann surfaces (that are not tori) with an arbitrary amount of tori.
\end{proposition}

We postpone the proof of this proposition to the appendix~\ref{sec:proofProp2}, as it is more technical than conceptual.

\subsection{Some remarks on the distribution}

At this point, we have a solution to the abelian condition that satisfies certain well-motivated conditions, in particular, the dimension being a natural number. In the remainder of the section, we will point out some properties of the ensemble of 2d TQFTs that corresponds to such distributions. First of all, the random variables specifying the ensemble are not independent in general, as can be seen from the explicit formulas for the distribution.

As was already pointed out above, the distributions for $\theta$ are also discrete. Their values are supported on $1/k$ for $k=1,\dots, n$, where $n$ is the value assumed by the random variable $N$. More conditions on the allowed values for $\theta$ can be deduced from the expression for the conditional probability.

\subsubsection{Expression for the conditional probability}

Consider the conditional probability, given $N = n$ for some $n \in \mathbb{N}$:
\begin{equation}
 f_{\theta_1^{-1}, \dots, \theta_n^{-1} | N}(m_1,\dots, m_n | n) = \frac{1}{n!\, c_n} \sum_{p \in P(\{ 1, \dots n\}) } c_{n-\sum_{q \in p}^l m_{q}} \prod_{q \in p}^l \delta(q) \, b_{m_{q}} \frac{m_{q}!}{(m_{q}- |q| )!}.
\end{equation}

We will now argue that this expression simplifies significantly. First of all, we show that the only partitions with a non-vanishing contribution have to satisfy:
\begin{equation}
    n = \sum_{q \in p} m_q.
\end{equation}
In particular, we show that if  $\sum_{q \in p} m_q \ge n + 1$ or $\sum_{q \in p} m_q \le n - 1$, then the contribution is zero. The first case is straightforward since the argument of $c_n$ has to be positive. In the second case, notice that $n = \sum_{q \in p} |q| $, hence $\sum_{q \in p} (m_q - |q|)$ is negative. This implies that there exists a $q \in p$ such that $m_q - |q|$ is negative and therefore:
\begin{equation}
\frac{m_q!}{(m_q - |q|)!} = 0.
\end{equation}

Now, the presence of the factor:
\begin{equation}
\delta(q) \,\frac{m_q!}{(m_q - |q|)!}
\end{equation}
implies that a specific partition contributes zero unless it selects $|q|$ equal entries $m_q \equiv m_{q_1} = \dots = m_{q_{|q|}}$. However, if $|q| \ge m_q + 1$, then the contribution is again zero due to the factorial in the denominator. 

Moreover, we now show that $|q|$ needs to be exactly equal to $m_q$. Indeed, if this was not the case, then there would exist a $\bar{q}\in p$ such that $|\bar{q}| \le m_{\bar{q}} - 1$. On the other hand, the others should satisfy $|q| \le m_{q}$. Hence:
\begin{equation}
n = \sum_{q \in p} |q| \le \left(\sum_{q \in p} m_q \right) - 1 = n - 1,
\end{equation}
yielding a contradiction. 

Thus, the only non-vanishing combinations of $\vec{m} =( m_1, \dots, m_n)$ are of the following form:
\begin{equation}
\vec{m} = (\,\overbrace{\underbrace{m_1,\dots, m_1}_{a_1 \cdot m_1 \ \text{times}} , \underbrace{m_2,\dots, m_2}_{a_2 \cdot m_2 \ \text{times}}, \dots, \underbrace{m_r,\dots, m_r}_{a_r \cdot m_r \ \text{times}}}^{n \ \text{times}}\,),
\end{equation}
and permutations thereof, with probability:
\begin{align}
\notag
 f_{\theta_1^{-1}, \dots, \theta_n^{-1} | N}(\vec{m}| n) &= \frac{1}{n! \,c_n} \sum_{p} \prod_{i=1}^r (m_i!  \,b_{m_i})^{a_i} \\
 \notag
 & = \frac{1}{n! \, c_n} \prod_{i=1}^r (m_i! \, b_{m_i})^{a_i} \frac{(a_i \,m_i)!}{(m_i!)^{a_i} \,a_i!} \\
 & =  \frac{1}{n! \,c_n} \prod_{i=1}^r  b_{m_i}^{a_i} \frac{(a_i\, m_i)!}{a_i!} .
\end{align}

Notice that, up to permutations, the possible choices of $\vec{m}$ --- i.e. those with non-zero probability --- are in a one-to-one correspondence with possible choices of positive integers $m_i$ and $a_i$ such that $n = \sum_{i=1}^r a_i m_i$ for some $r$.

\subsubsection{Explicit formulae for low dimensions}

We derive explicit expressions for the conditional probabilities, given $ n=1, 2, 3$. When $n=1$, the only non-zero probability is for $m=1$ and it is consistent with the normalization condition $f_{\theta^{-1}|N}(1|1) = 1$.

For $n=2$, only two configurations are non-vanishing:
\begin{equation}
    f_{\theta_1^{-1},\theta_2^{-1}|N}(1,1|2) =\frac{b_1^2}{b_1^2+2b_2},  \qquad
    f_{\theta_1^{-1},\theta_2^{-1}|N}(2,2|2) =\frac{2b_2}{b_1^2+2b_2}.
\end{equation}
While, in the case $n=3$, we have:
\begin{align}
    \notag
    &f_{\theta_1^{-1},\theta_2^{-1},\theta_3^{-1}|N}(1,1,1|3) =\frac{b_1^3}{b_1^3+6b_1b_2+6b_3} , \qquad
    f_{\theta_1^{-1},\theta_2^{-1},\theta_3^{-1}|N}(1,2,2|3) =\frac{2b_1b_2}{b_1^3+6b_1b_2+6b_3}, \\
    &f_{\theta_1^{-1},\theta_2^{-1},\theta_3^{-1}|N}(3,3,3|3) =\frac{6b_3}{b_1^3+6b_1b_2+6b_3} ,
\end{align}
and permutations thereof.

\subsection{Recovering the distribution for a given 3d TQFT}
\label{sec:holography}

Now, we will combine the results of this section with the analysis performed in Section \ref{sec:sum-over-topologies}. We have shown that under some mild hypotheses that we are going to verify, it is enough to compute $\mathbb{E}[N\theta^{1-k}]$. These fix uniquely the distributions for $N$ and $\theta$.

As already argued at the beginning of this section, one can express the sum over topologies via terms that only include connected manifolds. The case of $\Sigma = (T^2)^{\sqcup k}$ is simple enough to obtain a closed relation:
\begin{equation}
    \mathcal{ZG}((T^2)^{\sqcup k};\mathcal{T}) = e^\zeta B_k ( \mathcal{Z}_1, \dots, \mathcal{Z}_k ) .
\end{equation}

On the other hand, by the abelian condition~\eqref{eqn:AbelianConditionOnTori}, we have:
\begin{equation}
    \avg{\mathrm{T}_\xi((T^2)^{\sqcup k})}_\xi = B_k\left( \mathbb{E}[N],\mathbb{E}[N\theta^{-1}], \dots, \mathbb{E}[N\theta^{1-k}] \right).
\end{equation}
Thus, by the holographic principle, we obtain:
\begin{equation}
    x_g=\mathbb{E}[N\theta^{1-g}] = \mathcal{Z}_g\equiv \mathcal{Z}_g(D,q;\mathfrak{L};\mu).
\end{equation}
That is, the moments of the distribution are closely related to the connected gravitational partition function we have analyzed in Section \ref{sec:sum-over-topologies}. In particular, it was shown there that the hypothesis~\eqref{eqn:DirichletSeriesHypothesis} is satisfied under certain assumptions on the measure. The coefficients $b_n$ are explicitly given by the formula (\ref{Dirichlet-coefficients-general-explicit}). The reality of the coefficients is automatic, assuming that the measure is real $\mu$ and respects the orientation reversal of the manifolds,  that is $\mu(A,\bar{\ell})=\mu(A,{\ell})$ and $\nu(A,{\ell},S)=\nu(A,\bar{\ell},S)$, where $\bar{\ell}$ is the representative linking pairing equivalent to $-\ell$. The reality then follows from the unitarity of the partition function of the 3d TQFT (cf. \cite{deloup2001abelian}):
\begin{equation}
    Z(A,\bar{\ell};D,q;-\theta)=\overline{Z(A,{\ell};D,q;\theta)}.
\end{equation}

In the case of untwisted DW theory with Neumann boundary conditions, it is also immediate to see that $b_n>0$, assuming the measure is positive. From the expression (\ref{untwisted-DW-partition-function}) for the partition function we have:
\begin{equation}
    b_n=\sum_{(A,\ell)} \frac{\mu(A,\ell)}{|\mathfrak{L}\otimes A|} \,\left|\left\{\alpha\in \hom(A,\mathfrak{L}^*):\;\nu(A,\ell,\ker\alpha)|\mathfrak{L}|=n\right\}\right|
    >0.
\end{equation}

\section{Generalizations}
\label{sec:extension}
In this section, we discuss the dependence of the sum over 3-manifolds on the choice of the topological boundary conditions. We also evaluate the sum with the special measure considered in Section \ref{sec:measure-example} for a case of untwisted DW with a boundary condition which is neither Neumann nor Dirichlet, and also for a case of twisted DW theory.

\subsection{Dependence on the boundary condition}
\label{sec:bcDependence}

In the previous sections, we have considered the  sum over topologies of the partition function of an abelian TQFT $\mathcal{T}$ with a fixed topological boundary condition $\mathfrak{L}$, that is, the sum:
\begin{equation}
    \langle \mathfrak{L}|\Sigma;\mu\rangle_\mathcal{T}\coloneqq\sum_{[M,\phi]:\,\Sigma\xrightarrow[\phi]{\sim} \partial M} \mu([M,\phi])\,\langle \mathfrak{L}|M,\phi\rangle_\mathcal{T}\qquad \in \mathcal{H}_\mathcal{T}(\Sigma),
    \label{state-sum-over-topologies}
\end{equation}
which was given precise meaning by rewriting it as a sum over certain algebraic data. 

We will now discuss the dependence on the choice of the topological boundary condition. We will argue that, under some additional assumptions, one can recover the actual state given by the sum of the values of the TQFT on the bordisms:
\begin{equation}
|\Sigma;\mu\rangle_\mathcal{T}\coloneqq\sum_{[M,\phi]:\,\Sigma\xrightarrow[\phi]{\sim} \partial M} \mu([M,\phi])\,|M,\phi\rangle_\mathcal{T}\qquad \in \mathcal{H}_\mathcal{T}(\Sigma),
\end{equation}
from the knowledge of $\langle \mathfrak{L}|\Sigma,\mu\rangle_\mathcal{T}$ for all possible topological boundary conditions $\mathfrak{L}$. The knowledge of such a state in particular provides us the value of the sum over topologies $\langle \mathfrak{B}|\Sigma;\mu\rangle_\mathcal{T}$ with an arbitrary, possibly non-topological, boundary condition $\mathfrak{B}$.  

The additional assumption that we impose for this to be achievable can be formulated as follows. First, as it was already done in the algebraic setting in Section \ref{sec:sum-over-topologies}, we assume that the measure $\mu([M,\phi])$ in (\ref{state-sum-over-topologies}) factorizes with respect to disjoint union of 3-manifolds (up to the combinatorial factor -- one over the number of permutations of homeomorphic closed components). The total sum (\ref{state-sum-over-topologies}) then can be expressed through the sums over connected 3-manifolds:
\begin{equation}
|\Sigma;\mu\rangle_\mathcal{T}^\text{conn}\coloneqq\sum_{\substack{[M,\phi]:\,\Sigma\xrightarrow[\phi]{\sim} \partial M\\ \text{connected }M}} \mu([M,\phi])\,|M,\phi\rangle_\mathcal{T}\qquad \in \mathcal{H}_\mathcal{T}(\Sigma),
\label{state-sum-over-connected-topologies}
\end{equation}
so that we can focus on such sums only.
Let $\Sigma$, the boundary surface, decompose into its connected component as $\Sigma=\sqcup_i \Sigma_i$, and $\phi_i\coloneqq\phi|_{\Sigma_i}$ be the inclusion of the $i$-th component into the boundary of the connected 3-manifold $M$. The sum over the bordisms equivalence classes $[M,\phi]:\emptyset \rightarrow \Sigma$ in particular includes sums over $\phi_i$'s for a fixed diffeomorphism class $M$. Let $\psi_i$ be some self-diffeomorphisms $\psi_i:\Sigma_i\rightarrow\Sigma_i$ of the connected components of $\Sigma$. One can then consider a change of the summation of variables $\phi_i=\phi_i'\circ\psi_i$ in the sum above:
\begin{multline}
        |\Sigma;\mu\rangle_\mathcal{T}^\text{conn}=\sum_{\substack{[M,\phi']:\,\Sigma\xrightarrow[\phi']{\sim} \partial M\\ \text{connected }M}} \mu([M,\{\phi_i'\circ \psi_i\}_i])\,|M,\{\phi_i'\circ \psi_i\}_i\rangle_\mathcal{T}
        =
\\
=
\bigotimes_i\mathcal{T}(\psi_i)
\sum_{\substack{[M,\phi']:\,\Sigma\xrightarrow[\phi']{\cong} \partial M\\ \text{connected }M}} \mu([M,\{\phi_i'\circ \psi_i\}_i])\,|M,\{\phi_i'\}_i\rangle_\mathcal{T}
        \qquad \in \bigotimes_i\mathcal{H}_\mathcal{T}(\Sigma_i),
\end{multline}
where $\mathcal{T}(\psi_i):\mathcal{H}_\mathcal{T}(\Sigma_i)\rightarrow \mathcal{H}_\mathcal{T}(\Sigma_i)$ is the representation of the action of the mapping class group of $\Sigma_i$ on the corresponding Hilbert space. If we further assume that the measure is invariant under the action of the mapping group on the boundary, that is:
\begin{equation}
    \mu([M,\{\phi_i'\circ \psi_i\}_i])=
    \mu([M,\{\phi_i'\}_i]),
\end{equation}
we obtain that the total state must be invariant under the action of the mapping class group on the Hilbert space:
\begin{equation}
    |\Sigma;\mu\rangle_\mathcal{T}^\text{conn}=\left(\bigotimes_i\mathcal{T}(\psi_i)\right) |\Sigma;\mu\rangle_\mathcal{T}^\text{conn}  \qquad \in \bigotimes_i\mathcal{H}_\mathcal{T}(\Sigma_i).
\end{equation}
It implies that the state must be of the following form:
\begin{equation}
   |\Sigma;\mu\rangle_\mathcal{T}^\text{conn}=\sum_{\{\text{Lagrangian }\mathfrak{L}_i\}_i} \left(\bigotimes_i|\mathfrak{L_i}\rangle\right) \,\mathcal{P}_{\{\mathfrak{L}_i\}_i}^\text{conn}(\Sigma)\qquad \in \bigotimes_i\mathcal{H}_\mathcal{T}(\Sigma_i).
\end{equation}
where the sum happens over all Lagrangian subgroups $\mathfrak{L}_i$ of the 3d TQFT discriminant group for each connected component $\Sigma_i$. Similarly, the total sum, which includes possible disconnected 3-manifolds, can be written as:
\begin{equation}
    |\Sigma;\mu\rangle_\mathcal{T}=\sum_{\{\text{Lagrangian }\mathfrak{L}_i\}_i} \left(\bigotimes_i|\mathfrak{L_i}\rangle\right) \,\mathcal{P}_{\{\mathfrak{L}_i\}_i}(\Sigma)\,e^{\zeta}\qquad \in \bigotimes_i\mathcal{H}_\mathcal{T}(\Sigma_i).
    \label{cosmological-state-lagrangian-decomposition}
\end{equation}
where, as before, $\zeta\equiv \mathcal{Z}_0$ is the sum over connected closed 3-manifolds. The coefficients $\mathcal{P}$ can be expressed as a polynomial in $\mathcal{P}^\text{conn}$. For example, when $\Sigma=\Sigma_1\sqcup \Sigma_2$ has two connected components, we have:
\begin{equation}
\mathcal{P}_{\mathfrak{L}_1,\mathfrak{L}_2}(\Sigma_1\sqcup\Sigma_2)=\mathcal{P}_{\mathfrak{L}_1,\mathfrak{L}_2}^\text{conn}(\Sigma_1\sqcup\Sigma_2)+\mathcal{P}_{\mathfrak{L}_1}^\text{conn}(\Sigma_1) \,\mathcal{P}_{\mathfrak{L}_2}^\text{conn}(\Sigma_2).   
\end{equation}

The next assumption that we make is less trivial, since we require that the final state in the sum over connected topologies only receives a nonzero contribution if all connected components share the same Lagrangian subgroup, that is:
\begin{equation}
    \mathcal{P}_{\{\mathfrak{L}_i\}_i}^\text{conn}(\Sigma)=\mathcal{P}^\text{conn}_{\mathfrak{L}_1}(\Sigma)\,\delta({\mathfrak{L}_1=\mathfrak{L}_2=\ldots)},
    \label{connected-delta-lagrangian}
\end{equation}
from which we deduce:
\begin{equation}
   |\Sigma;\mu\rangle_\mathcal{T}^\text{conn}=\sum_{\text{Lagrangian }\mathfrak{L}} \left(\bigotimes_i|\mathfrak{L}\rangle^{(i)}\right) \,\mathcal{P}_{\mathfrak{L}}^\text{conn}(\Sigma)\qquad \in \bigotimes_i\mathcal{H}_\mathcal{T}(\Sigma_i),
   \label{same-lagrangian-condition}
\end{equation}
where the sum is now over a single Lagrangian subgroup $\mathfrak{L}$, with $|\mathfrak{L}\rangle^{(i)}$ being the corresponding state in $\mathcal{H}_\mathcal{T}(\Sigma_i)$.

This assumption can be motivated as follows. As was discussed in Section \ref{sec:abelian-tqfts}, an abelian TQFT gives the same values on 3-manifolds that are related by drilling tubes between the boundary components. In particular, we have isomorphisms between the Hilbert spaces associated to the disjoint union of surfaces $\Sigma_i$ and their connected sum:
\begin{equation}
   \mathcal{H}_\mathcal{T}(\Sigma)\equiv  \bigotimes_i \mathcal{H}_\mathcal{T}(\Sigma_i)\cong\mathcal{H}_\mathcal{T}(\Sigma_1\#\Sigma_2\#\ldots).
   \label{iso-connected-sum}
\end{equation}
Under such isomorphism, the state $|\mathfrak{L}\rangle\in \mathcal{H}_\mathcal{T}(\Sigma_1\#\Sigma_2\#\ldots)$ is identified with $\bigotimes_i|\mathfrak{L}\rangle^{(i)} \in  \mathcal{H}_\mathcal{T}(\Sigma)$.

Moreover, we have identifications of the TQFT values $|\tilde{M};\tilde{\phi}\rangle\in \mathcal{H}_\mathcal{T}(\Sigma_1\#\Sigma_2\#\ldots)$ and $|M;\phi\rangle\in \mathcal{H}_\mathcal{T}(\Sigma)$, where $(\tilde{M},\tilde{\phi})$ is obtained from $(M,\phi)$ by drilling a thin tube between $\Sigma_i$ and $\Sigma_{i+1}$ for each $i$ -- so that $\tilde{\phi}:\Sigma_1\#\Sigma_2\#\ldots \xrightarrow{\cong}\partial M'$ and $\tilde{\phi}|_{\Sigma_{i}}=\phi|_{\Sigma_i}$. The assumption (\ref{same-lagrangian-condition}) can then be interpreted as a consequence of the condition that the summation measure also respects the operations of filling and drilling tubes between the boundary components inside a connected 3-manifold. Since we have not rigorously defined the sum over connected topologies (\ref{state-sum-over-connected-topologies}), but only its algebraic counterpart (which assumes the same topological boundary condition on all components), we will omit the explicit formulation of this condition on the measure $\mu([M,\phi])$. 

With this assumption the full state can be recovered from the sums $\mathcal{Z}_{g;\mathfrak{L}}$ that were considered in Section \ref{sec:sum-over-topologies} (in the subscript we now include explicit dependence on the choice of the topological boundary condition $\mathfrak{L}$). Namely, we have:
\begin{multline}
    \mathcal{Z}_{g;\mathfrak{L}'}\equiv 
\langle \mathfrak{L}'|\Sigma;\mu\rangle_\mathcal{T}^\text{conn}=\sum_{\text{Lagrangian }\mathfrak{L}'} \langle \mathfrak{L}'|\left(\bigotimes_i|\mathfrak{L}\rangle^{(i)}\right)\,\mathcal{P}_{\mathfrak{L}}^\text{conn}(\Sigma)=
\\
=
\sum_{\text{Lagrangian }\mathfrak{L}'}|\mathfrak{L}'\cap \mathfrak{L}|^g \,\mathcal{P}_{\mathfrak{L}}^\text{conn}(\Sigma),
\end{multline}
where $g=g_1+g_2+\ldots$ is the total genus of $\Sigma=\Sigma_1\sqcup \Sigma_2\sqcup\ldots$ We then have:
\begin{equation}
   \mathcal{P}^\text{conn}_{\mathfrak{L}}(\Sigma)=\sum_{\text{Lagrangian }\mathfrak{L}'}
   \Pi_{g;\mathfrak{L},\mathfrak{L}'}\,\mathcal{Z}_{g;\mathfrak{L}'},
   \label{cosmological-state-sum-over-lagrangians}
\end{equation}
where $\Pi_{g;\mathfrak{L},\mathfrak{L}'}$ are elements of the matrix $\Pi_{g}$ which is inverse to the matrix with components $|\mathfrak{L}'\cap \mathfrak{L}|^g$. Note that when $\Sigma$ itself is connected, the operator:
\begin{equation}
    \sum_{\mathfrak{L},\mathfrak{L}'}|\mathfrak{L}\rangle\Pi_{g;\mathfrak{L},\mathfrak{L}'}\langle \mathfrak{L}'|
\end{equation}
plays the role of the projector on the subspace of states invariant under the mapping class group action. In the case of a connected $\Sigma$, the formula (\ref{cosmological-state-sum-over-lagrangians}) then follows just from invariance of the measure under self-diffeomorphisms of $\Sigma$. 

We will now consider the implications of the above assumptions to the holographic correspondence. 
In Sections \ref{sec:general-distribution} and \ref{sec:holography}, we considered a holographically dual distribution of boundary 2d topological theories $\mathrm{T}_{\xi;\mathfrak{L}}$, where $\xi$ denotes the distribution variables and $\mathfrak{L}$ now explicitly indicates the choice of the topological boundary conditions. The holographic correspondence worked in the following way:
\begin{equation}
    e^{-\zeta}\,\langle \mathfrak{L}|\Sigma;\mu\rangle_\mathcal{T}=\langle \mathrm{T}_{\xi;\mathfrak{L}}(\Sigma)\rangle_\xi,
    \label{holographic-correspondence-bc-dependence}
\end{equation}
where $\mathrm{T}_{\xi;\mathfrak{L}}(\Sigma)$ is the partition function of the 2d TQFT on $\Sigma$, $\langle \ldots\rangle_\xi$ denotes the ensemble average, and $\zeta\equiv \mathcal{Z}_0$ is the sum over connected closed 3-manifolds.

Suppose the 2d TQFTs in the ensemble are of the following form\footnote{Recall the definition of the direct sum of $d$-dimensional TQFTs: it is the TQFT for which the Hilbert space associated to a \textit{connected} $(d-1)$-dimensional manifold is the direct sum of the Hilbert spaces of the summands, and the value on a \textit{connected} $d$-manifold is the sum of the values of the summands.}: 
\begin{equation}
    \mathrm{T}_{\xi;\mathfrak{L}}=\bigoplus_{\mathfrak{L}'}\mathrm{T}_{\mathfrak{L},\mathfrak{L}'}\otimes \hat{\mathrm{T}}_{\xi;\mathfrak{L}'},
    \label{2d-TQFT-factorization}
\end{equation}
where $\mathrm{T}_{\mathfrak{L},\mathfrak{L}'}$ is the TQFT obtained by the reduction of the 3d TQFT $\mathcal{T}$ on the interval with topological boundary conditions $\mathfrak{L}$ and $\mathfrak{L}'$ at the ends. More explicitly, its partition function on a surface $\Sigma$ of total genus $g$ is given by $\mathrm{T}_{\mathfrak{L},\mathfrak{L}'}(\Sigma)=\langle \mathfrak{L}|\mathfrak{L}'\rangle=|\mathfrak{L}\cap\mathfrak{L}'|^g$.

Plugging (\ref{2d-TQFT-factorization}) and (\ref{cosmological-state-lagrangian-decomposition}) into the left and right-hand side of (\ref{holographic-correspondence-bc-dependence}) respectively, we obtain:
\begin{equation}
    \mathcal{P}_{\{\mathfrak{L}_i\}_i}\left(\sqcup_i\Sigma_i\right)=\left\langle\textstyle\prod_i \hat{\mathrm{T}}_{\mathfrak{L}_i}(\Sigma_i)\right\rangle_\xi .
\end{equation}
The condition (\ref{connected-delta-lagrangian}) that the contribution to $\mathcal{P}_{\{\mathfrak{L}_i\}}(\Sigma)$ from connected topologies vanishes unless all $\mathfrak{L}_i$ are the same then can be interpreted as the condition that the distributions of the 2d TQFTs $\hat{\mathrm{T}}_{\xi;\mathfrak{L}}$ are independent for different $\mathfrak{L}$.

Moreover, the correlators of the partition functions of $\hat{\mathrm{T}}_{\xi;\mathfrak{L}}$ must satisfy the same abelian condition property as the one considered in Section \ref{sec:abelian-condition} for the correlators of the partition functions of  ${\mathrm{T}}_{\xi;\mathfrak{L}}$. Thus, the distribution of $\hat{\mathrm{T}}_{\xi;\mathfrak{L}}$ can be recovered from the knowledge of the expectation values $\langle \hat{\mathrm{T}}_{\xi;\mathfrak{L}}(\Sigma_g) \rangle_\xi$ for connected $\Sigma_g$ in the same way the distribution of ${\mathrm{T}}_{\xi;\mathfrak{L}}$ was recovered from $\langle \mathrm{T}_{\xi;\mathfrak{L}}(\Sigma_g) \rangle_\xi= \mathcal{Z}_g=x_g$ in Section \ref{sec:solving-abelian-condition}. However, unlike the previously obtained distribution of $\mathrm{T}_{\xi;\mathfrak{L}}$, the distribution of $\hat{\mathrm{T}}_{\xi;\mathfrak{L}}$ will be generically only formal: it will be non-normalizable and will lack positivity. This is because:
\begin{equation}
  \langle \hat{\mathrm{T}}_{\xi;\mathfrak{L}}(\Sigma_g) \rangle_\xi=\sum_{\text{Lagrangian }\mathfrak{L}'}
   \Pi_{g;\mathfrak{L},\mathfrak{L}'}\,\mathcal{Z}_{g;\mathfrak{L}'},
\end{equation}
and the matrix $\Pi_{g;\mathfrak{L},\mathfrak{L}'}$, which is the inverse to the matrix with the components $|\mathfrak{L}\cap\mathfrak{L}'|^g$, has negative components and is singular when $g=0$.

\subsection{An explicit example of general topological boundary conditions}

In the previous sections, we have focused on the most natural choice for the topological boundary conditions in untwisted Dijkgraaf-Witten theories with gauge group $G$. Namely, we have adopted the Neumann boundary condition, which is equivalent to $\mathfrak{L}_{\text{Neu}} = G \subset G \oplus G^*$. Another standard choice is the Dirichlet boundary condition, corresponding to $\mathfrak{L}_{\text{Dir}} = G^*$. Notice that untwisted DW theories always admit both boundary conditions, though these do not in general exhaust all possible choices. Indeed, depending on $G$, there may be more than just two Lagrangian subgroups. In particular, in the case of $G = \Z / {p^r}\Z$, the 3d theory admits a family of Lagrangian subgroups:
\begin{equation}
    \mathfrak{{L}}_k = \{(cp^{r-k},dp^k) \mid c \in \Z /{p^k}\Z, d\in \Z/{p^{r-k}} \Z\} \simeq \Z/{p^k}\Z \oplus (\Z/{p^{r-k}}\Z)^* \subset G \oplus G^*.
\end{equation}
Notice that the Neumann and Dirichlet boundary conditions are included in this family as extremal cases: $\mathfrak{L}_{\text{Neu}} \equiv \mathfrak{L}_r$ and $\mathfrak{L}_{\text{Dir}} \equiv \mathfrak{L}_0$. 

We plan to demonstrate how the procedure described in Section \ref{sec:bcDependence} can be carried over in a simple example. In particular, we compute the connected gravitational partition function, given any topological boundary condition, for a specific choice of the gauge group $G$.

Let $p$ be an odd prime and consider an untwisted Dijkgraaf-Witten theory with gauge group $G = \Z/{p^2}\Z$ and topological boundary condition given by $\mathfrak{L} \simeq \Z/p\Z \oplus \Z /p\Z \subset G \oplus G^*$. It is straightforward to repeat the computations of Section~\ref{sec:cyclicDW} with boundary state $c=(c_1p,c_2p) \in \mathfrak{L}^g$, for any $c_1,c_2 \in (\Z/p\Z)^g$.

As discussed in Section~\ref{sec:finiteForms}, it is enough to compute the partition function for the bilinear forms $X_{q^a}$, $Y_{q^a}$, $A_{2^a}$, $B_{2^a}$, $C_{2^a}$, $D_{2^a}$, $E_{2^a}$ and $F_{2^a}$.

Since $p$ is an odd prime, the partition function on a 3-manifold with the first homology group being 2-torsion is trivial:
\begin{equation}
    {Z}(\Z / 2^a \Z,A_{2^a}, f; D, q; c) = 1 / |G|,
\end{equation}
and the same holds by replacing $A_{2^a}$ with any other bilinear form on $\Z/{2^a}\Z$ or $(\Z/{2^a}\Z)^2$. A similar reasoning shows:
\begin{equation}
    Z(\Z/{q^a}\Z,X_{q^a}, f; D, q; c) = {Z}(\Z/{q^a}\Z,Y_{q^a}, f; D, q; c) = 1 /|G|,
\end{equation}
whenever $q \ne p$. 

We are now left with the case $q = p$. For the quadratic forms on $\Z/{p^a}\Z$ with $a \ge 2$, we get:
\begin{equation}
    Z(\Z/{p^a}\Z,X_{p^a}, f; D, q; c) = {Z}(\Z/{p^a}\Z, Y_{p^a}, f; D, q; c) = \frac{|G \otimes A_p|}{|G|} \, \delta( \theta_1) \delta( \theta_2 ),
\end{equation}
where $\theta = \sum_{i=1}^g c_i \otimes f_i \in \mathfrak{L}\otimes \Z/{p^a}\Z$ splits as $\theta = \theta_1 \oplus\theta_2$. On the other hand, when $a=1$ we obtain:
\begin{align}
    &{Z}(\Z/{p}\Z, X_{p}, f; D, q; c) = \frac{|G \otimes \Z/{p}\Z|}{|G|}  \exp\left[ 2 \pi i  X_p(\theta_1,\theta_2) \right],
     \\
    &{Z}(\Z/{p}\Z, Y_{p}, f; D, q; c) = \frac{|G \otimes \Z/{p}\Z|}{|G|} \exp\left[ 2 \pi i  Y_p(\theta_1,\theta_2) \right],
\end{align}
where it is understood that $\theta_i \in \Z/p\Z \cong \Z /p \Z \otimes \Z /p\Z$. 

For later convenience, we recombine these results in the partition function of the generic pair $(A_p,\ell_p) = ((\Z/p\Z)^l\oplus \tilde{A}_p, xX_p\oplus yY_p \oplus \tilde{\ell}_p)$, where $\tilde{A}_p$ consists of copies of $\Z/p^a\Z$ for $a>1$. Accordingly decomposing $\theta \oplus \tilde{\theta} \in \mathfrak{L}\otimes ((\Z/p\Z)^l\oplus \tilde{A}_p) $, we have:
\begin{equation}
    \label{eqn:weird-boundary-dw-partition-function}
    Z(A_p,\ell_p;D,q;\theta \oplus \tilde{\theta}) = \frac{|G \otimes A_p |}{|G|}  \,\delta(\tilde{\theta}_1) \delta(\tilde{\theta}_2) \exp\left[2\pi i (xX_p\oplus yY_p)(\theta_1,\theta_2) \right].
\end{equation}

Obtaining the gravitational partition function for this choice of topological boundary conditions is now straightforward, assuming the measure~\eqref{eqn:explicitMeasure}. Let us consider the factor corresponding to the particular prime $p$ in the expression~\eqref{sum-over-topologies-specific-measure} and plug into the TQFT partition function~\eqref{eqn:weird-boundary-dw-partition-function}, which results in:
\begin{multline}
    \sum_{\substack{(A_p,\ell_p):\\\text{finite abelian}\\p\text{-groups}\\\text{with pairings}}}|A_p|^{-s}  \frac{|\mathfrak{L}|^{g-1}\,}{|\mathfrak{L}\otimes {A_p}|}\, |G \otimes A_p | \times \\
     \sum_{\theta\in \mathfrak{L}\otimes (\Z/p\Z)^l}\sum_{\substack{\alpha\in (\mathfrak{L}\otimes {A}_p)^*\\\cong \hom({A}_p,\mathfrak{L}^*)}}|\ker{\alpha}|^{g} \exp \left[ 2\pi i\alpha\left(\theta \oplus 0 \right)+ 2\pi i (xX_p\oplus yY_p)(\theta_1,\theta_2) \right].
\end{multline}

Note that we can understand $\alpha$ as a $2\times b$ matrix over the finite field $\Z /p\Z$, where $b \equiv \text{dim}_{\Z/p\Z}(A_p \otimes \Z/p\Z)$ and the rows are given by $\alpha_i \in (\Z/p\Z)^b $ in $\alpha = \alpha_1 \oplus \alpha_2$. Then, $|\ker{\alpha}| = |A_p|, |A_p| p^{-1},$ or $|A_p|p^{-2}$ when $\alpha$ is the zero matrix, the rows of $\alpha$ are proportional one to each other, or the rows are linearly independent over $\Z/{p}\Z$ respectively. Then, we can rewrite the sum as:
\begin{align}
    & \notag
    \begin{aligned}
    (1-p^{-2g})\sum_{(A_p,\ell_p)} &|A_p|^{g-s} 
     \frac{|\mathfrak{L}|^{g-1}\,}{|\mathfrak{L}\otimes {A_p}|}\, |G \otimes A_p| \times \\
     &\sum_{\theta\in \mathfrak{L}\otimes (\Z/p\Z)^l} \exp \left[  2\pi i (xX_p\oplus yY_p)(\theta_1,\theta_2) \right]
     +
     \end{aligned} \\
     & \notag \begin{aligned}
     (p^{-g}-p^{-2g})
     \sum_{(A_p,\ell_p)}&|A_p|^{g-s} 
      \frac{|\mathfrak{L}|^{g-1}\,}{|\mathfrak{L}\otimes {A_p}|}\, |G \otimes A_p | \times \\ & \sum_{\theta\in \mathfrak{L}\otimes (\Z/p\Z)^l} \exp [2\pi i (xX_p\oplus yY_p)(\theta_1,\theta_2)] \sum_{\substack{\alpha\in (\mathfrak{L}\otimes {A}_p)^* \ne0\\\alpha_1 \varpropto \alpha_2}}\exp \left[ 2\pi i\alpha\left(\theta \oplus 0 \right) \right] +
      \end{aligned} \\
      & \begin{aligned}
     p^{-2g}
     \sum_{(A_p,\ell_p)}&|A_p|^{g-s} 
      \frac{|\mathfrak{L}|^{g-1}\,}{|\mathfrak{L}\otimes {A_p}|}\, |G \otimes A_p | \times \\ & \sum_{\theta\in \mathfrak{L}\otimes (\Z/p\Z)^l} \exp [2\pi i (xX_p\oplus yY_p)(\theta_1,\theta_2)] \sum_{\alpha\in (\mathfrak{L}\otimes {A}_p)^* }\exp \left[ 2\pi i\alpha\left(\theta \oplus 0 \right) \right].
      \end{aligned}
\end{align}
Now, we proceed to compute separately each summand. First, we easily evaluate the Gauss sum:
\begin{equation}
    \sum_{\theta\in \mathfrak{L}\otimes (\Z/p\Z)^l} \exp [2\pi i (xX_p\oplus yY_p)(\theta_1,\theta_2)]  = p^l.
\end{equation}
Then, we are left with the term corresponding to the sum over proportional rows in $\alpha$:
\begin{align}
     & \notag \begin{aligned}
     \sum_{\theta\in \mathfrak{L}\otimes (\Z/p\Z)^l} \exp [2\pi i (xX_p\oplus yY_p)(\theta_1,\theta_2)] \sum_{\substack{\alpha\in (\mathfrak{L}\otimes {A}_p)^* \ne0\\\alpha_1 \varpropto \alpha_2}}\exp \left[ 2\pi i\alpha\left(\theta \oplus 0\right) \right] =
      \end{aligned} \\
      & \notag \quad \begin{aligned}
     &2  \sum_{\theta\in \mathfrak{L}\otimes (\Z/p\Z)^l} \exp [2\pi i (xX_p\oplus yY_p)(\theta_1,\theta_2)]\sum_{\alpha \in (\Z/p\Z \otimes A_p)^* \ne 0} \exp \left[ 2\pi i\alpha\left(\theta_1 \oplus 0\right) \right] + \\
    &\sum_{\lambda \in (\Z/p\Z)^{\times}} \sum_{\theta\in \mathfrak{L}\otimes (\Z/p\Z)^l} \exp [2\pi i (xX_p\oplus yY_p)(\theta_1,\theta_2)]\sum_{\alpha \in (\Z/p\Z \otimes A_p)^* \ne 0} \exp \left[ 2\pi i\alpha\left(\theta_1 + \lambda \theta_2 \oplus 0\right) \right] =
    \end{aligned} \\
    & \notag \quad \begin{aligned}
        &2 \sum_{\theta\in \mathfrak{L}\otimes (\Z/p\Z)^l}  |\Z/p\Z\otimes A_p| \, \delta(\theta_1) - (p+1) \sum_{\theta\in \mathfrak{L}\otimes (\Z/p\Z)^l} \exp [2\pi i (xX_p\oplus yY_p)(\theta_1,\theta_2)] + \\ 
        & \sum_{\lambda \in (\Z/p\Z)^{\times}} \sum_{\theta\in \mathfrak{L}\otimes (\Z/p\Z)^l} |\Z/p\Z\otimes A_p| \, \delta(\theta_1 + \lambda \theta_2) \exp [2\pi i (xX_p\oplus yY_p)(\theta_1,\theta_2)] =
    \end{aligned} \\
    & \quad \begin{aligned}
    2p^l \, \sqrt{| \mathfrak{L} \otimes A_p|} - (p+1)p^l+\sqrt{| \mathfrak{L} \otimes A_p|}\sum_{\lambda \in (\Z/p\Z)^{\times}} (-1)^y \,\varepsilon_p^l\, p^{l/2} \left(\frac{(-\lambda)^l}{p} \right) ,
    \end{aligned}
\end{align}
where in the last line we have evaluated the Gauss sum:
\begin{equation}
    \sum_{\theta\in  (\Z/p\Z)^l} \exp [- 2\pi i (xX_p\oplus yY_p)( \lambda \theta,\theta)] =  (-1)^y \,\varepsilon_p^l\, p^{l/2} \left(\frac{(-\lambda)^l}{p} \right) .
\end{equation}
Putting all together, we obtain:
\begin{multline}
    \sum_{(A_p,\ell_p)}|A_p|^{g-s}  \frac{|\mathfrak{L}|^{g-1}\,}{|\mathfrak{L}\otimes {A_p}|}\, |G \otimes A_p | \biggl( (1-p^{-2g}) p^l + \\
     (p^{-g}-p^{-2g})\left(2p^l \sqrt{|\mathfrak{L}\otimes A_p|}- (p+1)p^l+(p-1)\, \delta(A_p=0)\right ) + p^{-2g}|\mathfrak{L} \otimes A_p | \biggr),
\end{multline}
after having noticed that the factor $(-1)^y$ localizes the sum over all pairs $(A_p,\ell_p)$ on the trivial group $A_p = 0$. Indeed, recall that on each $(\Z/{p^m}\Z)^d$ there are only two quadratic forms with $y=0$ and $y=1$, respectively.
Finally, introducing the quantities $b$ and $l$ as in the decomposition:
\begin{equation}
    A_p = \bigoplus_{i=1}^b \Z / p^{m_i} \Z = (\Z/p\Z)^l \oplus \tilde{A}_p,
\end{equation}
we can rewrite the above as:
\begin{align}
    \notag
    &p^{-2} \sum_{(A_p,\ell_p)}|A_p|^{g-s}  \left((p^g-1)(p^g -p) + 2(p^g-1)p^b+(p-1)(p^g-1) \,\delta(A_p = 0) + p^{2b-l}\right) = \\
    & \quad  \begin{aligned}
    p^{-2} \biggl((p^g-1)(p^g -p) & \,\Xi_p(s-g)  + 2(p^g-1)\prod_{m\ge1}\frac{1+p^{1-m(s-g)}}{1-p^{1-m(s-g)}} \\
    &+(p-1)(p^g-1) + \frac{1+p^{1-(s-g)}}{1-p^{1-(s-g)}} \prod_{m\ge 2}\frac{1+p^{2-m(s-g)}}{1-p^{2-m(s-g)}}  \biggr).
    \end{aligned}
\end{align}
In the last equality, we have used a refined version of~\eqref{counting-pnot3-groups}:
\begin{equation}
    \sum_{\substack{(A_p,\ell_p):\\\text{finite abelian}\\p\text{-groups}\\\text{with pairings}}}|A_p|^{-s} \, x^b\, y^l= \frac{1+xy\,p^{-s}}{1-xy\,p^{-s}} \prod_{m\ge 2}\frac{1+xp^{-ms}}{1-xp^{-ms}}.
\end{equation}

Combining the factors for different primes, we arrive at the final result for the sum over connected topologies:
\begin{multline}
     \mathcal{Z}_g(D,q;\mathfrak{L};\mu)= \\
    \frac{\Xi(s-g)}{1-\lambda|D|^{\frac{1}{2}}}\cdot p^{-2}
    \biggl(
  (p^g-1)(p^g-p)+2(p^g-1)\prod_{m\geq 1}\frac{1-p^{-m(s-g)}}{1+p^{-m(s-g)}}\frac{1+p^{1-m(s-g)}}{1-p^{1-m(s-g)}} + \\ 
  (p-1)(p^g-1)\prod_{m\geq 1}\frac{1-p^{-m(s-g)}}{1+p^{-m(s-g)}} + \frac{1+p^{1-(s-g)}}{1-p^{1-(s-g)}} \frac{1-p^{-(s-g)}}{1+p^{-(s-g)}} \prod_{m\geq 2}\frac{1-p^{-m(s-g)}}{1+p^{-m(s-g)}} \frac{1+p^{2-m(s-g)}}{1-p^{2-m(s-g)}}
    \biggr).
\end{multline}

For completeness, we include the results corresponding to Neumann and Dirichlet boundary conditions. We avoid diving into the details of the computation since it is a slight modification of the above derivation. We find:
\begin{multline}
     \mathcal{Z}_g(D,q;\mathfrak{L}_{\text{Neu}};\mu)= \\
    \frac{\Xi(s-g)}{1-\lambda|D|^{\frac{1}{2}}}\cdot p^{-2}
    \biggl(
  (p^{2g}-p^g)+(p^g-1)\prod_{m\geq 1}\frac{1-p^{-m(s-g)}}{1+p^{-m(s-g)}}\frac{1+p^{1-m(s-g)}}{1-p^{1-m(s-g)}} + \\ 
  \frac{1+p^{1-(s-g)}}{1-p^{1-(s-g)}} \frac{1-p^{-(s-g)}}{1+p^{-(s-g)}} \prod_{m\geq 2}\frac{1-p^{-m(s-g)}}{1+p^{-m(s-g)}} \frac{1+p^{2-m(s-g)}}{1-p^{2-m(s-g)}}
    \biggr).
\end{multline}
Notice that imposing Dirichlet boundary conditions would yield the same connected gravitational partition function since $\mathfrak{L}_{\text{Dir}} =(\mathfrak{L}_{\text{Neu}})^*$ are dual to each other.

Thus, we have obtained all the ingredients needed to recover the full state via the expressions (\ref{same-lagrangian-condition}) and (\ref{cosmological-state-sum-over-lagrangians}).

\subsection{A simple twisted Dijkgraaf-Witten theory}

One should be interested also in twisted Dijkgraaf-Witten theories because, together with the untwisted theories, they represent all the possible abelian 3d TQFTs admitting topological boundary conditions. Indeed, it was proven in~\cite{Kaidi:2021gbs} that a theory admitting topological boundary conditions is necessarily a Dijkgraaf-Witten theory.
The problem, however, is that the expression for partition function of such theories is more involved, so that the explicit evaluation of the corresponding gravitational partition function becomes more complicated. Thus we only include a simple example of a twisted DW to show that in principle it is still possible to explicitly compute the gravitational partition function for a simple choice of measure.

Let $p$ be an odd prime. Consider an abelian Chern-Simons theory with level:
\begin{equation}
    K = 
    \begin{pmatrix}
    0 &p \\
    p &2
    \end{pmatrix}.
\end{equation}
This theory realizes a twisted Dijkgraaf-Witten with gauge group $\Z /p\Z$. See e.g.~\cite{Guo:2018vij} for a review of the correspondence between these two description of the 3d theory. For our purposes, we just remark that the discriminant group is given by $D = \Z/{p^2}\Z$, while the quadratic form is:
\begin{equation}
    q(x) = \left\{
    \begin{array}{@{}cl}
         x^2/p^2  & \text{if } p= 1 \mod 4,  \\
         -x^2/p^2  &\text{if } p= 3 \mod 4.
    \end{array}
    \right.
\end{equation}

Being a DW TQFT, this theory admits topological boundary conditions. In particular, we fix the Lagrangian subgroup $\mathfrak{L} = \Z /p \Z \subseteq D$, given by the inclusion $1 \mapsto p$.

As in the previous subsection, it is sufficient to consider the partition function on 3-manifolds with the first homology group being $p$-torision, for the chosen prime $p$. Let $c$ be an element of $\mathfrak{L}^g$. Then, for $ \Z /{p^a}\Z$ with $a \ge 2$, we obtain:
\begin{equation}
    Z(\Z /{p^a}\Z, X_{p^a}, f; D, q; c) = {Z}(\Z /{p^a}\Z, Y_{p^a}, f; D, q; c) = \frac{1}{|\mathfrak{L}|} \sqrt{|\Z /{p^a}\Z \otimes D|} \, \delta( \theta),
\end{equation}
where we recall $\theta = \sum_{i=1}^gc_i \otimes f_i \in \mathfrak{L}\otimes \Z /{p^a}\Z$. Finally, we are left with the case $a=1$:
\begin{align}
    &{Z}(\Z /{p}\Z, X_{p}, f; \Z_{p^2}, q; c) = \frac{1}{|\mathfrak{L}|} \sqrt{|\Z /{p}\Z \otimes D|} \left\{
    \begin{array}{@{}cl}
         \exp \left[ 2\pi i (b\otimes  X_p) (\theta,\theta) \right]  & \text{if } p=1 \mod 4, \\
         i\exp \left[ -2\pi i(b\otimes  X_p)(\theta,\theta) \right]  & \text{if } p=3 \mod 4,
    \end{array}
    \right. \\
    &{Z}(\Z /{p}\Z, Y_{p}, f; \Z_{p^2}, q; c) =-\frac{1}{|\mathfrak{L}|} \sqrt{|\Z /{p}\Z \otimes D|} \left\{
    \begin{array}{@{}cl}
        \exp \left[ 2\pi i (b \otimes  Y_p)(\theta,\theta) \right]  & \text{if } p=1 \mod 4, \\
        i \exp \left[ -2\pi i  (b \otimes  Y_p)(\theta,\theta) \right] & \text{if } p=3 \mod 4,
    \end{array}
    \right.
\end{align}
where we have introduced the function $b: \mathfrak{L}\times\mathfrak{L} \to \Z$, defined by $b(x,y) = \hat{x} \hat{y}$, given a lift $\hat{(\cdot)}$ from $\mathfrak{L}$ to $\Z$.

For later convenience, as in the previous subsection, we recombine these results in the partition function of the generic pair $(A_p,\ell_p) = ((\Z/p\Z)^l\oplus \tilde{A}_p, xX_p\oplus yY_p \oplus \tilde{\ell}_p)$. Accordingly decomposing $\theta \oplus \tilde{\theta} \in \mathfrak{L}\otimes ((\Z/p\Z)^l\oplus \tilde{A}_p) $, we have:
\begin{equation}
    \label{eqn:twisted-dw-partition-function}
    Z(A_p,\ell_p;D,q;\theta \oplus \tilde{\theta}) = \frac{\sqrt{|A_p \otimes D|}}{|\mathfrak{L}|} (-1)^y\varepsilon_p^l  \, \delta(\tilde{\theta}) \exp\left[2\pi i \varepsilon_p^2 (b \otimes(xX_p\oplus yY_p))(\theta,\theta) \right],
\end{equation}
where $\varepsilon_p$ is either $1$ or $i$, depending whether $p =1 \mod{4}$ or $p=3 \mod{4}$.

We now assume the measure~\eqref{eqn:explicitMeasure} and compute the gravitational partition functions for this specific twisted DW TQFT. It is enough to consider the factor corresponding to the particular prime $p$ in the expression~\eqref{sum-over-topologies-specific-measure}. Plugging~\eqref{eqn:twisted-dw-partition-function} into  the partition function, yields:
\begin{multline}
    \sum_{\substack{(A_p,\ell_p):\\\text{finite abelian}\\p\text{-groups}\\\text{with pairings}}}|A_p|^{-s} (-1)^y  \varepsilon_p^l \,\frac{|\mathfrak{L}|^{g-1}\,}{|\mathfrak{L}\otimes {A_p}|}\, \sqrt{|D \otimes A_p |} \times \\
     \sum_{\theta\in \mathfrak{L}\otimes (\Z/p\Z)^l}\sum_{\substack{\alpha\in (\mathfrak{L}\otimes {A}_p)^*\\\cong \hom({A}_p,\mathfrak{L}^*)}}|\ker{\alpha}|^{g} \exp \left[ 2\pi i\alpha\left(\theta \oplus 0\right)+ 2\pi i \varepsilon_p^2(b \otimes(xX_p\oplus yY_p))(\theta,\theta) \right].
\end{multline}

Now, the trick is again to notice that all maps $\alpha_p:A_p\rightarrow \Z/p\Z$ have kernel $|\ker{\alpha}| = |A_p| p^{-1}$, except the zero map. Then, we reorganize the sum as:
\begin{align}
    & \notag
    \begin{aligned}
    (1-p^{-g})
    \sum_{(A_p,\ell_p)} &|A_p|^{g-s} (-1)^y \varepsilon_p^l \,
     \frac{|\mathfrak{L}|^{g-1}\,}{|\mathfrak{L}\otimes {A_p}|}\, \sqrt{|D \otimes A_p|} \times \\
     &\sum_{\theta\in \mathfrak{L}\otimes (\Z/p\Z)^l} \exp \left[  2\pi i \varepsilon_p^2(b \otimes(xX_p\oplus yY_p))(\theta,\theta) \right]
     +
     \end{aligned} \\
     & \begin{aligned}
     p^{-g}
     \sum_{(A_p,\ell_p)}&|A_p|^{g-s} (-1)^y \varepsilon_p^l \,
      \frac{|\mathfrak{L}|^{g-1}\,}{|\mathfrak{L}\otimes {A_p}|}\, \sqrt{|D \otimes A_p |} \times \\ & \sum_{\theta\in \mathfrak{L}\otimes (\Z/p\Z)^l} \exp [2\pi i \varepsilon_p^2(b \otimes(xX_p\oplus yY_p))(\theta,\theta)] \sum_{\substack{\alpha\in (\mathfrak{L}\otimes {A}_p)^*\\\cong \hom({A}_p,\mathfrak{L}^*)}}\exp \left[ 2\pi i\alpha\left(\theta \oplus 0\right) \right].
      \end{aligned}
\end{align}

Finally, we make use of the following equality for quadratic Gauss sums:
\begin{equation}
    (-1)^y\ \varepsilon_p^l\sum_{\theta\in \mathfrak{L}\otimes (\Z/p\Z)^l} \exp \left[  2\pi i \varepsilon_p^2(b \otimes(xX_p\oplus yY_p))(\theta,\theta) \right] = p^{l/2},
\end{equation}
and obtain:
\begin{multline}
    |\mathfrak{L}|^{g-1}
    \sum_{(A_p,\ell_p)}|A_p|^{g-s} \left(  (1-p^{-g}) 
    +p^{-g}
     (-1)^y \varepsilon_p^l
     \, \sqrt{|D \otimes A_p |} \right)
    = \\ |\mathfrak{L}|^{g-1} \left((1-p^{-g})\,\Xi_p(s-g) + p^{-g} \right).
\end{multline}
In the last equality, we have used a refined version of~\eqref{counting-pnot3-groups}:
\begin{equation}
    \sum_{(A_p,\ell_p)}|A_p|^{g-s} (-1)^y \varepsilon_p^l
     \, \sqrt{|D \otimes A_p |} = 1.
\end{equation}
We again notice the factor $(-1)^y$, which kills every contribution but the one corresponding to the trivial group $A_p = 0$.

Combining the factors for different primes, we arrive at the final result for the sum over topologies:
\begin{equation}
     \mathcal{Z}_g(D,q;\mathfrak{L};\mu)=
    \frac{\Xi(s-g)}{1-\lambda|D|^{\frac{1}{2}}}\cdot 
    \left(
  p^{-1}(p^g-1)+ p^{-1}\prod_{m\geq 1}\frac{1-p^{-m(s-g)}}{1+p^{-m(s-g)}}
    \right).
\end{equation}

\section*{Acknowledgments}

P.P. would like to thank Lorenzo Di Pietro, Anton Kapustin, Zohar Komargodski, Kimyeong Lee for discussions on the topic. T.N. is grateful to Mohamed Aliouane and Leonardo Goller for valuable conversations regarding this subject. The research of T.N. is partly supported by the INFN Iniziativa Specifica GAST and Indam GNFM.

\begin{appendix}

\section{Proof of the Proposition~\ref{prop:dist-solution}}
\label{sec:proofProp2}

We give a proof of the Proposition~\ref{prop:dist-solution}, starting with the first statement. Let us introduce:
\begin{equation}
    \Gamma(p) \coloneqq c_{n- \sum_i m_{p_i}} \prod_i \delta(p_i) \, \frac{m_{p_i}!}{(m_{p_i} -|p_i|)!},
\end{equation}
so that:
\begin{equation}
    f_{N,\theta^{-1}_1, \dots, \theta^{-1}_k}(n, m_1, \dots, m_k) = e^{-x_0} \frac{(n-k)!}{n!} \sum_{p \in P(\{ 1, \dots k\}) } \Gamma(p).
\end{equation}

Given a partition $q$ of the set $\{1,\dots,k+1\}$, we can always obtain a partition $p = \{p_1,\dots, p_l\}$ of $\{1,\dots,k\}$ by removing the element $k+1$. In particular, $q$ can only have two different forms:
\begin{align}
\notag
    q^{(j)} &= \{ p_1, \dots, p_j \cup \{k+1\}, \dots, p_l \} ,\qquad j=1,\dots,l,\\
    \bar{q} &= \{ p_1, \dots, p_l, \{ k+1\}\}.
\end{align}
Thus, instead of summing over $q\in P(\{1,\dots,k+1\})$, we may take the sum over $p\in P(\{1,\dots,k\})$ and then, at a fixed $p$, sum over the $l+1$ ways to include $\{k+1\}$ in $p$.

Having to evaluate:
\begin{equation}
    \sum_{m_{k+1}} f_{N,\theta^{-1}_1, \dots, \theta^{-1}_{k+1}}(n, m_1, \dots, m_{k+1}),
\end{equation}
we start with computing:
\begin{align}
    &\sum_{m_{k+1}} \Gamma(q^{(j)}) = (m_{p_j} - |p_j|) \, \Gamma(p), \\
    &\sum_{m_{k+1}} \Gamma(\bar{q}) = (n - \sum_i m_{p_i}) \,\Gamma(p).
\end{align}

It is then easy to deduce:
\begin{align}
\notag
    \sum_{m_{k+1}} f_{N,\theta^{-1}_1, \dots, \theta^{-1}_{k+1}}(n, m_1, \dots, m_{k+1}) 
    &= e^{-x_0} \frac{(n-k-1)!}{n!} \, (n-k ) \sum_{p \in P(\{ 1, \dots k\}) } \Gamma(p)\\ 
    &= f_{N,\theta^{-1}_1, \dots, \theta^{-1}_{k}}(n, m_1, \dots, m_{k}),
\end{align}
as claimed.

Next, we show that the proposed distribution satisfies the abelian condition:
\begin{equation}
\label{eqn:partitionAbelianCondition}
\sum_{p \in P(\{ g_1, \dots g_k\}) } \prod_{q \in p} \mathbb{E} \left[ N \theta^{1- \sum_{g \in q} g } \right]
= \sum_{p \in P(\{ g_1, \dots g_k\}) }   \mathbb{E} \left[ \frac{N!}{(N-|p|)!} \prod_{q \in p} \theta_q^{\sum_{g \in q} (1-g)} \right].
\end{equation}

The first step is to consider the addend on the left-hand side.

\begin{lemma} 

Fix a partition $p = \{p_1,\dots,p_l\} \in P(\{g_1,\dots,g_k\})$, then:
\begin{equation}
\label{eqn:lefthandsideComputed}
\prod_{i=1}^l \mathbb{E} \left[ N \theta^{1- \sum_{g \in p_i} g} \right] = e^{-x_0} \sum_{n\ge l} \sum_{m_1 + \dots + m_l = l}^n c_{n - m_1 - \dots - m_l} \prod_{i=1}^l b_{m_i} \, m_i^{\sum_{g \in p_i} g} .
\end{equation}

Moreover, it is convergent as long as $x_k$ -- defined as in~\eqref{eqn:DirichletSeriesHypothesis} -- is absolutely convergent, where $k = \max_{i}( \sum_{g \in p_i} g )$.
\end{lemma}

\begin{proof}
We start by looking at the case $l=1$:
\begin{multline}
\mathbb{E} \left[ N \theta^{1- \sum_{i =1}^r g_i} \right] = 
\\
=
\sum_{n\ge1} \sum_{m=1}^n n \, m^{\sum_{i=1}^r g_i - 1} f_{N,\theta^{-1}}(n,m) 
= e^{-x_0} \sum_{n\ge1} \sum_{m=1}^n m^{\sum_{i=1}^r g_i}  \, b_m c_{n-m}.
\end{multline}
We want to prove the statement by induction on $l$. However, first notice that we can also write:
\begin{equation}
\mathbb{E} \left[ N \theta^{1-G} \right] 
= e^{-x_0} \sum_{n\ge1} \sum_{m=1}^n m^G \,b_m c_{n-m} = e^{-x_0} \sum_{m\ge1} m^G \, b_m \sum_{n \ge m} c_{n-m} = \sum_{m\ge1} m^G \, b_m = x_G,
\end{equation}
since, by normalization, we have:
\begin{equation}
e^{-x_0} \sum_{n \ge 0} c_n = \mathbb{E}[1] = 1.
\end{equation}

Suppose that the claim holds for $l$, we show that it also holds for $l+1$:
\begin{multline}
\prod_{i=1}^{l+1} \mathbb{E} \left[ N \theta^{1- \sum_{g \in p_i} g} \right]  =\mathbb{E} \left[ N \theta^{1-\sum_{g \in p_{l+1}} g} \right]   \prod_{i=1}^l \mathbb{E} \left[ N \theta^{1- \sum_{g \in p_i} g} \right] =
\\
=  e^{-x_0} \sum_{n\ge l} \sum_{m_1 + \dots + m_l = l}^n c_{n - m_1 - \dots - m_l} \prod_{i=1}^l b_{m_i} \,m_i^{\sum_{g \in p_i} g}  \sum_{m_{l+1}\ge1} m_{l+1}^{\sum_{g \in p_{l+1}} g} \, b_{m_{l+1}},
\end{multline}
by having used the inductive hypothesis and the remark about the case $l=1$. Finally, we use the Cauchy product of series to  get:
\begin{equation}
\prod_{i=1}^{l+1} \mathbb{E} \left[ N \theta^{1- \sum_{g \in p_i} g} \right] = e^{-x_0} \sum_{n \ge l+1} \sum_{m_{l+1} = 1}^{n-l} \sum_{m_1 + \dots + m_l =l}^{n-m_{l+1}} c_{n - m_1 - \dots - m_l - m_{l+1}} \prod_{i=1}^{l+1} b_{m_i} \, m_i^{\sum_{g \in p_i} g} .
\end{equation}
\end{proof}

Let us introduce the following notation:
\begin{equation}
    a(n,\{p_1, \dots, p_l\}) = \sum_{m_1 + \dots + m_l = l}^n c_{n - m_1 - \dots - m_l} \prod_{i=1}^l b_{m_i} \, m_i^{\sum_{g \in p_i} g} ,
\end{equation}
for each addend in the series in equation~\eqref{eqn:lefthandsideComputed}.

Now, we focus on the study of the addends of the right-hand side of~\eqref{eqn:partitionAbelianCondition}. We simply substitute for the expression of the probability density function. Let us do it for the following addend:
\begin{align}
b(n, \{p_1, \dots, p_l\}) &= e^{x_0} \sum_{m_1, \dots, m_l} f_{N,\theta_1^{-1}, \dots, \theta_l^{-1}}(n,m_1, \dots, m_l) \, \frac{n!}{(n-l)!} \prod_{i=1}^l m_i^{\sum_{g \in p_i} (g-1)} \\
& =  \sum_{m_1, \dots, m_l} \sum_{\rho \in P(\{1, \dots, l\})} c_{n - \sum_{q \in \rho} m_{q} } \prod_{q \in \rho} \delta(q) \, b_{m_q} \frac{m_q!}{(m_q - |q|)!} \prod_{i=1}^l m_i^{\sum_{g \in p_i} (g-1)}. \notag
\end{align}

Notice that, taking the sum over all the $m_i$'s, the effect of $\delta(\rho_\alpha)$ is to identify some of them. In particular, $\delta( \rho_\alpha = \{ \mu_1, \dots, \mu_r \})$ identifies $m_{\mu_1} = \dots = m_{\mu_r} \eqqcolon  m_\alpha$, so that the sum over the $m_i$'s is restricted to a sum over the $m_\alpha$'s, $\alpha = 1, \dots, \lambda$. This is equivalent to \textit{fusing} the partitions $p = \{p_1, \dots, p_l\}$ and $\rho$ into a new partition $F_p(\rho) \in P(\{ g_1, \dots, g_k\})$, which we call the fusion of $p$ along $\rho$.

\begin{definition}
Define the fusion function as follows:
\begin{align}
\notag
F \colon P(\{g_1, \dots, g_k\}) \times P(\{1, \dots, k\}) &\to P(\{g_1, \dots, g_k\}), \\
 (\{p_1, \dots, p_l\},\{\rho_1, \dots, \rho_\lambda\}) &\mapsto \{ \tilde{p}_1, \dots, \tilde{p}_\lambda \}
\end{align}
where $\tilde{p}_\alpha = \bigcup_{i \in \rho_\alpha} p_i$, with $\tilde{p}_\alpha$ possibly empty if $i \ge |p| = l$ for all $i \in \rho_\alpha$.
\end{definition}
In our case $\rho$ will always be an  element of $P(\{1, \dots, l\})$, which can be seen as a subset of $P(\{1, \dots, k\})$ by appending $\{ l+1, \dots, k \}$ to $\rho$.

Then, noticing that:
\begin{equation}
m_\alpha^{\sum_{g \in \tilde{p}_\alpha}(g-1)} = m_\alpha^{\sum_{i = 1}^{|\rho_\alpha|} \sum_{g \in p_i}(g-1)} = \prod_{i=1}^{|\rho_\alpha|} m_\alpha^{\sum_{g \in p_i}(g-1)},
\end{equation}
we can use the fusion to rewrite:
\begin{align}
\notag
b(n, \{p_1, \dots, p_l\}) &= \sum_{\rho \in P(\{1,\dots,l\})} \sum_{m_1, \dots, m_\lambda} c_{n - \sum_{\alpha = 1}^\lambda m_\alpha} \prod_{\alpha = 1}^\lambda b_{m_\alpha} \frac{m_\alpha!}{(m_\alpha- |\rho_\alpha|)!} \, m_\alpha^{\sum_{g \in \tilde{p}_\alpha}(g-1)} \\
& = \sum_\rho \sum_{m_1, \dots, m_{|F(p,\rho)|} } c_{n - \sum_{q \in F(p,\rho)} m_q} \prod_{q \in F(p,\rho)} b_{m_q} \frac{m_\alpha!}{(m_\alpha - |\rho_\alpha|)!} \,  m_q^{\sum_{g \in q} (g-1) }.
\end{align}



In other words, $b$ does not depend directly on $p$ but on $F(p,\rho)$. Hence, we would like to perform a smart change of summation indices. We have the following:

\begin{lemma}
The function
\begin{align}
\notag
P(\{g_1,\dots,g_k\}) \times P(\{1, \dots, k\}) &\to P(\{g_1,\dots,g_k\}) \times P(\{g_1,\dots,g_k\}),\\
(p,\rho) & \mapsto (F(p,\rho), p)
\end{align}
establishes a one-to-one correspondence between pairs $(p,\rho)$ with $\rho \in P(\{1,\dots,|p|\})\subseteq P(\{1, \dots, k\})$ and pairs $(p',q)$ with $p' \in P(\{g_1,\dots,g_k\})$ and $q \in \mathcal{F}(p') \subseteq P(\{g_1,\dots,g_k\})$. We have introduced $\mathcal{F}(p')$ as the set of possible partitions that can yield $p$ via fusion, i.e.:
\begin{equation}
    \mathcal{F}(p') = \left\{ q \in P(\{g_1,\dots,g_k\}) \ | \ \exists \rho \in P(\{1, \dots, |p'|\}) \colon F(q,\rho) =p' \right\}.
\end{equation}
\end{lemma}

\begin{proof}
    Given the pair $(p,\rho)$, we can construct the pair $(p',q)$ as $(F(p,\rho),p)$. Clearly, $p \in \mathcal{F}(F(p,\rho))$. Moreover, this pair is unique since $F(p,\rho) = F(p, \rho')$ implies $\rho = \rho' \in P(\{1,\dots,|p|\})$. Notice that this last statement would not hold in $P(\{1,\dots,k\})$.

    We easily invert this construction. Fix a pair $(p', q)$ with $q \in \mathcal{F}(p')$. Then, there exists some $\rho$ such that the pair $(q, \rho)$ yields $(p' = F(q,\rho), q)$.
\end{proof}

Therefore, instead of summing over $(p,\rho)$, we can sum over pairs $(p',q)$ with $p' \in P(\{g_1,\dots,g_k\})$ and $\rho \in \mathcal{F}(p')$. In particular, notice that $q \in \mathcal{F}(p')$ means that we can write $q = \{ \bar{q}_1, \dots, \bar{q}_{|p'|} \}$, where each $\bar{q}_\alpha$ is itself a partition of the set $p'_\alpha$, with $p' = \{ p'_1, \dots, p'_{|p'|} \}$. Performing this change of variables yields:
\begin{equation}
\sum_{p \in P(\{g\})} b(n,p) =\sum_{p \in P(\{g\})} \sum_{m_1, \dots, m_{|p|} } c_{n-\sum_{\alpha \in p} m_\alpha} \prod_{\alpha \in p} b_{m_\alpha} \,m_\alpha^{\sum_{g \in \alpha} (g-1)}  \sum_{q \in \mathcal{F}(p)}  \prod_{\alpha =1}^{|p|} \frac{m_\alpha!}{(m_\alpha - |q_\alpha|)!} .
\end{equation}

To prove proposition~\ref{prop:dist-solution}, we will show the equality between the coefficients of the two series:
\begin{equation}
    \sum_{p \in P(\{g\})} a(n,p) = \sum_{p \in P(\{g\})} b(n,p). 
\end{equation}
It is enough to manipulate the expression for $b$. Indeed, notice that for a fixed partition $p$ we have:
\begin{equation}
 \sum_{q \in \mathcal{F}(p)}  \prod_{\alpha =1}^{|p|} \frac{m_\alpha!}{(m_\alpha - |q_\alpha|)!} = \prod_{\alpha = 1}^{|p|} \sum_{q \in P(p_\alpha)} \frac{m_\alpha!}{(m_\alpha - |q|)!}.
\end{equation}
Another useful result is the following:
\begin{lemma}
The following identity holds true:
\begin{equation}
\sum_{q \in P(p)} \frac{m!}{(m-|q|)!} = m^{|p|}.
\end{equation}
\end{lemma}

\begin{proof}
First of all, decompose $P(p)$ into sets of partitions of fixed length $k$, denoted by $P_k(p)$. Then:
\begin{equation}
\sum_{q \in P(p)} \frac{m!}{(m-|q|)!} = \sum_{k=0}^{|p|} \frac{m!}{(m-k)!}  \sum_{q \in P_k(p)}  1.
\end{equation}
Notice that the number of partitions of length $k$ of a set of order $|p|$ is equal to the Stirling numbers of second kind $S(|p|,k)$. Thus, the result follows by a fundamental property of the Stirling numbers:
\begin{equation}
\sum_{q \in P(p)} \frac{m!}{(m-|q|)!} = \sum_{k=0}^{|p|} \frac{m!}{(m-k)!} \, S(|p|,k) = m^{|p|}.
\end{equation}

\end{proof}

Using the above remarks, we get:
\begin{align}
\notag
\sum_{p \in P(\{g\})} b(n,p) &= \sum_{p \in P(\{g\})} \sum_{m_1, \dots, m_{|p|} } c_{n-\sum_{\alpha \in p} m_\alpha} \prod_{\alpha \in p} b_{m_\alpha} \,m_\alpha^{\sum_{g \in \alpha} g} \prod_{\alpha =1}^{|p|} m_\alpha^{-|p_\alpha|} \, m_\alpha^{|p_\alpha|} \\
&=  \sum_{p \in P(\{g\})} \sum_{m_1, \dots, m_{|p|} } c_{n-\sum_{\alpha \in p} m_\alpha} \prod_{\alpha \in p} b_{m_\alpha} \, m_\alpha^{\sum_{g \in \alpha} g}.
\end{align}
This concludes the proof of the Proposition \ref{prop:dist-solution}. Notice that this guarantees in particular the convergence of the right-hand side of~\eqref{eqn:partitionAbelianCondition} as long as the left-hand side converges.

\end{appendix}

\bibliography{top}

\providecommand{\href}[2]{#2}\begingroup\raggedright\begin{thebibliography}{10}

\bibitem{Coleman:1988cy}
S.R.~Coleman, \emph{{Black holes as red herrings: Topological fluctuations and the loss of quantum coherence}}, \href{https://doi.org/10.1016/0550-3213(88)90110-1}{\emph{Nucl. Phys. B} {\bfseries 307} (1988) 867}.

\bibitem{Banks:1988je}
T.~Banks, \emph{{Prolegomena to a Theory of Bifurcating Universes: A Nonlocal Solution to the Cosmological Constant Problem Or Little Lambda Goes Back to the Future}}, \href{https://doi.org/10.1016/0550-3213(88)90455-5}{\emph{Nucl. Phys. B} {\bfseries 309} (1988) 493}.

\bibitem{Giddings:1988wv}
S.B.~Giddings and A.~Strominger, \emph{{Baby Universes, Third Quantization and the Cosmological Constant}}, \href{https://doi.org/10.1016/0550-3213(89)90353-2}{\emph{Nucl. Phys. B} {\bfseries 321} (1989) 481}.

\bibitem{Fischler:1989ka}
W.~Fischler, I.R.~Klebanov, J.~Polchinski and L.~Susskind, \emph{{Quantum Mechanics of the Googolplexus}}, \href{https://doi.org/10.1016/0550-3213(89)90290-3}{\emph{Nucl. Phys. B} {\bfseries 327} (1989) 157}.

\bibitem{Marolf:2020xie}
D.~Marolf and H.~Maxfield, \emph{{Transcending the ensemble: baby universes, spacetime wormholes, and the order and disorder of black hole information}}, \href{https://doi.org/10.1007/JHEP08(2020)044}{\emph{JHEP} {\bfseries 08} (2020) 044} [\href{https://arxiv.org/abs/2002.08950}{{\ttfamily 2002.08950}}].

\bibitem{Balasubramanian:2020jhl}
V.~Balasubramanian, A.~Kar, S.F.~Ross and T.~Ugajin, \emph{{Spin structures and baby universes}}, \href{https://doi.org/10.1007/JHEP09(2020)192}{\emph{JHEP} {\bfseries 09} (2020) 192} [\href{https://arxiv.org/abs/2007.04333}{{\ttfamily 2007.04333}}].

\bibitem{Gardiner:2020vjp}
J.G.~Gardiner and S.~Megas, \emph{{2d TQFTs and baby universes}}, \href{https://doi.org/10.1007/JHEP10(2021)052}{\emph{JHEP} {\bfseries 10} (2021) 052} [\href{https://arxiv.org/abs/2011.06137}{{\ttfamily 2011.06137}}].

\bibitem{deMelloKoch:2021lqp}
R.~de~Mello~Koch, Y.-H.~He, G.~Kemp and S.~Ramgoolam, \emph{{Integrality, duality and finiteness in combinatoric topological strings}}, \href{https://doi.org/10.1007/JHEP01(2022)071}{\emph{JHEP} {\bfseries 01} (2022) 071} [\href{https://arxiv.org/abs/2106.05598}{{\ttfamily 2106.05598}}].

\bibitem{Banerjee:2022pmw}
A.~Banerjee and G.W.~Moore, \emph{{Comments on summing over bordisms in TQFT}}, \href{https://doi.org/10.1007/JHEP09(2022)171}{\emph{JHEP} {\bfseries 09} (2022) 171} [\href{https://arxiv.org/abs/2201.00903}{{\ttfamily 2201.00903}}].

\bibitem{Maloney:2020nni}
A.~Maloney and E.~Witten, \emph{{Averaging over Narain moduli space}}, \href{https://doi.org/10.1007/JHEP10(2020)187}{\emph{JHEP} {\bfseries 10} (2020) 187} [\href{https://arxiv.org/abs/2006.04855}{{\ttfamily 2006.04855}}].

\bibitem{Cotler:2020ugk}
J.~Cotler and K.~Jensen, \emph{{AdS$_{3}$ gravity and random CFT}}, \href{https://doi.org/10.1007/JHEP04(2021)033}{\emph{JHEP} {\bfseries 04} (2021) 033} [\href{https://arxiv.org/abs/2006.08648}{{\ttfamily 2006.08648}}].

\bibitem{Maxfield:2020ale}
H.~Maxfield and G.J.~Turiaci, \emph{{The path integral of 3D gravity near extremality; or, JT gravity with defects as a matrix integral}}, \href{https://doi.org/10.1007/JHEP01(2021)118}{\emph{JHEP} {\bfseries 01} (2021) 118} [\href{https://arxiv.org/abs/2006.11317}{{\ttfamily 2006.11317}}].

\bibitem{Barbar:2023ncl}
A.~Barbar, A.~Dymarsky and A.D.~Shapere, \emph{{Global Symmetries, Code Ensembles, and Sums over Geometries}}, \href{https://doi.org/10.1103/PhysRevLett.134.151603}{\emph{Phys. Rev. Lett.} {\bfseries 134} (2025) 151603} [\href{https://arxiv.org/abs/2310.13044}{{\ttfamily 2310.13044}}].

\bibitem{Aharony:2023zit}
O.~Aharony, A.~Dymarsky and A.D.~Shapere, \emph{{Holographic description of Narain CFTs and their code-based ensembles}}, \href{https://doi.org/10.1007/JHEP05(2024)343}{\emph{JHEP} {\bfseries 05} (2024) 343} [\href{https://arxiv.org/abs/2310.06012}{{\ttfamily 2310.06012}}].

\bibitem{Dymarsky:2024frx}
A.~Dymarsky and A.~Shapere, \emph{{TQFT gravity and ensemble holography}}, \href{https://doi.org/10.1007/JHEP02(2025)091}{\emph{JHEP} {\bfseries 02} (2025) 091} [\href{https://arxiv.org/abs/2405.20366}{{\ttfamily 2405.20366}}].

\bibitem{Jafferis:2024jkb}
D.L.~Jafferis, L.~Rozenberg and G.~Wong, \emph{{3d Gravity as a random ensemble}},  \href{https://arxiv.org/abs/2407.02649}{{\ttfamily 2407.02649}}.

\bibitem{Dymarsky:2025agh}
A.~Dymarsky, J.~Henriksson and B.~McPeak, \emph{{Holographic duality from Howe duality: Chern-Simons gravity as an ensemble of code CFTs}},  \href{https://arxiv.org/abs/2504.08724}{{\ttfamily 2504.08724}}.

\bibitem{contreras2015borromean}
E.~Contreras and K.~Habiro, \emph{Borromean surgery equivalence of spin 3-manifolds with boundary}, {\emph{ANNALI SCUOLA NORMALE SUPERIORE-CLASSE DI SCIENZE} (2015) 1271}.

\bibitem{Belov:2005ze}
D.~Belov and G.W.~Moore, \emph{{Classification of Abelian spin Chern-Simons theories}},  \href{https://arxiv.org/abs/hep-th/0505235}{{\ttfamily hep-th/0505235}}.

\bibitem{Kapustin:2010hk}
A.~Kapustin and N.~Saulina, \emph{{Topological boundary conditions in abelian Chern-Simons theory}}, \href{https://doi.org/10.1016/j.nuclphysb.2010.12.017}{\emph{Nucl. Phys. B} {\bfseries 845} (2011) 393} [\href{https://arxiv.org/abs/1008.0654}{{\ttfamily 1008.0654}}].

\bibitem{Kaidi:2021gbs}
J.~Kaidi, Z.~Komargodski, K.~Ohmori, S.~Seifnashri and S.-H.~Shao, \emph{{Higher central charges and topological boundaries in 2+1-dimensional TQFTs}}, \href{https://doi.org/10.21468/SciPostPhys.13.3.067}{\emph{SciPost Phys.} {\bfseries 13} (2022) 067} [\href{https://arxiv.org/abs/2107.13091}{{\ttfamily 2107.13091}}].

\bibitem{deloup1999linking}
F.~Deloup, \emph{Linking forms, reciprocity for gauss sums and invariants of 3-manifolds}, {\emph{Transactions of the American Mathematical Society} {\bfseries 351} (1999) 1895}.

\bibitem{deloup2001abelian}
F.~Deloup, \emph{On abelian quantum invariants of links in 3-manifolds}, {\emph{Mathematische Annalen} {\bfseries 319} (2001) 759}.

\bibitem{Witten:1988hf}
E.~Witten, \emph{{Quantum Field Theory and the Jones Polynomial}}, \href{https://doi.org/10.1007/BF01217730}{\emph{Commun. Math. Phys.} {\bfseries 121} (1989) 351}.

\bibitem{atiyah1990framings}
M.~Atiyah, \emph{On framings of 3-manifolds}, {\emph{Topology} {\bfseries 29} (1990) 1}.

\bibitem{Geiko:2022qjy}
R.~Geiko and G.W.~Moore, \emph{{When Does a Three-Dimensional Chern\textendash{}Simons\textendash{}Witten Theory Have a Time Reversal Symmetry?}}, \href{https://doi.org/10.1007/s00023-023-01303-3}{\emph{Annales Henri Poincare} {\bfseries 25} (2024) 673} [\href{https://arxiv.org/abs/2209.04519}{{\ttfamily 2209.04519}}].

\bibitem{wall1963quadratic}
C.T.C.~Wall, \emph{Quadratic forms on finite groups, and related topics}, {\emph{Topology} {\bfseries 2} (1963) 281}.

\bibitem{kawauchi1980algebraic}
A.~Kawauchi and S.~Kojima, \emph{Algebraic classification of linking pairings on 3-manifolds}, {\emph{Mathematische Annalen} {\bfseries 253} (1980) 29}.

\bibitem{ivic2012riemann}
A.~Ivic, \emph{The Riemann zeta-function: theory and applications}, Courier Corporation (2012).

\bibitem{miranda1984nondegenerate}
R.~Miranda, \emph{Nondegenerate symmetric bilinear forms on finite abelian 2-groups}, {\emph{Transactions of the American Mathematical Society} {\bfseries 284} (1984) 535}.

\bibitem{deloup2007reciprocity}
F.~Deloup and V.~Turaev, \emph{On reciprocity}, {\emph{Journal of Pure and Applied Algebra} {\bfseries 208} (2007) 153}.

\bibitem{Guo:2018vij}
M.~Guo, K.~Ohmori, P.~Putrov, Z.~Wan and J.~Wang, \emph{{Fermionic Finite-Group Gauge Theories and Interacting Symmetric/Crystalline Orders via Cobordisms}}, \href{https://doi.org/10.1007/s00220-019-03671-6}{\emph{Commun. Math. Phys.} {\bfseries 376} (2020) 1073} [\href{https://arxiv.org/abs/1812.11959}{{\ttfamily 1812.11959}}].

\bibitem{Dijkgraaf:1989pz}
R.~Dijkgraaf and E.~Witten, \emph{{Topological Gauge Theories and Group Cohomology}}, \href{https://doi.org/10.1007/BF02096988}{\emph{Commun. Math. Phys.} {\bfseries 129} (1990) 393}.

\bibitem{Freed:1991bn}
D.S.~Freed and F.~Quinn, \emph{{Chern-Simons theory with finite gauge group}}, \href{https://doi.org/10.1007/BF02096860}{\emph{Commun. Math. Phys.} {\bfseries 156} (1993) 435} [\href{https://arxiv.org/abs/hep-th/9111004}{{\ttfamily hep-th/9111004}}].

\bibitem{wakui1992dijkgraaf}
M.~Wakui, \emph{On dijkgraaf-witten invariant for 3-manifolds}, {\emph{Osaka Journal of Mathematics} {\bfseries 29} (1992) 675}.

\bibitem{Furlan:2024pvy}
M.~Furlan and P.~Putrov, \emph{{On finite group global and gauged $q$-form symmetries in TQFT}},  \href{https://arxiv.org/abs/2403.04677}{{\ttfamily 2403.04677}}.

\bibitem{Benini:2022hzx}
F.~Benini, C.~Copetti and L.~Di~Pietro, \emph{{Factorization and global symmetries in holography}}, \href{https://doi.org/10.21468/SciPostPhys.14.2.019}{\emph{SciPost Phys.} {\bfseries 14} (2023) 019} [\href{https://arxiv.org/abs/2203.09537}{{\ttfamily 2203.09537}}].

\end{thebibliography}\endgroup
\bibliographystyle{JHEP}

\end{document}